\pdfoutput=1

\documentclass[11pt,twoside,a4paper,notitlepage]{article}
\usepackage[a4paper,margin=2.5cm]{geometry}

\makeatletter{}%

\usepackage{hyperref}
\usepackage[utf8]{inputenc}
\usepackage{authblk}

\usepackage{amsmath,amssymb,stmaryrd,url,nicefrac}
\usepackage{xparse}
\usepackage{todonotes}
\usepackage{enumerate}
\usepackage{mathdots}
\usepackage{graphicx}
\usepackage{mathtools}
\usepackage{pdfsync}
\usepackage{theorem}

\newtheorem{theorem}{Theorem}[section]
\newtheorem{lemma}[theorem]{Lemma}
\newtheorem{claim}[theorem]{Claim}
\newtheorem{corollary}[theorem]{Corollary}
\newtheorem{proposition}[theorem]{Proposition}

\newtheorem{remark}[theorem]{Remark}
\theorembodyfont{\upshape}
\newtheorem{example}[theorem]{Example}
\newtheorem{definition}[theorem]{Definition}

\newcommand{\qed}{\null\hfill$\square$}
\newenvironment{proof}{\textit{Proof:}\ }{}

\makeatletter
\newcommand*{\transpose}{%
  {^{\mathpalette\@transpose{}}}%
}
\newcommand*{\@transpose}[2]{%
  \raisebox{\depth}{$\m@th#1\intercal$}%
}
\makeatother

 \DeclareMathOperator{\lr}{{hlr}}
\newcommand{\marked}{\mu}
\newcommand{\unmarked}{\bar\mu}
 \DeclareMathOperator{\real}{real}
\newcommand{\bN}{{{\mathbb N}}}
\newcommand{\bZ}{{{\mathbb Z}}}
\newcommand{\cA}{{\ensuremath{\mathcal{A}}}}
\newcommand{\cF}{{\ensuremath{\mathcal{F}}}}
\newcommand{\cT}{{\ensuremath{\mathcal{T}}}}
\newcommand{\cB}{{\ensuremath{\mathcal{B}}}}
 \DeclareMathOperator{\dist}{dist}
\newcommand{\sem}[1]{\llbracket #1 \rrbracket}

 \newcommand{\quants}[1]{\ensuremath{\mathbf{#1}}}
 \newcommand{\Ps}{\quants{P}}
 \newcommand{\Es}{\quants{E}}
 \newcommand{\Us}{\quants{U}}
 
 \DeclareMathOperator{\sph}{sph}
 \newcommand{\ov}[1]{\ensuremath{\overline{#1}}}

 \newcommand*{\size}[1]{\ensuremath{|\!|#1|\!|}}
 \newcommand{\bigO}{\mathcal{O}}

 \newcommand{\True}{\ensuremath{\textit{true}}}
 \newcommand{\False}{\ensuremath{\textit{false}}}

 \DeclareMathOperator{\ar}{ar}

 \DeclareMathOperator{\br}{br} 
 \DeclareMathOperator{\bw}{bw} 
  \DeclareMathOperator{\nqr}{nqr} 
 \DeclareMathOperator{\free}{free}

 \DeclareMathOperator{\poly}{poly}

\newcommand{\nc}[1]{\newcommand{#1}}
\newcommand{\rnc}[1]{\renewcommand{#1}}

\nc{\true}{\True}
\nc{\false}{\False}

\nc{\deff}{:=}

\rnc{\leq}{\leqslant}
\rnc{\geq}{\geqslant}

\rnc{\le}{\leqslant}
\rnc{\ge}{\geqslant}

\rnc{\phi}{\varphi}

\nc{\NN}{\bN}
\nc{\NNpos}{\ensuremath{\NN_{\geq 1}}}
\nc{\ZZ}{\bZ}

\nc{\Structure}[1]{\ensuremath{\mathcal{#1}}}
\nc{\A}{\cA}
\nc{\B}{\cB}
\nc{\F}{\cF}
\nc{\T}{\cT}
\nc{\I}{\Structure{I}}

\nc{\isom}{\ensuremath{\cong}}
\nc{\isomorph}{\isom}
\nc{\equivd}{\ensuremath{\equiv_d}}

\nc{\bigoh}{\mathcal{O}}
\nc{\bigOh}{\bigoh}
\nc{\littleoh}{o}
\nc{\littleOh}{\littleoh}

\nc{\set}[1]{\ensuremath{\{ #1 \}}}
\nc{\setc}[2]{\set{#1 : #2}}
\nc{\bigset}[1]{\ensuremath{\big\{ #1 \big\}}}
\nc{\bigsetc}[2]{\bigset{#1 : #2}}
\nc{\setsize}[1]{\ensuremath{|#1|}}
\nc{\bigsetsize}[1]{\ensuremath{\big|#1\big|}}
\nc{\Setsize}[1]{\bigsetsize{#1}}

\nc{\sphere}[2]{\ensuremath{\sph_{#1}(#2)}}

\nc{\und}{\ensuremath{\wedge}}
\nc{\Und}{\ensuremath{\bigwedge}}
\nc{\oder}{\ensuremath{\vee}}
\nc{\Oder}{\ensuremath{\bigvee}}
\nc{\nicht}{\ensuremath{\neg}}
\nc{\impl}{\ensuremath{\to}}
\nc{\gdw}{\ensuremath{\leftrightarrow}}

\nc{\FO}{\ensuremath{\textup{FO}}}
\nc{\FOC}{\ensuremath{\textup{FOC}}}
\nc{\FOCN}{\ensuremath{\textup{FOCN}}}

\nc{\VARS}{\ensuremath{\textsf{\upshape vars}}}
\nc{\NVARS}{\ensuremath{\textsf{\upshape nvars}}}

\nc{\Count}[2]{\ensuremath{\# {#1}.{#2}}}
\nc{\quant}[1]{\ensuremath{\textsf{\upshape #1}}}
\rnc{\P}{\quant{P}}

\rnc{\S}{\ensuremath{\mathcal{S}}} 

\nc{\neighb}[3]{\ensuremath{N_{#1}^{#2}(#3)}} 
\nc{\nbset}[2]{\ensuremath{N_{#1}^{#2}}}
\nc{\nrA}[1]{\ensuremath{\neighb{r}{\A}{#1}}} 
\nc{\nrAStrich}[1]{\ensuremath{\neighb{r}{\B}{#1}}} 
\nc{\nrT}[1]{\ensuremath{\neighb{r}{\T}{#1}}} 
\nc{\nRT}[1]{\ensuremath{\neighb{R}{\T}{#1}}} 

\nc{\Neighb}[3]{\ensuremath{\mathcal{N}_{#1}^{#2}(#3)}} 
\nc{\NrA}[1]{\ensuremath{\Neighb{r}{\A}{#1}}} 
\nc{\NRA}[1]{\ensuremath{\Neighb{R}{\A}{#1}}} 
\nc{\NrB}[1]{\ensuremath{\Neighb{r}{\B}{#1}}} 
\nc{\NRB}[1]{\ensuremath{\Neighb{R}{\B}{#1}}} 
\nc{\NrDStrich}[1]{\ensuremath{\Neighb{r}{\B}{#1}}} 
\nc{\NrT}[1]{\ensuremath{\Neighb{r}{\T}{#1}}} 

\nc{\Types}[3]{\ensuremath{\mathcal{T}_{#1}^{\sigma,#2}(#3)}} %
\nc{\Typesrdk}{\Types{r}{d}{k}}
\nc{\Typesrd}[1]{\Types{r}{d}{#1}}
\nc{\Typesd}[2]{\Types{#1}{d}{#2}}
\nc{\isotypes}[2]{\ensuremath{\mathfrak{T}_{#1}^{#2}}}

\nc{\Typeslist}[3]{\ensuremath{\mathcal{L}_{#1}^{\sigma,#2}(#3)}}
\nc{\Typeslistrdk}{\Typeslist{r}{d}{k}}
\nc{\Typeslistrd}[1]{\Typeslist{r}{d}{#1}}
\nc{\Typeslistd}[2]{\Typeslist{#1}{d}{#2}}
\nc{\TypeslistRdk}{\Typeslist{R}{d}{k}}
\nc{\TypeslistRhatdk}{\Typeslist{\hat{R}}{d}{k}}
\nc{\TypeslistRd}[1]{\Typeslist{R}{d}{#1}}
\nc{\TypeslistRhatd}[1]{\Typeslist{\hat{R}}{d}{#1}}
\nc{\TypeslistRStrichdk}{\Typeslist{R'}{d}{k}}
\nc{\TypeslistRStrichd}[1]{\Typeslist{R'}{d}{#1}}

\nc{\type}{\ensuremath{\tau}}

\nc{\inducedSubStr}[2]{\ensuremath{#1[#2]}} 

\nc{\mytime}{\ensuremath{\textit{time}^{\sigma,d}_{n,k,r}}}

\begin{document}

\title{%
 First-Order Logic with Counting:\\
 At Least, \emph{Weak} Hanf Normal Forms Always Exist\\ and Can Be Computed!%
}

\author[1]{Dietrich Kuske}
\author[2]{Nicole Schweikardt}
\affil[1]{TU Ilmenau}
\affil[2]{Humboldt-Universit\"at zu Berlin}

\maketitle

\begin{abstract}
  We introduce the logic $\FOCN(\Ps)$ which extends first-order logic
  by counting and by numerical predicates from a set $\Ps$, and which can be
  viewed as a natural generalisation of various counting logics that
  have been studied in the literature.
 
  We obtain a locality result showing that every $\FOCN(\Ps)$-formula
  can be transformed into a formula in Hanf normal form that is
  equivalent on all finite structures of degree at most $d$. A formula is in
  Hanf normal form if it is a Boolean combination of formulas
  describing the neighbourhood around its tuple of free variables and
  arithmetic sentences with predicates from $\Ps$ over atomic
  statements describing the number of realisations of a type with a
  single centre. The transformation into Hanf normal form can be
  achieved in time elementary in $d$ and the size of the input
  formula.
  From this locality result, we infer the following applications:
  \begin{enumerate}[$\bullet$]
  \item The Hanf-locality rank of first-order formulas of bounded
    quantifier alternation depth only grows polynomially with the
    formula size.
  \item The model checking problem for the fragment $\FOC(\Ps)$ of
    $\FOCN(\Ps)$ on structures of bounded degree is fixed-parameter
    tractable (with elementary parameter
    dependence).
  \item The query evaluation problem for fixed queries from
    $\FOC(\Ps)$ over fully dynamic databases of degree at most $d$ can
    be solved efficiently: there is a dynamic algorithm that can
    enumerate the tuples in the query result with constant delay, and
    that allows to compute the size of the query result and to test if
    a given tuple belongs to the query result within constant time
    after every database update.
  \end{enumerate}
\end{abstract}

\section{Introduction}\label{section:introduction}

The counting ability of first-order logic is very limited: it can only
make statements of the form ``there are at least $k$ witnesses $x$ for
$\varphi(x)$'' for a constant $k\in\bN$. To overcome this problem, one
can add number variables $\kappa$ to first-order logic and means to
express that $\kappa$ equals the number of witnesses for the formula
$\varphi(x)$. In order to make use of these number variables, one also
adds numerical functions like addition and numerical predicates like
$\kappa\le\kappa'$ or ``$\kappa$ is a prime'' to the logic. These and
similar ideas led to the extensions of first-order logic by the
Rescher quantifier, the H\"artig quantifier, or arbitrary unary counting quantifiers
\cite{Res62,Hae65,Vae97}, to logics like $\FO(\quant{D}_p)$
from~\cite{Nur00}, $\FO(\mathrm{Cnt})$ from~\cite{Lib04}, and
$\FO{+}\mathsf{C}$ from~\cite{Gro13}. In this paper, we introduce an
extension $\FOCN(\Ps)$ of first-order logic by counting, number
variables, and numerical predicates from a set $\Ps$. By choosing
$\Ps$ appropriately, we use this extension as a general framework for
counting extensions of first-order logic (it subsumes all the logics
mentioned above).

Clearly, two isomorphic graphs cannot be distinguished by logical
sentences. Even more: suppose there is a bijection between two
undirected graphs $\cA$ and $\cB$ such that, for every node of $\cA$,
the neighbourhood of radius $2^{\bigOh(q)}$ of that node is isomorphic to
the neighbourhood of its image. Then, first-order logic cannot
distinguish the two graphs by first-order sentences of quantifier rank
at most $q$ (this goes back to \cite{Han65}, the actual bound
$2^{q-1}-1$ was obtained in~\cite{Lib00}).  Consequently, to determine
whether a sentence $\varphi$ of quantifier rank $q$ holds in an
undirected graph $\cA$, it suffices to count how often each
neighbourhood type of radius $2^{\bigO(q)}$ is realised in $\cA$.

It actually suffices to count these realisations up to a certain
threshold (that depends on $q$ and the degree $d$ of the graph $\cA$)
\cite{FagSV95}. Bounding the degree of $\cA$ by $d$, there are only
finitely many neighbourhood types of radius $2^{\bigO(q)}$ that can be
realised. Consequently, this condition can be expressed as a
first-order sentence; i.e., as a first-order sentence in \emph{Hanf
  normal form}.

A similar story can be told, e.g., for the extension
$\FO(\quant{D}_2)$ of first-order logic by the ability to express that
the number of witnesses for $\varphi(x)$ is even. To determine whether
such a sentence holds in an undirected graph $\cA$, one has to count
the number of realisations up to a certain threshold and one has to
determine the parity of this number~\cite{Nur00}. Again, this leads to
a sentence in Hanf normal form that expresses the said condition in
the graph~$\cA$.

We say that a logic can only express local properties if validity of a
sentence in a structure can be determined by solely counting the
number of realisations of neighbourhood types. This property has
traditionally been proven by suitable notions of games. Often, the
\emph{existence} of a Hanf normal form follows from this directly. But
there is no obvious way to extract an \emph{algorithm} for the
construction of it.  On the other hand, these Hanf normal forms have
also found various applications in algorithms and complexity (cf.,
e.g., \cite{See96,Lib04,FriG01,%
  DurandGrandjean,Kreutzer-AMT-Survey,KazanaSegoufin-boundedDegree,%
  DBLP:conf/stacs/Segoufin14,BolK12,HeiKS16,BKS_ICDT17}).  In
particular, there are very general algorithmic meta-theorems stating
that model checking is fixed-parameter tractable for various classes
of structures, and that the results of queries against various classes
of databases can be enumerated with constant delay after a linear-time
preprocessing phase.  In this context, questions about the efficiency
of the normal forms have recently attracted interest (cf.\ e.g.,
\cite{DGKS07,Lin08,BolK12,HeiKS16}).

The main result of this paper is the effective construction of a Hanf
normal form from an arbitrary formula from our logic
$\FOCN(\Ps)$. This construction extends the constructions from
\cite{BolK12,HeiKS16} and can be carried out in 5-fold exponential
time.
We also provide a 4-fold exponential lower bound. 
From the existence and the computability of Hanf normal forms, we
infer four applications:
\begin{itemize}
\item The model checking problem for the (large) fragment 
  $\FOC(\Ps)$ of the logic $\FOCN(\Ps)$ on
  structures of bounded degree is fixed-parameter tractable (with
  elementary parameter dependence) where we assume an oracle for the
  numerical predicates from $\Ps$. 
\item The Hanf-locality rank of first-order formulas of bounded
  quantifier alternation depth only grows polynomially with the
  formula size. This complements Libkin's bound $2^{q-1}-1$ for $q$
  the quantifier rank of the formula \cite{Lib00} and (partly) proves
  a conjecture from \cite{KusL11}.
\item From a sentence $\varphi$ in $\FOCN(\Ps)$, we can compute a
  first-order description of the numerical condition that is
  equivalent to validity of $\varphi$. This first-order description is
  expressed in an extension of integer arithmetic with the predicates
  from~$\Ps$.
\item The query evaluation problem for fixed queries from $\FOC(\Ps)$
  over fully dynamic databases of degree $\leq d$ can be solved
  efficiently: there is a dynamic algorithm that can enumerate the
  tuples in the query result with constant delay, and that allows to
  compute the size of the query result and to test if a given tuple
  belongs to the query result within constant time after every
  database update.
\end{itemize}

Above, we said that the existence of a Hanf normal form follows
``often''. A counterexample to this is the fragment $\FO(\Ps)$ of
$\FOCN(\Ps)$ that we consider in \cite{HeiKS16}. The problem there is
that, in general, $\FO(\Ps)$ does not allow to formulate the necessary
numerical condition. In Corollary~\ref{cor:whnf} we present a
weakening of the notion of a Hanf normal form that also works in this
case.

The rest of the paper is structured as follows. 
Sections~\ref{section:FOCN} and \ref{section:GeneralsedHNF} introduce
the logic $\FOCN(\Ps)$ and the according notion of Hanf normal form.
Theorem~\ref{thm:main} summarises the paper's technical main result,
the proof of which is given in Sections~\ref{sec:upperbound} and
\ref{sec:lowerbound}.
Section~\ref{sec:applications} describes the mentioned applications.

\smallskip

\textbf{Acknowledgements.}
The second author would like to acknowledge the financial support by the
Deutsche Forschungsgemeinschaft (DFG, German Research Foundation) under grant SCHW~837/5-1.

\section{First-order logic with counting and numerical predicates}\label{section:FOCN}\label{section:BasicNotation}

We write 
$\bZ$, $\bN$, and $\bN_{\ge1}$ for the sets of integers, non-negative integers, and
positive integers, resp.
For all $m,n\in\bN$, we write $[m,n]$ for the set
$\setc{k\in\bN}{m\le k\le n}$ and $[m]=[1,m]$.
For a $k$-tuple $\ov{x}=(x_1,\ldots,x_k)$ we write $|\ov{x}|$ to
denote its \emph{arity} $k$.
The exponential functions $\exp_k\colon\bN\to\bN$ are defined by
induction on $k$ via $\exp_0(n)=n$ and $\exp_{k+1}(n)=2^{\exp_k(n)}$ for
all $k\in\bN$. We write $\poly(n)$ for the set of functions
$\bigcup_{k\in\bN}\bigO(n^k)$.

A \emph{signature} $\sigma$ is a \emph{finite} set of relation and
constant symbols.
Associated with every relation symbol $R\in\sigma$ is a
positive integer $\ar(R)$ called the \emph{arity} of $R$.  The
\emph{size} $\size{\sigma}$ of a signature $\sigma$ is the number of
its constant symbols plus the sum of the arities of its relation
symbols.  We call a signature \emph{relational} if it does not contain
any constant symbol.
A \emph{$\sigma$-structure} $\cA$ consists of a \emph{finite}
non-empty
set $A$ called the \emph{universe} of $\cA$, a relation
$R^{\cA} \subseteq A^{\ar(R)}$ for each relation symbol $R \in \sigma$,
and an element $c^{\cA}\in A$ for each constant symbol $c \in \sigma$.
Note that according to these definitions, all signatures and all
structures considered in this paper are \emph{finite}. To indicate
that two $\sigma$-structures $\cA$ and $\cB$ are isomorphic, we write
$\cA\cong\cB$.

We define the extension $\FOCN(\Ps)$ of first-order logic $\FO$ by
counting and by numerical predicates from a set $\Ps$.  Our notation extends
standard notation concerning first-order logic,
cf.~\cite{EbbF95,Lib04}.

Let $\VARS$ and $\NVARS$ be  fixed disjoint countably infinite sets
of \emph{structure} and \emph{number variables}, respectively. In our logic,
structure variables from $\VARS$ will always denote elements of the
structure, and number variables from $\NVARS$ will denote integers.
Typical structure variables are $x$ and $y$, typical
number variables are $\lambda$ and $\kappa$. Often, we use $z$ as an
arbitrary variable from $\VARS\cup\NVARS$. 

A \emph{$\sigma$-interpretation} $\I=(\A,\beta)$ consists of a
$\sigma$-structure $\A$ and an \emph{assignment $\beta$ in $\A$}, i.e.,
$\beta\colon\VARS\cup\NVARS\to A\cup\ZZ$ with
$\beta(x)\in A$ for $x\in\VARS$ and $\beta(\kappa)\in\bZ$ for
$\kappa\in\NVARS$. For $k,\ell\in\NN$, for $a_1,\ldots,a_k\in A$,
$n_1,\ldots,n_\ell\in\ZZ$, and for pairwise distinct
$y_1,\ldots,y_k\in\VARS$ and $\kappa_1,\dots,\kappa_\ell\in\NVARS$,
we write
$\beta\frac{a_1,\ldots,a_k}{y_1,\ldots,y_k}\frac{n_1,\ldots,n_\ell}{\kappa_1,\ldots,\kappa_\ell}$
for the assignment $\beta'$ in $\A$ with $\beta'(y_j)=a_j$ for all
$j\in [k]$, $\beta(\kappa_j)=n_j$ for all $j\in[\ell]$, and
$\beta'(z)=\beta(z)$ for all
$z\in(\VARS\cup\NVARS)\setminus\set{y_1,\ldots,y_k,\kappa_1,\ldots,\kappa_\ell}$.
For $\I=(\A,\beta)$ we let
$\I\frac{a_1,\ldots,a_k}{y_1,\ldots,y_k}\frac{n_1,\ldots,n_\ell}{\kappa_1,\ldots,\kappa_\ell}
\ = \
\big(\A,\beta\frac{a_1,\ldots,a_k}{y_1,\ldots,y_k}\frac{n_1,\ldots,n_\ell}{\kappa_1,\ldots,\kappa_\ell}\big)$.

\begin{definition}[\mbox{$\FO[\sigma]$}]\label{def:FO} \ \\
Let $\sigma$ be a signature.  The set of
\emph{$\FO[\sigma]$-formulas} is built according to the following
rules:
  \begin{enumerate}[(1)]
  \item\label{item:atomic} $x_1{=}x_2$ and $R(x_1,\ldots,x_{\ar(R)})$
    are \emph{formulas}, where $R\in\sigma$ and
    $x_1,x_2,\ldots,x_{\ar(R)}$ are structure variables or constant
    symbols in $\sigma$
  \item\label{item:bool} if $\varphi$ and $\psi$ are \emph{formulas},
    then so are $\lnot\varphi$ and $(\varphi\lor\psi)$
  \item\label{item:exists} if $\varphi$ is a \emph{formula} and
    $y\in\VARS$, then $\exists y\,\varphi$ is a \emph{formula}
  \end{enumerate}

  The semantics $\sem{\phi}^{\I}\in\set{0,1}$ for a
  $\sigma$-interpretation $\I=(\cA,\beta)$ and a formula
  $\phi$ is defined 
  as usual:

  \begin{enumerate}[(1)]

  \item $\sem{x_1{=}x_2}^{\I}=1$ if $a_1=a_2$ and $\sem{x_1{=}x_2}^{\I}=0$
    otherwise,
    \\
    $\sem{R(x_1,\ldots,x_{\ar(R)})}^{\I}=1$ if
    $(a_1,\ldots,a_{\ar(R)})\in R^\A$, and
    $\sem{R(x_1,\ldots,x_{\ar(R))}}^{\I}=0$ otherwise,
    \\
    where for $j\in\set{1,\ldots,\max\set{2,\ar(R)}}$ we let
    $a_j=\beta(x_j)$ if $x_j\in\VARS$ and $a_j=x_j^{\A}$ if $x_j$ is a
    constant symbol in $\sigma$

  \item $\sem{\nicht\varphi}^{\I}=1-\sem{\phi}^{\I}$ and
    $\sem{(\phi\oder\psi)}^{\I}=
    \max\set{\sem{\phi}^{\I},\sem{\psi}^{\I}}$

  \item
    $\sem{\exists
      y\,\varphi}^{\I}=\max\setc{\sem{\phi}^{\I\frac{a}{y}}}{a\in A}$
  \end{enumerate}
\end{definition}

In a first step, we extend first-order logic such that numerical
statements on the number of witnesses for a formula are
possible. These numerical statements are based on numerical predicates
that we define first.

\begin{definition}[Numerical predicate collection]\label{def:Qs} \ \\
A \emph{numerical predicate collection} is a triple
$(\Ps,\ar,\sem{.})$ where $\Ps$ is some countable set of
\emph{predicate names}, $\ar\colon\Ps\to\bN_{\ge1}$ assigns the
\emph{arity} to every predicate name, and
$\sem{\P}\subseteq\ZZ^{\ar(\P)}$ is the \emph{semantics} of the
predicate name $\P\in\Ps$.  
\end{definition}

Basic examples of numerical predicates are
$\P_+$, $\P_{\cdot}$, $\P_=$, $\P_{\leq}$, $\quant{Prime}$ with
$\sem{\P_+}=\setc{(m,n,m+n)}{m,n\in\bZ}$, 
$\sem{\P_\cdot}=\setc{(m,n,m\cdot n)}{m,n\in\bZ}$, 
$\sem{\P_=}=\setc{(m,m)}{ m\in\bZ}$, 
$\sem{\P_\leq}=\setc{(m,n)\in\bZ^2}{ m\leq n}$, and
$\sem{\quant{Prime}} = \setc{n\in\NN}{n\text{ is a prime number}}$.
Also, $\quant{D}_p$ with $\sem{\quant{D}_p} = p\bZ$ (for each fixed $p\in\NNpos$) and the halting problem
(i.e., the set of indices of Turing machines that halt with empty
input) are possible numerical predicates.

\begin{definition}[\mbox{$\FO(\Ps)[\sigma]$}]
  \label{def:FO(Ps)} 
  \ \\
  Let $\sigma$ be a signature and $(\Ps,\ar,\sem{.})$
  be a numerical predicate collection. 
  The sets of \emph{formulas} and \emph{counting terms for
    $\FO(\Ps)[\sigma]$} are built according to the rules
  \eqref{item:atomic}--\eqref{item:exists} and the following rules:
  \begin{enumerate}[(1)]
    \addtocounter{enumi}{3}
  \item\label{item:Q} if $\P\in\Ps$, $m= \ar(\P)$, and
    $t_1,\ldots,t_m$ are \emph{counting terms}, then
    $\P(t_1,\ldots,t_m)$ is a \emph{formula}
  \item[(5')] if $\varphi$ is a formula, $y\in\VARS$, and $k\in\bN$, then
    $\Count{(y)}{\varphi}-k$ is a \emph{counting term}.
  \end{enumerate}

  Let $\I=(\cA,\beta)$ be a $\sigma$-interpretation. For every
  formula $\phi$ and every
  counting term $t$ from $\FO(\Ps)[\sigma]$, the semantics
  $\sem{\phi}^{\I}\in\set{0,1}$ of $\phi$ in $\I$ and the semantics
  $\sem{t}^{\I}\in\ZZ$ of $t$ in $\I$ extend the definition for
  $\FO[\sigma]$-formulas as follows:

  \begin{enumerate}
  \item[(4)] $\sem{\P(t_1,\ldots,t_m)}^{\I}=1$ if
    $\big(\sem{t_1}^{\I},\ldots,\sem{t_m}^{\I}\big)\in\sem{\P}$, and
    $\sem{\P(t_1,\ldots,t_m)}^{\I}=0$ otherwise

  \item[(5')]
    $\sem{\Count{(y)}{\varphi}-k}^{\I}= \big| \bigsetc{
      a\in A \;}{\;
      \sem{\phi}^{\I\frac{a}{y}}=1 } \big|-k$
  \end{enumerate}
\end{definition} 
 
We will write $\Count{(y)}{\varphi}$ as a shorthand for the counting term
$\Count{(y)}{\varphi}-0$.

\begin{remark}
  For the logic $\FO(\Ps)$ and the following logics $\FOC(\Ps)$ and $\FOCN(\Ps)$, an
  \emph{expression} is a formula or a counting term.

  As usual, for a \emph{formula} $\phi$ and a $\sigma$-interpretation $\I$ we will often write
  $\I\models\phi$ to indicate that $\sem{\phi}^\I =1$. Accordingly, $\I\not\models\phi$ indicates that $\sem{\phi}^{\I}=0$.
\end{remark}

For structure variables $y\in\VARS$, the  quantifier $\exists y$ 
can be replaced by using a suitable numerical predicate:

\begin{example}
  Let $\P_{\exists}$ be the numerical predicate with
  $\ar(\P_{\exists})=1$ and $\sem{\P_{\exists}}=\NNpos$.  
  Consider an arbitrary $\sigma$-interpretation $\I=(\cA,\beta)$.
  Since $\cA$ is finite, we have
  \begin{align*}
    \I\models\P_{\exists}(\Count{(y)}{\varphi})
      &\iff |\setc{a\in A\;}{\;\textstyle\I\frac{a}{y}\models\varphi}  |\ \ \in \ \ \sem{\P_{\exists}} \ = \ \bN_{\ge1}\\
      &\iff
        \text{ there is some $a\in A$ with }\textstyle\I\frac{a}{y}\models\varphi
       &&\text{ (since $A$ is finite)}\\
      &\iff \I\models\exists y\,\varphi\,.
  \end{align*}
 Thus, we have
 \[
   \I\models\P_{\exists}(\Count{(y)}{\phi})
   \ \iff \
   \I\models\exists y\,\phi\,.
 \]
\end{example}

The following examples provide choices of $\Ps$ for which the logic
$\FO(\Ps)$ has been studied in the literature.

\begin{example}\label{E-list} \ \hfill
  \begin{enumerate}[(a)]
  \item Let $\Es=\setc{\exists^{\ge k}}{k\in\NNpos}$ with
    $\ar(\exists^{\geq k})=1$ and
    $\sem{\exists^{\ge k}}=\set{k,k{+}1,\dots}$ for every $k\geq 1$.
    The logic $\FO(\Es)$ is equivalent to the logic $\FO(\textup{C})$
    of \cite{EbbF95}.
  \item\label{E-divisibility} The logic $\FO(\set{\quant{D}_p})$ is
    equivalent to the extension of first-order logic by the
    divisibility quantifier $\quant{D}_p$, considered by Nurmonen
    in~\cite{Nur00}.
  \item Let $(\Ps,\ar,\sem{.})$ be a numerical predicate collection
    with $\ar(\P)=1$ for all $\P\in\Ps$. Then, $\FO(\Ps)$ is equivalent
    to the logic considered in \cite{HeiKS16}.
  \item For all formulas $\phi$ and $\psi$ and for every $\sigma$-interpretation $\I=(\A,\beta)$ we have
  \[
    \I\ \models \ \P_\le\big(\Count{(y)}{\phi},\,\Count{(y)}{\psi}\big)
    \ \ \iff \ \
   |\setc{a\in A\ }{\ \textstyle\I\frac{a}{y}\models\phi}|
   \ \ \leq \ \
   |\setc{a\in A\ }{\ \textstyle\I\frac{a}{y}\models\psi}|.
  \]
  Analogously, we have 
  \[
    \I\ \models \ \P_=\big(\Count{(y)}{\phi},\,\Count{(y)}{\psi}\big)
    \ \ \iff \ \
   |\setc{a\in A\ }{\ \textstyle\I\frac{a}{y}\models\phi}|
   \ \ = \ \
   |\setc{a\in A\ }{\ \textstyle\I\frac{a}{y}\models\psi}|.
  \]
  Thus,
  the logics $\FO(\set{\P_\le})$ and $\FO(\set{\P_=})$ are
    equivalent to the extension of first-order logic by the
    \emph{Rescher quantifier} and the \emph{H\"artig quantifier}, resp.~\cite{Res62,Hae65}.
  \item\label{E-ultimately-periodic} Let
    $\Us=\set{0,1}^*\$\set{0,1}^+$ and, for $u\$v\in\Us$, let
    $\sem{u\$v}\subseteq\bN$ be the set with characteristic sequence
    $uv^\omega$, i.e., $i\in\sem{u\$v}$ if, and only if, the
    $\omega$-word $uv^\omega$ carries a $1$ at position $i\in\bN$
    (here, we follow the convention that the leftmost position of an
    $\omega$-word is position 0).
    Note that a set $X\subseteq\bN$ is ultimately periodic (or,
    semilinear) if, and only if, there is some $w\in\Us$ with
    $\sem{w}=X$.
    The logic $\FO(\Us)$ is equivalent to the extension of first-order
    logic by ultimately periodic unary counting quantifiers,
    considered in \cite{HeiKS16}.
  \item \label{example:FOunary} In \cite[Sect.~8.1]{Lib04}, Libkin
    considers the extension $\FO(\text{unary})$ of first-order logic
    by the class of all unary generalised quantifiers.  
    It is not
    difficult to see that every $\FO(\text{unary})$-formula is equivalent to an $\FO(\Ps)$-formula, for a suitable
    numerical predicate collection
    $(\Ps,\ar,\sem{.})$:

  For the definition of $\FO(\text{unary})$, let
  $\nu_n=\set{R_1,\ldots,R_n}$ be the relational signature that consists of
  $n$ unary relation symbols. Formulas of $\FO(\text{unary})$ are
  built from the rules (1)--(3) and the following additional rule:
  \begin{enumerate}[(U)]
  \item[(U)] if $n\in\bN$, $\mathcal K$ is a class of
    $\nu_n$-structures that is closed under isomorphism, $y\in\VARS$, and
    $\varphi_i$ for $i\in[n]$ are formulas,  then
    $\quant{Q}_{\mathcal K}y\,(\varphi_1,\varphi_2,\dots,\varphi_n)$ is a
    formula.
  \end{enumerate}
  For the semantics of the logic $\FO(\text{unary})$, we only need to
  explain the meaning of this last formula: For a $\sigma$-interpretation $\I=(\A,\beta)$
  we have
  $\I\models\quant{Q}_{\mathcal K}y\,(\varphi_1,\dots,\varphi_n)$
  iff the structure
  $(A,\varphi_1^{\I},\varphi_2^{\I},\dots,\varphi_n^{\I})$
  belongs to $\mathcal K$ where
  \[
     \varphi_i^{\I} \ \ = \ \ \setc{\,a\in A \ }{ \ \textstyle\I\frac{a}{y}\models\varphi_i\,}\,.
  \]
  To construct a numerical predicate collection $(\Ps,\ar,\sem{.})$ for which our logic $\FO(\Ps)$ 
  is at least as expressive as 
  $\FO(\text{unary})$, we proceed as follows: Let $n\in\bN$ and let
  $\cB$ be a (finite) $\nu_n$-structure. The \emph{characteristic sequence of $\cB$} is
  the tuple $\chi(\cB)=(b_S)_{S\subseteq[n]}\in\bN^{2^n}$ that for
  all $S\subseteq[n]$ gives the number of elements in
  $\bigcup_{i\in S}R^\cB_i$, i.e.,
  \[
     b_S \ \ = \ \ \left|\,\bigcup_{i\in S}R_i^\cB\,\right|\,.
  \]
  Since $\nu_n$ contains only unary predicates, we get that
  $\cB_1\cong\cB_2$ iff $\chi(\cB_1)=\chi(\cB_2)$, for all
  $\nu_n$-structures $\cB_1$ and $\cB_2$.

  Now let $\mathcal K$ be a class of $\nu_n$-structures. We define a
  numerical predicate $\P^{\mathcal K}$ of arity $2^n$ with
  \[
    \sem{\P^{\mathcal K}} \ = \ \setc{\,\chi(\cB)\,}{\,\cB\in\mathcal K\,}\,.
  \]
  Then, for every $\sigma$-interpretation $\I=(\A,\beta)$ we have
  \[
    \I\ \models \ \quant{Q}_{\mathcal K}y\,(\varphi_1,\dots,\varphi_n)
    \quad \iff \quad
    \I \ \models \ \P^{\mathcal K}
    \Big( 
      \big( \Count{(y)}{\bigvee_{i\in S}\varphi_i} \big)_{S\subseteq[n]}
    \Big)\,.
  \]
  Thus, the $\FO(\text{unary})$-formula
  $\quant{Q}_{\mathcal K}y\,(\varphi_1,\dots,\varphi_n)$ 
  is equivalent to the $\FO(\set{\P^{\mathcal K}})$-formula
  \[
    \P^{\mathcal K}
    \Big( 
      \Big( \Count{(y)}{\bigvee_{i\in S}\varphi_i} \Big)_{S\subseteq[n]}
    \Big)\,.
  \] 
  \end{enumerate}
\end{example}

Our next logic $\FOC(\Ps)$ allows not only numerical statements on
numbers given by counting terms of the form
$\Count{(y)}{\varphi}-k$, but on polynomials over such
terms. In addition, the logic $\FOC(\Ps)$ allows to count tuples.

\begin{definition}[\mbox{$\FOC(\Ps)[\sigma]$}]
  \label{def:FOC} \ \\
  Let $\sigma$ be a signature and let $(\Ps,\ar,\sem{.})$ be
  a numerical predicate collection.  The set of \emph{expressions for
    $\FOC(\Ps)[\sigma]$} is built according to the rules
  \eqref{item:atomic}--\eqref{item:Q} and the following rules:
  \begin{enumerate}[(1)]
  \addtocounter{enumi}{4}
  \item\label{item:countterm} if $\varphi$ is a \emph{formula},
    $k\in\NNpos$, and $\ov{y}=(y_1,\ldots,y_k)$ is a tuple of pairwise
    distinct structure variables (i.e., variables in $\VARS$), then $\Count{\ov{y}}{\varphi}$ is a
    \emph{counting term}
  \item\label{item:constterm} every integer $i\in\ZZ$ is a
    \emph{counting term}
  \item\label{item:plustimesterm} if $t_1$ and $t_2$ are
    \emph{counting terms}, then so are $(t_1+t_2)$ and
    $(t_1\cdot t_2)$
  \end{enumerate}
  Let $\I=(\cA,\beta)$ be a $\sigma$-interpretation. For every
  expression $\xi$ of $\FOC(\Ps)[\sigma]$, the semantics
  $\sem{\xi}^\I$ is given by the semantics for the rules (1)--(4) and the following:
  \begin{enumerate}[(1)]
    \addtocounter{enumi}{4}
  \item
    $\sem{\Count{\ov{y}}{\varphi}}^{\I}= \big| \bigsetc{
      (a_1,\ldots,a_k)\in A^k \;}{\;
      \sem{\phi}^{\I\frac{a_1,\ldots,a_k}{y_1,\ldots,y_k}}=1 } \big|$,
    \ where $\ov{y}=(y_1,\ldots,y_k)$ 
  \item $\sem{i}^{\I}=i$
  \item $\sem{(t_1 + t_2)}^{\I}= \sem{t_1}^{\I} + \sem{t_2}^{\I}$, \
    and \
    $\sem{(t_1 \cdot t_2)}^{\I}= \sem{t_1}^{\I} \cdot \sem{t_2}^{\I}$
  \end{enumerate}
\end{definition}

If $s$ and $t$ are counting terms, then we write $s-t$ for the
counting term $(s+((-1)\cdot t))$. With this convention, we can understand
$\FO(\Ps)$ as a fragment of $\FOC(\Ps)$. Note that counting terms of
$\FOC(\Ps)$ are polynomials while counting terms of $\FO(\Ps)$ are
special linear polynomials. In addition, counting terms of $\FOC(\Ps)$
can count \emph{tuples} of elements of the universe while counting
terms of $\FO(\Ps)$ only count single \emph{elements} of the universe.

\begin{example}
  The following $\FOC(\set{\quant{Prime}})$-formula (expressing that
  the sum of the numbers of nodes and edges of a graph is a prime) is
  not an $\FO(\set{\quant{Prime}})$-formula:
  \[
     \quant{Prime}\,\big(\,(\,\Count{(x)}{x{=}x}\ + \ \Count{(x,y)}{E(x,y)}\,)\,\big)\,.
  \]
\end{example}

Our final extension of the logic allows, besides structure variables
also number variables, and it allows to quantify over ``small''
numbers, i.e., over numbers in $\set{0,1,\ldots,|A|}$, when evaluated in a $\sigma$-structure $\A$:

\begin{definition}[\mbox{$\FOCN(\Ps)[\sigma]$}]
\label{def:FOCN}
\ \\
  Let $\sigma$ be a signature and let $(\Ps,\ar,\sem{.})$ be
  a numerical predicate collection.  The set of \emph{expressions for
    $\FOCN(\Ps)[\sigma]$} is built according to the rules
  \eqref{item:atomic}--\eqref{item:plustimesterm} and the following
  rules:
  \begin{enumerate}[(1)]
    \addtocounter{enumi}{7}
  \item\label{item:varterm} every variable from $\NVARS$ is a
    \emph{counting term}
  \item\label{item:exists:number} if $\varphi$ is a \emph{formula} and
    $\kappa\in\NVARS$, then $\exists \kappa\,\varphi$ is a \emph{formula}
  \end{enumerate}

  Let $\I=(\cA,\beta)$ be a $\sigma$-interpretation. For every
  expression $\xi$ of $\FOCN(\Ps)[\sigma]$, the semantics
  $\sem{\xi}^\I$ is given by the semantics for the rules (1)--(7) and the following:
  \begin{enumerate}[(1)]
    \addtocounter{enumi}{7}
  \item $\sem{\kappa}^{\I}=\beta(\kappa)$ for $\kappa\in\NVARS$
  \item  $\sem{\exists
      \kappa\,\varphi}^{\I}=\max\setc{\sem{\phi}^{\I\frac{k}{\kappa}}}{k\in\set{0,1,\ldots,|A|}}$
  \end{enumerate}
\end{definition}

By $\FOCN(\Ps)$, we denote the union of all $\FOCN(\Ps)[\sigma]$ for
arbitrary signatures~$\sigma$, and similarly for $\FOC(\Ps)$,
$\FO(\Ps)$, and $\FO$.

\begin{example}
  Let $\Ps=\set{\quant{Prime},\P_{=}}$ and
  consider the formula
  \[
   \exists\kappa\ \quant{Prime}
    \Big(\Count{(y)}{\P_=\big(\kappa,\;\Count{(z)}{E(y,z)}\big)}\Big)\,.
  \]
  The counting term $\Count{(z)}{E(y,z)}$ denotes the out-degree of
  $y$, hence the formula $\P_=\big(\kappa,\,\Count{(z)}{E(y,z)}\big)$ expresses
  that $\kappa$ is the out-degree of $y$. Consequently, the whole
  formula says that there is some degree $\kappa$ such that the number
  of nodes of out-degree $\kappa$ is a prime. Since $0$ is not a prime,
  this $\FOCN(\Ps)$ formula is equivalent to the following
  $\FO(\Ps)$-formula
  \[
    \exists x\ \quant{Prime}
    \biggl(\Count{(y)}{\P_=\bigl(\,\Count{(z)}{E(x,z)},\;\Count{(z)}{E(y,z)}\,\bigr)}\biggr)\,.
  \]
\end{example}

\begin{remark}
  The logics $\FO(\mathrm{Cnt})$ from \cite{Lib04} and
  $\FO{+}\mathsf{C}$ from \cite{Gro13} can be viewed as fragments of
  $\FOCN(\Ps)$ where $\Ps$ contains the predicates $\P_+$, $\P_\cdot$,
  $\P_=$ and $\P_\le$. But these two logics have no mechanism for
  counting \emph{tuples}.  E.g., it is not clear how to express in
  $\FO(\mathrm{Cnt})$ or $\FO{+}\mathsf{C}$ that the number of edges
  of a graph is a square number, while this is
  $\FOCN(\Ps)$-expressible by
  $\exists\kappa\
  \P_=\big(\,\Count{(x,y)}{E(x,y)}\,,\;(\kappa\cdot\kappa)\,\big)$.
\end{remark}

Note that we restrict the quantification over numbers to the size of
the universe of the structure~$\cA$.  This is analogous to the
semantics of the logics $\FO(\mathrm{Cnt})$ and $\FO{+}\mathsf{C}$
from \cite{Lib04,Gro13}. As a consequence, the logic
$\FOCN(\Ps)[\sigma]$ does not have the full power of integer
arithmetic.  Let us mention that our main result
Theorem~\ref{thm:main} also holds for the variant of $\FOCN(\Ps)$
where quantifications of number variables range over arbitrary
integers (rather than just numbers in $\set{0,1,\ldots,|A|}$); the
model-checking algorithm described in
Section~\ref{sec:applications}, however, does not carry over to this
variant of $\FOCN(\Ps)$.

The construct $\exists z$ binds the variable $z\in\VARS\cup\NVARS$, and
the construct $\#\ov{y}$ in a counting term binds the (structure)
variables from the tuple~$\ov{y}$; all other occurrences of variables
are free. We denote the set of free variables of the expression $\xi$
by $\free(\xi)$. I.e., 
the free variables $\free(\xi)$ of $\FOCN(\Ps)$-expressions
$\xi$ are inductively defined as follows:

\begin{enumerate}[(1)]
 \item
  $\free(x_1{=}x_2)=\set{x_1,x_2}\cap\VARS$
  \ and \ 
  $\free(R(x_1,\ldots,x_{\ar(R)}))=\set{x_1,\ldots,x_{\ar(R)}}\cap\VARS$
 \item
  $\free(\nicht\phi)=\free(\phi)$ \ and \ $\free((\phi\oder\psi))=\free(\phi)\cup\free(\psi)$
 \item
  $\free(\exists y\,\phi)=\free(\phi)\setminus\set{y}$
 \item
  $\free(\P(t_1,\ldots,t_m))=\free(t_1)\cup\cdots\cup\free(t_m)$
 \item
  $\free(\Count{(y_1,\dots,y_k)}{\phi}) = \free(\phi)\setminus\set{y_1,\ldots,y_k}$
 \item
  $\free(i)=\emptyset$ for $i\in\bZ$
 \item
  $\free((t_1+t_2)) \ = \ \free((t_1\cdot t_2)) \ = \ \free(t_1)\cup\free(t_2)$
 \item 
  $\free(\kappa)=\{\kappa\}$ for $\kappa\in\NVARS$
 \item
  $\free(\exists \kappa\,\phi) = \free(\phi)\setminus\set{\kappa}$
\end{enumerate}

We will often write $\xi(\ov{z})$, for
$\ov{z} = (z_1,\ldots,z_n)$ with $n\ge 0$, to indicate that at most
the variables from $\set{z_1,\ldots,z_n}$ are free in the expression~$\xi$.

A \emph{sentence} is a formula without free variables, a
\emph{ground term} is a counting term without free variables.
Furthermore, a \emph{number formula} is a formula whose free variables
all belong to $\NVARS$. For instance,
$\P(\kappa,\Count{(y)}{\varphi(y,\kappa)})$ is a number formula, but
not a sentence since $\kappa$ is free in this formula.

Note that the semantics $\sem{\xi}^\I$ for an expression
$\xi(\ov{x},\ov{\kappa})$ and a $\sigma$-interpretation $\I=(\A,\beta)$ only
depends on $\cA$ and $\beta(z)$ for the variables $z$ in
$\ov{x},\ov{\kappa}$.

Let us consider an $\FOCN(\Ps)[\sigma]$-counting term
$t(\ov{x},\ov{\kappa})$, for $\ov{x}=(x_1,\ldots,x_m)$ and
$\ov{\kappa}=(\kappa_1,\ldots,\kappa_n)$.  If $\cA$ is a
$\sigma$-structure, $\ov{a} = (a_1,\ldots,a_m)\in A^m$ and
$\ov{k} = (k_1,\ldots,k_n)\in \set{0,\ldots,|A|}^n$, we write
$t^{(\A,\ov{a},\ov{k})}$ or $t^\A[\ov{a},\ov{k}]$ for the integer
$\sem{t}^{(\A,\beta)}$, where $\beta$ is an assignment in $\A$ with
$\beta(x_j)=a_j$ for all $j\in [m]$ and $\beta(\kappa_j)=k_j$ for all
$j\in[n]$. Furthermore, for an $\FOCN(\Ps)[\sigma]$-formula
$\phi(\ov{x},\ov{\kappa})$ we write
$(\cA,\ov{a},\ov{k}) \models \varphi$ or
$\cA \models \varphi[\ov{a},\ov{k}]$ to indicate that
$\sem{\phi}^{(\A,\beta)}=1$, i.e., the formula
$\varphi(\ov{x},\ov{\kappa})$ is satisfied in $\cA$ when interpreting
the free occurrences of the structure variables $x_1,\ldots,x_m$ with
$a_1,\ldots,a_m$ and the free occurrences of the number variables
$\kappa_1,\ldots,\kappa_n$ with $k_1,\dots,k_n$.  In case that
$m=n=0$ (i.e., $\phi$ is a sentence and $t$ is a ground term), we
simply write $t^{\A}$ instead of $t^{\A}[\ov{a},\ov{k}]$, and we write
$\A\models \phi$ instead of $\A\models\phi[\ov{a},\ov{k}]$.

Let $\ov{x}\in\VARS^m$, $\ov{y}\in\VARS^j$,
$\ov{\kappa}\in\NVARS^n$, $\cA$ a $\sigma$-structure, $\ov{a}\in A^m$,
and $\ov{k}\in\ZZ^n$.  Note that if
$\phi(\ov{x},\ov{\kappa},\ov{y})$ is a formula, then
$t(\ov{x},\ov{\kappa}):=\Count{\ov{y}}{\phi}$ is a counting term, such
that $t^{\A}[\ov{a},\ov{k}]$ is the number of tuples $\ov{b}\in A^j$
for which $\A\models\phi[\ov{a},\ov{k},\ov{b}]$.
Furthermore, a formula $\psi(\ov{x},\ov{\kappa})$ of the form $\P(t_1,\ldots,t_\ell)$ is
satisfied by a $\sigma$-structure $\A$ and tuples $\ov{a}\in A^m$ and
$\ov{k}\in\ZZ^n$ iff the tuple of integers
$(i_1,\ldots,i_\ell)$ belongs to the relation $\sem{\P}$, where
$i_j=t_j^{\A}[\ov{a},\ov{k}]$ for every $j\in [\ell]$.

Two formulas or two counting terms $\xi$ and $\xi'$ are \emph{equivalent} (for short,
$\xi\equiv\xi'$), if $\sem{\xi}^{\I}=\sem{\xi'}^{\I}$ for every
$\sigma$-interpretation~$\I$.

The size $\size{\xi}$ of an expression is its
length when viewed as a word over the alphabet
$\sigma \cup \VARS \cup \NVARS \cup \Ps \cup \set{,} \cup
\set{=,\exists,\neg,\lor, (, )}\cup\set{\#,.}$.

The \emph{number quantifier rank} $\nqr(\xi)$ of an
$\FOCN(\Ps)$-expression $\xi$ is the maximal nesting depth of
quantifiers of the form $\exists\kappa$ with $\kappa\in\NVARS$. The
\emph{binding rank} $\br(\xi)$ of $\xi$ is  the maximal
nesting depth of constructs of the form $\exists y$ with $y\in\VARS$
and $\#\ov{y}$ with $\ov{y}$ a tuple in $\VARS$.  The \emph{binding
  width} $\bw(\xi)$ is  the maximal arity $|\ov{y}|$ of a
term of the form $\Count{\ov{y}}{\psi}$ occurring in $\xi$; if $\xi$
contains no such term, then $\bw(\xi)=1$ if $\xi$ contains an
existential quantifier $\exists y$ with $y\in\VARS$, and $\bw(\xi)=0$
otherwise. Note that quantification over number variables does not
influence the binding rank or the binding width and, conversely,
quantification over structure variables does not influence the number
quantifier rank. 
Precisely, the notions are defined as follows.
\begin{enumerate}[(1)]
\item
  $\nqr(\phi)=\br(\phi)=\bw(\phi)= 0$, if $\phi$ is of the form $x_{1}{=}x_2$ or
  $R(x_1,\ldots,x_{\ar(R)})$
\item
  for each $f\in\set{\nqr,\br,\bw}$ we let
  $f(\nicht\phi)=f(\phi)$
  and 
  $f((\phi\oder\psi))=\max\set{f(\phi),f(\psi)}$
\item
  for all structure variables $y\in\VARS$ we let 
  $\nqr(\exists y\,\phi) = \nqr(\phi)$,
  $\br(\exists y\,\phi) = \br(\phi)+1$ \ and \ 
  $\bw(\exists y\,\phi) = \max\set{\bw(\phi),1}$
\item
  for each $f\in\set{\nqr,\br,\bw}$ we let
  $f(\P(t_1,\ldots,t_m)) = \max\set{f(t_1),\ldots,f(t_m)}$,
\item
  for all tuples $\ov{y}$ of structure variables we let \
  $\nqr(\Count{\ov{y}}{\phi}) = \nqr(\phi)$,
  $\br(\Count{\ov{y}}{\phi}) = \br(\phi)+1$ \ and \ 
  $\bw(\Count{\ov{y}}{\phi}) = \max\set{|\ov{y}|,\bw(\phi)}$
\item
  $\nqr(i)=\br(i)=\bw(i)=0$  for $i\in\bZ$
\item
  for all $f\in\set{\nqr,\br,\bw}$ we let
  $f((t_1+t_2)) = f((t_1\cdot t_2)) =\max\set{f(t_1),f(t_2)}$
\item $\nqr(\kappa)=\br(\kappa)=\bw(\kappa)=0$ for $\kappa\in\NVARS$
\item 
  for all number variables $\kappa\in\NVARS$ we let
  $\nqr(\exists \kappa\,\phi) = \nqr(\phi)+1$,
  $\br(\exists \kappa\,\phi) = \br(\phi)$ \ and \
  $\bw(\exists \kappa\,\phi) = \bw(\phi)$.
\end{enumerate}

\begin{example}
  The sentence
  $\exists x\,\quant{Prime} \big(\Count{(y)}{E(x,y)\big)}$ has number
  quantifier rank $0$, binding rank $2$, binding width $1$, and size
  $16$.  When evaluated in a directed graph $\A=(A,E^{\A})$, the
  sentence states that $\A$ contains a node whose out-degree is a
  prime number.
\end{example}

\section{Hanf Normal Form}\label{section:GeneralsedHNF}

\subsubsection*{Gaifman graph and bounded structures}

Let $\sigma$ be a signature.  The \emph{Gaifman graph} $G_\cA$ of a
$\sigma$-structure $\A$ is the undirected graph with vertex set $A$
and an edge between two distinct vertices $a,b\in A$ iff there exists
$R\in\sigma$ and a tuple $(a_1,\ldots,a_{\ar(R)})\in R^{\cA}$ such
that $a,b\in\set{a_1,\ldots,a_{\ar(R)}}$.  The structure $\A$ is
called \emph{connected} if its Gaifman graph $G_{\A}$ is connected;
the \emph{connected components} of $\A$ are the connected components
of $G_{\A}$.  The \emph{degree} of $\A$ is the degree of its Gaifman
graph, i.e., the maximum number of neighbours of a node of $G_\A$.
For $d\in \NN$, a $\sigma$-structure $\cA$ is \emph{$d$-bounded} if
its degree is at most $d$.

Two formulas or two counting terms
$\xi$ and $\xi'$ over a signature $\sigma$ are
\emph{$d$-equivalent} (for short, $\xi\equiv_d\xi'$), if
$\sem{\xi}^{\I} = \sem{\xi'}^{\I}$ for every
$\sigma$-interpretation $\I=(\cA,\beta)$ with $\cA$ 
\emph{$d$-bounded}.

Let $\cA$ be some $\sigma$-structure, $\ov{a}\in A^n$ for some
$n\ge1$, and $b\in A$.  The \emph{distance} $\dist^\cA(\ov{a},b)$
between $\ov{a}$ and $b$ is the minimal number of edges of a path from
some element of the tuple $\ov{a}$ to $b$ in $G_\cA$ (if no such path
exists, we let $\dist^\cA(\ov{a},b)=\infty$).  For every $r \ge 0$,
the \emph{$r$-neighbourhood of $\ov{a}$ in $\cA$} is the set
\ $ N_r^\cA(\ov{a}) \ = \ \setc{b\in A\,}{\,\dist^\cA(\ov{a},b)\le r}$.

\subsubsection*{Types, spheres, and sphere-formulas}

Let $\sigma$ be a relational signature and let $c_1,c_2,\ldots$ be a
sequence of pairwise distinct constant symbols. For every $r\ge 0$ and
$n\ge 1$, a \emph{type with $n$ centres and radius~$r$} (for short:
\emph{$r$-type with $n$ centres}) is a structure
$\tau = (\cA, a_1,\ldots,a_n)$ over the signature
$\sigma \cup \set{c_1,\ldots,c_n}$, where $\cA$ is a $\sigma$-structure
and $(a_1,\ldots,a_n)\in A^n$ with $A=N_r^\cA(a_1,\ldots,a_n)$, i.e.,
each element of $\cA$ is ``close'' to some element from
$\set{a_1,\dots,a_n}$.  The elements $a_1,\ldots,a_n$ are the
\emph{centres} of~$\tau$.

Let $\cA$ be a $\sigma$-structure. For every non-empty set
$B\subseteq A$, we write $\inducedSubStr{\cA}{B}$ to denote the
restriction of the structure $\cA$ to the universe $B\subseteq A$,
i.e., the $\sigma$-structure with universe $B$, where
$R^{\inducedSubStr{\cA}{B}} = R^\cA \cap B^{\ar(R)}$ for each 
symbol $R\in\sigma$.

For each tuple $\ov{a} = (a_1,\ldots,a_n) \in A^n$, the
\emph{$r$-sphere of $\ov{a}$ in $\cA$} is defined as the $r$-type with
$n$ centres
\begin{gather*}
  \Neighb{r}{\cA}{\ov{a}} \ = \ \big(\inducedSubStr{\cA}{\neighb{r}{\cA}{\ov{a}}},\, \ov{a}\big)
\end{gather*}
over the signature $\sigma \cup \set{c_1,\ldots,c_n}$. 
We say that \emph{$\ov{a}$ is of (or, realises the) type $\tau$ in $\cA$} iff
$\Neighb{r}{\cA}{\ov{a}} \cong \tau$.

For any $d$-bounded structure $\cA$, any node $a\in A$, and any
$r\in\bN$, we have
\[
   |N_r^{\cA}(a)| \quad\le \quad \nu_d(r) \quad:= \quad 1 \ + \ d\cdot\sum_{0\le i<r}(d-1)^i\,.
\]
Observe that for all $r\ge 0$ we have $\nu_0(r) = 1$,
$\nu_1(r) \le 2$, $\nu_2(r) = 2r{+}1$, and
$(d-1)^r \le \nu_d(r) \le d^{r+1}$ for $d\geq 3$, i.e.,
$\nu_d$ grows linearly for $d\le 2$ and exponentially for $d\geq 3$.

For every $d,r\ge 0$ and $n\ge 1$, the universe of every
$d$-bounded $r$-type $\tau$ with $n$ centres contains at most
$n\cdot\nu_d(r)$ elements.  Thus, given $\tau$ and $r$, one can construct a
\emph{sphere-formula $\sph_{\tau}(\ov{x})$} (depending on $\tau$ and $r$), i.e., an
$\FO[\sigma]$-formula such that for every $\sigma$-structure $\cA$ and
every tuple $\ov{a}\in A^n$ we have
\[
  \cA \models \sph_{\tau}[\ov{a}] 
  \ \ \iff \ \
  \Neighb{r}{\cA}{\ov{a}} \ \isom \ \tau\,.
\]
The formula $\sph_{\tau}(\ov{x})$ can be constructed in time
$\bigO(\size{\sigma})$ if $n\cdot\nu_d(r) = 1$, and otherwise in time
$(n \cdot \nu_d(r))^{\bigO(\size{\sigma})}$.

\subsection{Formulas in Hanf normal form for $\FO(\Ps)$}

In this subsection, we fix a relational signature $\sigma$ and a
\emph{unary} numerical predicate collection, i.e., a numerical
predicate collection $(\Ps,\ar,\sem{.})$ with $\ar(\P)=1$ for
all $\P\in\Ps$.  We recall the notion of formulas in Hanf normal form
for the logic $\FO(\Ps)$ from \cite{HeiKS16} (it extends the classical
notion for first-order logic $\FO$, see, e.g., \cite{BolK12}).

A \emph{numerical condition on occurrences of types with one
  centre} (or \emph{numerical oc-type condition})
for $\FO(\Ps)[\sigma]$ is a sentence of the
form
\begin{equation*}
  \P\big(\Count{(y)}{\sph_{\tau}(y)}-k\big)\, ,  
\end{equation*}
where $\P\in\Ps \cup \set{\P_{\exists}}$, $k\in\bN$, and $\tau$ is an
$r$-type with one centre, for some $r\in\NN$ (in \cite{HeiKS16}, such
sentences are called \emph{Hanf-sentences}).  We call $r$ the
\emph{locality radius} of the numerical oc-type condition. The
condition expresses that the number of interpretations for $y$ such
that the $r$-sphere around $y$ is isomorphic to $\tau$ belongs to the
set $\sem{\P}+k$.

A formula $\varphi(\ov{x})$ is in \emph{Hanf normal form for
  $\FO(\Ps)[\sigma]$} if it is a Boolean combination of numerical
oc-type conditions for $\FO(\Ps)[\sigma]$ and sphere-formulas from
$\FO[\sigma]$; in particular, this means that
$\phi\in\FO(\Ps\cup\set{\P_{\exists}})[\sigma]$. Accordingly, a
\emph{sentence} is in Hanf normal form if it is a Boolean combination
of numerical oc-type conditions.  
We will speak of \emph{hnf-formulas}
(for $\FO(\Ps)[\sigma]$) when we mean
``formulas in Hanf normal form'' 
(for $\FO(\Ps)[\sigma]$), and similarly for \emph{hnf-sentences}.
The \emph{locality radius} of
an hnf-formula is the maximum of the locality radii of its numerical
oc-type conditions and its sphere-formulas.

The following theorem summarises the main results of \cite{HeiKS16} and was the
starting point of the work to be reported in the present paper.

\begin{theorem}[\cite{HeiKS16,BolK12}]\ 
\label{thm:main-HeiKS16}
  \begin{enumerate}[(a)]
  \item\label{item1:thm:main-HeiKS16} Let $(\Ps,\ar,\sem{.})$ be a
    numerical predicate collection with $\ar(\P)=1$ for all
    $\P\in\Ps$.
    The following are equivalent:
    \begin{enumerate}[$\bullet$]
    \item For any relational signature $\sigma$, any degree bound
      $d\in\bN$, and any formula $\varphi\in\FO(\Ps)[\sigma]$, there
      exists a $d$-equivalent hnf-formula for $\FO(\Ps)[\sigma]$.
    \item 
      For all $\P\in\Ps$, the set $\sem{\P}$ is ultimately periodic.
    \end{enumerate}
  \item\label{item2:thm:main-HeiKS16} Let $(\Us,\ar,\sem{.})$ be the numerical predicate collection
    from Example~\ref{E-list}\eqref{E-ultimately-periodic}.
    There is an algorithm which receives as input a degree bound
    $d\geq 2$, a relational signature $\sigma$, and a formula
    $\varphi\in\FO(\Us)[\sigma]$, and constructs a $d$-equivalent
    hnf-formula $\psi$ for $\FO(\Us)[\sigma]$.
    The algorithm's running time is in
    \[
      \exp_3\bigl({\bigO(\size{\varphi}+\size{\sigma})+\log\log(d)}\bigr)\,.
    \]
  \item\label{item3:thm:main-HeiKS16} There exists a relational signature~$\sigma$ and a sequence of
    $\FO[\sigma]$-sentences $\varphi_n$ of size $\bigO(n)$ such that
    every 3-equivalent hnf-sentence
    $\psi_n\in\FO(\set{\P_{\exists}})[\sigma]$ has at least
    $\exp_3(n)$ 
    subformulas.
  \end{enumerate}
\end{theorem}

Claim~\eqref{item1:thm:main-HeiKS16} above implies in particular the existence of
$d$-equivalent hnf-formulas for first-order logic
(cf.~e.g.~\cite{BolK12}). For $\Ps=\set{\quant{D}_p}$
(cf.~Example~\ref{E-list}\eqref{E-divisibility}), the existence of
$d$-equivalent hnf-formulas for $\FO(\Ps)$ also follows from
Nurmonen's work~\cite{Nur00}, and claim~\eqref{item2:thm:main-HeiKS16} provides an
algorithmic version of Nurmonen's theorem.  Claim~\eqref{item2:thm:main-HeiKS16} also
implies the main result from~\cite{BolK12}.  Finally, claim~\eqref{item3:thm:main-HeiKS16}
was already shown in \cite{BolK12}.

\subsection{Hanf normal form for $\FOCN(\Ps)$}
\label{section:main-results}
\label{sec:mainresults}

To also allow some kind of ``Hanf normal form'' for numerical
predicates that are not ultimately periodic, we introduce the notion
of a formula in ``Hanf normal form'', where the ``numerical
oc-type conditions for $\FO(\Ps)$'' are replaced by more general
``numerical oc-type conditions for $\FOCN(\Ps)$'' and, in addition to
Boolean combinations, we also allow quantification over number
variables (but not over structure variables). Recall that the
numerical oc-type condition for $\FO(\Ps)$ is of the form
$\P(\Count{(y)}{\sph_\tau}(y)-k)$ and expresses that the number of
realisations of the type $\tau$ with a single centre, decremented by
$k$, belongs to the set $\sem{\P}$. In numerical oc-type conditions for
$\FOCN(\Ps)$, the ``difference between the number of realisations of
$\tau$ and $k$'' is replaced by an arbitrary multivariate integer polynomial
whose variables are the number of realisations of one-centred
types $\tau_1,\dots,\tau_n$ and number variables from $\NVARS$. In
addition, we give up the restriction to unary numerical predicate
collections. The precise definition is as follows.

A \emph{basic counting term} for $\FOCN(\Ps)[\sigma]$ is a counting term
$t$ of the form
\[
  \Count{(y)}{\sph_\tau(y)}
\]
where $y\in\VARS$, $r\in\NN$ and $\tau$ is an $r$-type with one centre
(over $\sigma$).  The number $r$ is called the \emph{locality radius}
of the basic counting term $t$.  In a $\sigma$-structure $\A$, the
basic counting term $t$ specifies the number $t^\A$ of elements $a\in A$ with
$\Neighb{r}{\A}{a}\isom \tau$.

A \emph{numerical oc-type condition for $\FOCN(\Ps)[\sigma]$} is a
formula that is built from basic counting terms, number variables, and integers,
by addition, multiplication, numerical predicates from
$\Ps\cup\set{\P_{\exists}}$,
Boolean combinations, and quantification of number variables. Its locality
radius is the maximal locality radius of the involved basic counting
terms.
More precisely, the numerical oc-type condition for $\FOCN(\Ps)[\sigma]$ are defined as follows.
\begin{itemize}
\item
The \emph{simple counting terms} for $\FOCN(\Ps)[\sigma]$ are built
from the basic counting terms, number variables from $\NVARS$, and the
rules \eqref{item:constterm} and \eqref{item:plustimesterm} of
Definition~\ref{def:FOC}, i.e., are polynomials over basic counting
terms and number variables with integer coefficients. The locality
radius of a simple counting term $t$ is the maximum of the locality
radii of the basic counting terms that occur in $t$.
\item
An \emph{atomic numerical oc-type condition for $\FOCN(\Ps)[\sigma]$}
is a formula $\chi$ of the form
\[
   \P(t_1,\ldots,t_m),
\]
where $\P\in\Ps\cup\set{\P_{\exists}}$, $m=\ar(\P)$, and $t_j$ is a
simple counting term for each $j\in [m]$. Since simple counting terms
do not have free structure variables from $\VARS$ (but possibly free
number variables from $\NVARS$), the formula $\chi$ is actually a
number formula.  The maximum locality radius of the $t_j$ is called
the \emph{locality radius} of $\chi$.
\item
\emph{Numerical oc-type conditions for $\FOCN(\Ps)[\sigma]$} are built
from atomic numerical oc-type conditions by Boolean combinations and
quantification over number variables, i.e., by applying the rules 
\eqref{item:bool} and \eqref{item:exists:number} of Definition~\ref{def:FO} and Definition~\ref{def:FOCN}.  
The locality radius of a numerical oc-type condition is
the maximal locality radius of the involved atomic numerical oc-type
conditions.
\end{itemize}

A formula $\phi(\ov{x},\ov{\kappa})$ is in
\emph{Hanf normal form for $\FOCN(\Ps)[\sigma]$} if it is a Boolean
combination of numerical oc-type conditions for $\FOCN(\Ps)[\sigma]$ and sphere-formulas from
$\FO[\sigma]$; in particular, this means that 
$\phi\in\FOCN(\Ps\cup\set{\P_{\exists}})[\sigma]$.  
The maximal locality radius of the involved conditions
and formulas is the \emph{locality radius} of the formula in Hanf
normal form.

We abbreviate ``formula in Hanf normal form (for $\FOCN(\Ps)[\sigma]$)''
by \emph{hnf-formula (for \allowbreak $\FOCN(\Ps)[\sigma]$)}. 
Accordingly, a
\emph{hnf-sentence} (for $\FOCN(\Ps)[\sigma]$) is a sentence 
in Hanf normal form, i.e., a Boolean combination of numerical
oc-type conditions without free number variables.

When speaking of the \emph{number of distinct numerical oc-type
  conditions} in an hnf-formula $\psi$ we mean the minimal number $s$
of numerical oc-type conditions $\chi_1,\ldots,\chi_s$ such that
each $\chi_i$ is either an atomic numerical oc-type condition or starts with a number quantifier (i.e., is of the 
form $\exists \kappa\,\chi_i'$ with $\kappa\in\NVARS$), and
$\psi$ is a Boolean combination of $\chi_1,\ldots,\chi_s$ and of
sphere-formulas from $\FO[\sigma]$.

In analogy to the first two statements in
Theorem~\ref{thm:main-HeiKS16}, the following is our main result
regarding the existence and computability of hnf-formulas for $\FOCN(\Ps)$.
\newcounter{saved-theorem-number}
\setcounter{saved-theorem-number}{\arabic{theorem}}
\begin{theorem}\label{thm:main}\ 
  Let $(\Ps,\ar,\sem{.})$ be a numerical predicate collection.
  \begin{enumerate}[(a)]
  \item\label{thm:main:ghnf} For any relational signature $\sigma$,
    any degree bound $d\in\NN$, and any $\FOCN(\Ps)[\sigma]$-formula
    $\phi$, there exists a $d$-equivalent hnf-formula $\psi$ for
    $\FOCN(\Ps)[\sigma]$ of locality radius less than
    $(2\bw(\varphi)+1)^{\br(\varphi)}$.  Moreover,
    $\free(\psi)=\free(\phi)$, $\nqr(\psi)\leq\nqr(\phi)$, and the
    number of distinct numerical oc-type conditions in $\psi$ is at
    most
    \[
      \exp_4\bigl(\poly(\size{\phi}+\size{\sigma})+\log\log(d)\bigr)\,.
    \]
  \item\label{thm:main:algo} There is an algorithm which receives as
    input a degree bound~$d\geq 2$, a relational signature~$\sigma$,
    and an $\FOCN(\Ps)[\sigma]$-formula $\phi$, and constructs such a
 hnf-formula $\psi$ in time
    \[
      \exp_5\bigl(\poly(\size{\varphi}+\size{\sigma})+\log\log(d)\bigr)\,.
    \]
  \end{enumerate}
\end{theorem}
The proofs of these two statements can be found in
Section~\ref{sec:upperbound}. 
Concerning the numerical predicates, our proofs are purely syntactical
and do not rely on the particular semantics $\sem{\P}$ of the
numerical predicates $\P\in\Ps$.

From statement \eqref{thm:main:ghnf}, we
infer in Section~\ref{sec:applications} a polynomial bound for the
locality rank of first-order formulas of bounded quantifier
alternation depth as well as a connection between our logic and
bounded arithmetic; algorithmic applications of statement
\eqref{thm:main:algo} for model checking and query evaluation are also
discussed in Section~\ref{sec:applications}.

Note that Theorem~\ref{thm:main}\eqref{thm:main:ghnf} 
implies that if $\phi$ is in $\FOC(\Ps)$ (i.e., contains no number
variables), then the hnf-formula $\psi$ is in $\FOC(\Ps\cup\set{\P_{\exists}})$. For
$\phi$ in $\FO(\Ps)$, however, Theorem~\ref{thm:main-HeiKS16}\eqref{item1:thm:main-HeiKS16}
ensures that $\psi$ is not always in $\FO(\Ps\cup\set{\P_\exists})$.

Suppose $\ar(\P)=1$ for all $\P\in\Ps$ (i.e., $\Ps$ is \emph{unary})
and let $\varphi\in\FO(\Ps)$.  
Since $\FO(\Ps)\subseteq\FOC(\Ps)$, there is a $d$-equivalent
hnf-formula $\psi$ for $\FOCN(\Ps)$, and we even know that
$\psi\in\FOC(\Ps\cup\set{\P_\exists})$. An analysis of the proof of 
Theorem~\ref{thm:main}\eqref{thm:main:ghnf} yields
that all counting terms that appear in $\psi$ have the form
$i$ with $i\in\bN$ or
\[
  \Big(\sum_{\tau\in T}\Count{(y)}{\sph_\tau(y)}\Big)\ -\ k
\]
for some set $T$ of types of radius~$r$ and some $k\in\bN$. Since all
predicates from $\Ps$ are unary, we can eliminate the constant
counting terms $i\in\bN$ by replacing $\P(i)$ with $\true$ or $\false$
depending on whether $i\in\sem{\P}$ or not.  
Clearly,
\[
   \sum_{\tau\in T}\Count{(y)}{\sph_\tau(y)}
   \quad \equiv \quad
   \Count{(y)}{\bigvee_{\tau\in T}\sph_\tau(y)}
\]
since the types from $T$ all have the same radius and no element can satisfy 
the sphere-formulas for two different types from $T$.  
Hence, $\psi$ can
be transformed into a Boolean combination of sphere-formulas and of
sentences of the form
\begin{equation}
  \label{eq:weak-numerical-oc-type-conditions}
  \P\big(\Count{(y)}{\bigvee_{\tau\in T}\sph_\tau(y)} \ - \ k\big)
\end{equation}
with $\P\in\Ps\cup\set{\P_{\exists}}$.
We call such a Boolean combination a \emph{formula in weak Hanf normal
  form for $\FO(\Ps)$} or \emph{whnf-formula} since it weakens the
condition on numerical oc-type conditions in hnf-formulas for
$\FO(\Ps)$ (that requires $\setsize{T}=1$). As a result, we obtain

\begin{corollary}\label{cor:whnf}
  Let $(\Ps,\ar,\sem{.})$ be a unary numerical predicate
  collection. For any relational signature $\sigma$, any degree bound
  $d\in\bN$, and any formula $\varphi\in\FO(\Ps)[\sigma]$, there exists a
  $d$-equivalent formula $\psi$ in weak Hanf normal form for $\FO(\Ps)[\sigma]$.
\end{corollary}

As can be seen from the proof above, the formula $\psi$ can be
constructed effectively and the bounds from Theorem~\ref{thm:main}
apply here as well; in particular, the number of subformulas of the
form~\eqref{eq:weak-numerical-oc-type-conditions} is 4-fold
exponential in the size of $\varphi$. For this setting, we get a matching
lower bound in analogy to
Theorem~\ref{thm:main-HeiKS16}\eqref{item3:thm:main-HeiKS16};
the proof can be found in Section~\ref{sec:lowerbound}:

\setcounter{theorem}{\arabic{saved-theorem-number}}
\begin{theorem}[continued]\ 
  \begin{enumerate}[(a)]
    \addtocounter{enumi}{2}
  \item\label{thm:main:lowerbound} There exists a unary numerical
    predicate collection $(\Ps,\ar,\sem{.})$, a relational
    signature~$\sigma$, and a sequence of $\FO(\Ps)[\sigma]$-sentences
    $\varphi_n$ of size $\bigO(n)$ such that for every 3-equivalent
    weak hnf-sentence $\psi_n$ for $\FO(\Ps)[\sigma]$, the number of
    distinct subformulas of 
    the form~\eqref{eq:weak-numerical-oc-type-conditions} in $\psi_n$
    is at least $\exp_4(n)$, for every $n\geq 1$.
  \end{enumerate}
\end{theorem}

\section{Construction of  hnf-formulas for $\FOCN(\Ps)$}\label{sec:upperbound}

The following lemma summarises easy facts concerning types (cf., \cite{HeiKS16,BKS_ICDT17}).

\begin{lemma}\label{lem:basic_facts}
  Let $d\geq 2$ and
  let $\A$ be a $d$-bounded $\sigma$-structure.
  Let $r\geq 0$, $k\geq 1$, and $\ov{a}=a_1,\ldots,a_k\in A$.
  \begin{enumerate}[(a)]
  \item\label{eq:Nsize:lem:basic_facts}
    $\Setsize{N_{r}^{\A}(\ov{a})} \ \leq \ k\cdot\nu_d(r) \ \leq \ k d^{r+1}$.
    For $d=2$ we even have $\Setsize{N_{r}^{\A}(\ov{a})} \ \leq \ k(2r{+}1)$. 
  \item\label{eq:Ncompute:lem:basic_facts}
    Given $\A$ and $\ov{a}$, the $r$-sphere $\NrA{\ov{a}}$ can
    be computed in time 
    $\bigl(k \nu_d(r)\bigr)^{\bigoh(\size{\sigma})} \leq \bigl(k
    d^{r+1}\bigr)^{\bigoh(\size{\sigma})}$.
    For $d=2$ it can even be computed in time 
    $\bigl(k (2r{+}1)\bigr)^{\bigoh(\size{\sigma})}$.
  \item\label{eq:Nconn2:lem:basic_facts}
    $\NrA{a_1,a_2}$ is connected if and only if
    $\dist^\A(a_1,a_2)\leq 2r+1$.
  \item\label{eq:Nconnk:lem:basic_facts}
    If $\NrA{\ov{a}}$ is connected, then
    $\nrA{\ov{a}}\subseteq \neighb{r+(k-1)(2r+1)}{\A}{a_i}$,
    for all $i\in[k]$.
  \item\label{eq:Nisom:lem:basic_facts}
    Let $\B$ be a $d$-bounded $\sigma$-structure and let
    $\ov{b}=b_1,\ldots,b_k\in B$. 

    It can be tested in time
    \ $(k\nu_d(r))^{\bigOh(\size{\sigma}+ k\nu_d(r))} \ \leq \ 
       2^{\bigOh(\size{\sigma}k^2 \nu_d(r)^2)} \ \leq \ 
       2^{\bigOh(\size{\sigma}k^2 d^{2r+2})}$ \
    whether 
    $\NrA{\ov{a}} \isomorph \NrB{\ov{b}}$.
    For $d=2$ the test can even be performed in time 
    $2^{\bigOh(\size{\sigma}k^2 r^2)}$.
  \end{enumerate}
\end{lemma}

For our algorithms it will be convenient to work with a fixed list of
representatives of $d$-bounded $r$-types, provided by the following
lemma (see \cite{HeiKS16,BKS_ICDT17} for a proof).

\begin{lemma}\label{lemma:isotypes}
There is an algorithm which upon input of a relational signature $\sigma$, a degree
bound $d\geq 2$, a radius $r\geq 0$, and a number $k\geq 1$,
computes a list $\Typeslistrdk= \tau_1,\ldots,\tau_\ell$ (for a
suitable $\ell\geq 1$) of $d$-bounded
$r$-types with $k$ centres (over $\sigma$), such that 
for every $d$-bounded $r$-type $\tau$ with $k$ centres (over
$\sigma$) there is exactly one $i\in[\ell]$ such that
$\tau\isom\tau_i$.
The algorithm's runtime is \
$2^{(k\nu_d(r))^{\bigOh(\size{\sigma})}}$.
Furthermore, upon input of a $d$-bounded $r$-type $\tau$ with $k$
centres (over $\sigma$), the particular $i\in[\ell]$ with
$\tau\isom\tau_i$ can be computed in time
$2^{(k\nu_d(r))^{\bigOh(\size{\sigma})}}$. 
\end{lemma}

Throughout the remainder of this paper, $\Typeslistrdk$ will always
denote the list provided by Lemma~\ref{lemma:isotypes}.
We will write $\tau\in\Typeslistrdk$ to express that $\tau$ is one of the types
$\tau_1,\ldots,\tau_\ell$ of the list $\Typeslistrdk$.

Our proof of parts \eqref{thm:main:ghnf} and \eqref{thm:main:algo} of
Theorem~\ref{thm:main} proceeds by induction on the construction of
the input formula. A major technical step for the construction is
provided by the following lemma.

\begin{lemma}\label{lem:upperbound:terms}
  Let $\sigma$ be a relational signature, let $r\geq 0$, $n\geq 0$,
  $k\geq 1$, let $(x_1,\ldots,x_n,y_1,\ldots,y_k)$ be a tuple of
  $n{+}k$ pairwise distinct variables in $\VARS$, let
  $\ov{x}=(x_1,\ldots,x_n)$, let $\ov{y}=(y_1,\ldots,y_k)$, let
  $\tau\in\Typeslistrd{n{+}k}$, and let
  \[
     t(\ov{x}) \ \ \deff \ \ \Count{\ov{y}}{\sphere{\tau}{\ov{x},\ov{y}}}.
  \]
  \begin{enumerate}[$\bullet$]
  \item If $n=0$, then there is a simple counting term $\hat{t}$
    without number variables such that \ $ t^{\A} = \hat{t}^{\A}$, \
    for any $d$-bounded $\sigma$-structure $\A$.
  \item If $n\neq 0$, then for every
    $R'\geq R\deff r+k{\cdot}(2r{+}1)$ and every
    $\rho\in\TypeslistRStrichd{n}$, there is a simple counting term
    $\hat{t}_\rho$ without number variables, such that \
    $ t^{\A}[\ov{a}] = \hat{t}_\rho^\A$ \ holds for any $d$-bounded
    $\sigma$-structure $\A$ and any tuple $\ov{a}\in A^n$ of type
    $\rho$ (i.e., any $\ov{a}\in A^n$ with
    $\A\models\sph_{\rho}[\ov{a}]$).
  \end{enumerate}
  Furthermore, $\hat{t}$ and $\hat{t}_\rho$ have locality radius at
  most $\hat{R}\deff r+(k{-}1)(2r{+}1)$. Moreover, there is an
  algorithm which constructs $\hat{t}$ and $\hat{t}_\rho$,
  respectively, within time 
  $\exp_1\bigl(((n+k)\cdot\nu_d(\hat{R}))^{\bigOh(\size{\sigma})}\bigr)$\,.
\end{lemma}
\begin{proof}
The proof relies on a similar analysis of neighbourhood types as the
proof of Lemma~4.7 in \cite{KS_CSLLICS14} and proceeds by an induction on the
number of \emph{components of $\tau$ w.r.t.\ $\ov{x}$}.
These \emph{components} are defined as follows.
Let $\tau=(\T,e_1,\ldots,e_n,f_1,\ldots,f_k)$ 
and let $G=(V,E)$ be the Gaifman graph of $\tau$. 
Decompose $G$ into its connected components $V_1,\ldots,V_s$.  
In case that $n=0$, the tuple $\ov{x}$ is the empty tuple,
the \emph{components of $\tau$ w.r.t.\ $\ov{x}$} are
defined as the connected components $V_1,\ldots,V_s$ of $G$, and we let
$m\deff s$ and $W_j\deff V_j$ for all $j\in [m]$.
In case that $n\neq 0$, we can assume
w.l.o.g.\ that there is an $i\in\set{1,\ldots,s}$ such
that each of the sets $V_1,\ldots,V_{i}$ and none of the sets
$V_{i+1},\ldots,V_s$ contains an element of $\set{e_1,\ldots,e_n}$.
The \emph{components of $\tau$ w.r.t.\ $\ov{x}$} are defined as the sets
$W_1,\ldots,W_{m}$ where $m\deff s{-}i{+}1$, $W_1\deff V_1\cup\cdots\cup V_i$ and
$W_{j}\deff V_{i+j{-}1}$ for every $j\in\set{2,\ldots,m}$.

Let $\ov{e}=(e_1,\ldots,e_n)$ and $\ov{f}=(f_1,\ldots,f_k)$.
For a set $I\subseteq\set{1,\ldots,k}$ and a $k$-tuple
$\ov{a}=(a_1,\ldots,a_k)$ we write $\ov{a}_I$ to denote the tuple of
length $|I|$ obtained from $\ov{a}$ by deleting all components that do
not belong to $I$.

In case that $n\neq 0$,
consider an arbitrary $R'\in\NN$ with $R'\geq R\deff r+k{\cdot}(2r{+}1)$ and an arbitrary 
$\rho\in\TypeslistRStrichd{n}$ and let
$(\S,\ov{a}')=\rho$ and $\ov{a}'=(a'_1,\ldots,a'_n)$.
\medskip

\emph{Case~1:} \ $n\neq 0$ and $m=1$.\\
Then, $f_1,\ldots,f_k\in \neighb{R-r}{\T}{\ov{e}}$ and thus
$T=\nrT{\ov{e},\ov{f}}=\neighb{R}{\T}{\ov{e}}$.
Therefore, for any
$\sigma$-structure $\A$ and any tuple $\ov{a}\in A^n$ of
type $\rho$ (i.e., $\Neighb{R'}{\A}{\ov{a}})\isom(\S,\ov{a}')$), the following is true:
\[
\begin{array}{rcl}
  t^{\A}[\ov{a}]
& =
& \setsize{\setc{\ov{b}\in A^k\ }{\ \A\models\sph_{\tau}[\ov{a},\ov{b}]}}
\medskip
\\
&  =
& \setsize{\setc{\ov{b}\in A^k\ }{\ \Neighb{r}{\A}{\ov{a},\ov{b}}\isom
    (\T,\ov{e},\ov{f})}}
\medskip
\\
&  =
& \setsize{\setc{\ov{b}\in (\neighb{R-r}{\A}{\ov{a}})^k\ }
  {\ \Neighb{r}{\A}{\ov{a},\ov{b}}\isom (\T,\ov{e},\ov{f})}}
\medskip
\\
& =
& \setsize{\setc{\ov{b}{}'\in (\neighb{R-r}{\rho}{\ov{a}{}'})^k\ }
  {\ \Neighb{r}{\rho}{\ov{a}{}',\ov{b}{}'}\isom (\T,\ov{e},\ov{f})}}\,.
\end{array}
\]
Thus, we can choose
\[
 \hat{t}_\rho
 \ \ \deff \ \ 
 i_\rho
 \ \ \deff \ \ 
 \setsize{\setc{\ov{b}{}'\in (\neighb{R-r}{\rho}{\ov{a}{}'})^k\ }
  {\ \Neighb{r}{\rho}{\ov{a}{}',\ov{b}{}'}\isom (\T,\ov{e},\ov{f})}}
 \ \ \in \ \ 
 \NN
\]
and we are done.
\medskip

\emph{Case~2:} \ $n=0$ and $m=1$.
\\ 
Then $\tau=(\T,f_1,\ldots,f_k)$ and $f_2,\ldots,f_k\in
\neighb{\hat{R}-r}{\T}{f_1}$ for $\hat{R}\deff r+(k{-}1)(2r{+}1)$.
Thus, $T=\neighb{r}{\T}{f_1,\ldots,f_k}=\neighb{\hat{R}}{\T}{f_1}$.
Therefore, for any $d$-bounded $\sigma$-structure $\A$, the following is true:
\begin{equation}\label{eq:lem:upperbound:terms:Case2}
\begin{array}{rcl}
  t^{\A}
& =
& \setsize{\setc{\ov{b}\in A^k\ }{\ \A\models\sph_{\tau}[\ov{b}]}}
\medskip
\\
& =
& \setsize{\setc{\ov{b}\in A^k\ }{\ \Neighb{r}{\A}{\ov{b}}\isom
    \tau}}
\medskip
\\
& =
& \setsize{\setc{(b_1,\ldots,b_k)\ }{\ b_1\in A, \ (b_2,\ldots,b_k)\in
    (\neighb{\hat{R}-r}{\A}{b_1})^{k-1},\ 
  {\Neighb{r}{\A}{b_1,\ldots,b_k}\isom \tau}}}
\medskip
\\
& =
& \displaystyle
  \sum_{\rho\in J}\ i_\rho \cdot \setsize{\setc{b_1\in A\ }{\ \Neighb{\hat{R}}{\A}{b_1}\isom\rho}}\,,
\end{array}
\end{equation}
where $J\deff\TypeslistRhatd{1}$ and,
for each $\rho=(\S,f_1)\in J$,
\[
  i_\rho 
  \ \ \deff \ \ 
  \setsize{
   \setc{(f_2,\ldots,f_k)\in
     (\neighb{\hat{R}-r}{\rho}{f_1})^{k-1}\ }{\ \Neighb{r}{\rho}{f_1,\ldots,f_k}\isom \tau}
  }
  \ \ \in \ \
  \NN\,.
\]
Note that for every $\rho\in J$ and for every $\sigma$-structure $\A$
we have
\[
 \setsize{\setc{b_1\in A\ }{\ \Neighb{\hat{R}}{\A}{b_1}\isom\rho}} 
 \ \ = \ \
 t_\rho^{\A},
\]
for 
\[
 t_\rho 
 \ \ \deff \ \ 
 \Count{(y_1)}{\sph_{\rho}(y_1)}\,.
\]
Using equation~\eqref{eq:lem:upperbound:terms:Case2}, we choose
\[
 \hat{t}
 \ \ \deff \ \ 
 \sum_{\rho\in J}\ 
  (\,i_\rho \cdot t_\rho\,)
\]
and we are done.
\medskip

\emph{Case~3:} \ $n\neq 0$ and $m\geq 2$.
\\
Let
\[
\begin{array}{rclcl}
 I_1 & \deff & \setc{i\in [k]\ }{\ f_i\in W_m},
\smallskip\\
 I_2 & \deff & [k]\setminus I_1 & = & \setc{i\in [k]\ }{\ f_i\in W_1\cup\cdots\cup W_{m-1}}\,.
\end{array}
\]
We consider the neighbourhood types
\[
  \tau_1
  \ \ \deff \ \ 
  \big(\inducedSubStr{\T}{W_m},\ov{f}_{I_1}\big)
\]
and
\[
  \tau_2
  \ \ \deff \ \ 
  \big(\inducedSubStr{\T}{W_1\cup\cdots\cup W_{m-1}},\ov{e},\ov{f}_{I_2}\big)\,.
\]
Clearly, $\tau=(\T,\ov{e},\ov{f})$ is the disjoint union of $\tau_1$
and $\tau_2$.
Furthermore, for any 
$\sigma$-structure $\A$ and any tuple
$\ov{a}=(a_1,\ldots,a_n)\in A^n$
the following is true:
\begin{equation}\label{eq:lem:upperbound:terms:Case3}
\begin{array}{rcl}
  t^{\A}[\ov{a}]
& = 
& \setsize{\setc{\ov{b}\in A^k\ }{\ \Neighb{r}{\A}{\ov{a},\ov{b}}\isom \tau}}
\medskip
\\
& =
& \setsize{\setc{\ov{b}\in A^k\ }{\ \Neighb{r}{\A}{\ov{a},\ov{b}}\isom
    \tau
    \text{ and }
    \Neighb{r}{\A}{\ov{b}_{I_1}}\isom\tau_1
    \text{ and }
    \Neighb{r}{\A}{\ov{a},\ov{b}_{I_2}}\isom\tau_2}
  }
\medskip
\\
& =
& i_1^\A \cdot i_2^\A[\ov{a}] \ - \ i_3^\A[\ov{a}]
\end{array}
\end{equation}
where
\[
\begin{array}{rcl}
  i^{\A}_1 
& \deff 
& \setsize{\setc{\ov{b}{}'\in A^{|I_1|}\ }{\ \Neighb{r}{\A}{\ov{b}{}'}\isom \tau_1}}
\medskip
\\
  i_2^{\A}[\ov{a}] 
& \deff 
& \setsize{\setc{\ov{b}{}''\in A^{|I_2|}\ }{\ \Neighb{r}{\A}{\ov{a},\ov{b}{}''}\isom \tau_2}}
\medskip
\\
  i_3^{\A}[\ov{a}] 
& \deff 
& \setsize{
   \setc{\ov{b}\in A^k\ }{\ 
    \Neighb{r}{\A}{\ov{b}_{I_1}}\isom \tau_1
    \text{ and }
    \Neighb{r}{\A}{\ov{a},\ov{b}_{I_2}}\isom \tau_2
    \text{ and }
    \Neighb{r}{\A}{\ov{a},\ov{b}}\not\isom \tau
   }
 }.
\end{array}
\]

Obviously,
\ $i_1^{\A} =  t_1^{\A}$ \ for 
\[
\begin{array}{rcl}
  t_1 
& \deff
& \Count{\ov{y}_{I_1}}{\sph_{\tau_1}(\ov{y}_{I_1})}\,.
\end{array}
\]
Since $\tau_1$ is connected and has no free variable(s), by
\emph{Case~2} we obtain a simple counting term $\hat{t}_1$ such
that $\hat{t}_1^\A=t_1^{\A}$ for all $d$-bounded $\sigma$-structures
$\A$.

Furthermore, if $I_2\neq\emptyset$, then
\ $i_2^{\A}[\ov{a}]=t_2^{\A}[\ov{a}]$ \ for 
\[
\begin{array}{rcl}
  t_2(\ov{x})
& \deff
& \Count{\ov{y}_{I_2}}{\sph_{\tau_2}(\ov{x},\ov{y}_{I_2})}\,.
\end{array}
\]
Since $\tau_2$ has fewer components w.r.t.\ $\ov{x}$ than $\tau$, by induction
hypothesis we obtain a simple counting term $\hat{t}_{2,\rho}$ such
that $\hat{t}_{2,\rho}^{\A}=t_2^{\A}[\ov{a}]$ holds for any
$d$-bounded $\sigma$-structure $\A$ and any tuple $\ov{a}\in A^n$ of
type $\rho$.

In case that $I_2=\emptyset$, we have
\[
   i_2^{\A}[\ov{a}] \ \ = \ \ 
   \left\{
    \begin{array}{lll}
     1 & \ & \text{if \ $\Neighb{r}{\A}{\ov{a}}\isom \tau_2$}
     \medskip\\
     0 & & \text{if \ $\Neighb{r}{\A}{\ov{a}}\not\isom \tau_2$}\,.
    \end{array}
   \right.
\]
Thus, letting
\[
   \hat{t}_{2,\rho} \ \ \deff \ \ 
   \left\{
    \begin{array}{lll}
     1 & \ & \text{if \ $\Neighb{r}{\rho}{\ov{a}{}'}\isom \tau_2$}
     \medskip\\
     0 & & \text{if \ $\Neighb{r}{\rho}{\ov{a}{}'}\not\isom \tau_2$}
    \end{array}
   \right.
\]
we obtain that 
$\hat{t}_{2,\rho}^{\A}=i_2^{\A}[\ov{a}]$, for any
$d$-bounded $\sigma$-structure $\A$ and any tuple $\ov{a}\in A^n$ of
type $\rho$.

Furthermore, letting $J$ be the set of all 
$\tau'=(\T',\ov{e}{}',\ov{f}{}')\in \Typeslistrd{n{+}k}$ such that
\begin{enumerate}[(i)]
\item 
  $\tau'\neq \tau$,
\item
  $\Neighb{r}{\tau'}{\ov{f}{}'_{I_1}}\isom \tau_1$, and
\item
  $\Neighb{r}{\tau'}{\ov{e}{}',\ov{f}{}'_{I_2}}\isom \tau_2$,
\end{enumerate}
we obtain on all \emph{$d$-bounded} $\sigma$-structures $\A$ and for all
$\ov{a}\in A^n$ that
\[
 i_3^{\A}[\ov{a}]
 \ \ = \ \
 \sum_{\tau'\in J}
  \setsize{
   \setc{\ov{b}\in A^k\ }{\ \Neighb{r}{\A}{\ov{a},\ov{b}}\isom \tau'}
  }
 \ \ = \ \
 \sum_{\tau'\in J} t_{\tau'}^{\A}[\ov{a}]\,,
\]
for
\[
  t_{\tau'}(\ov{x})
  \ \ \deff \ \ 
  \Count{\ov{y}}{\sph_{\tau'}(\ov{x},\ov{y})}\,.
\]
Note that each $\tau'\in J$ has fewer components w.r.t.\ $\ov{x}$ than
$\tau$.
Thus, by induction hypothesis we obtain for each $\tau'\in J$ a simple
counting 
term $\hat{t}_{\tau',\rho}$ such
that $\hat{t}_{\tau',\rho}^{\A}=t_{\tau'}^{\A}[\ov{a}]$ holds for any
$d$-bounded $\sigma$-structure $\A$ and any tuple $\ov{a}\in A^n$ of
type $\rho$.
Using equation \eqref{eq:lem:upperbound:terms:Case3}, we are done by
choosing
\[
 \hat{t}_\rho
 \ \ \deff \ \ 
 \hat{t}_1 \cdot \hat{t}_{2,\rho} \ \ + \ \ (-1)\cdot \sum_{\tau'\in J}\hat{t}_{\tau',\rho}\,.
\]

\emph{Case~4:} \ $n=0$ and $m\geq 2$.
\\
The proof can be taken almost verbatim from the proof for \emph{Case~3} by
always letting $\ov{e}$, $\ov{a}$, $\ov{x}$, and $\ov{e}{}'$ be the
empty tuple, omitting the case that $I_2=\emptyset$, and dropping the
type $\rho$ wherever mentioned. 

\medskip

It is straightforward to verify that in each of the four cases, the
lemma's statement concerning the locality radius of 
$\hat{t}$ and $\hat{t}_{\rho}$, respectively, is correct.
Furthermore, the above proof can easily be translated into an algorithm for
constructing $\hat{t}$ and $\hat{t}_\rho$, respectively.
To analyse the algorithm's runtime, let us write
$\mytime(m)$ for the algorithm's runtime for the case where $\tau$
has at most $m$ components w.r.t.\ $\ov{x}$. By using
Lemma~\ref{lem:basic_facts} and \ref{lemma:isotypes} it is 
straightforward (although a bit tedious) to verify the following.
\begin{itemize}

\item
In \emph{Case 1}, $\mytime(1)\leq 2^{\bigOh(\size{\sigma} (n+k)^2 \nu_d(r)^2)}$.

\item
In \emph{Case 2}, $\mytime(1)\leq 
2^{(\nu_d(\hat{R})^{\bigOh(\size{\sigma})})}\cdot 
\nu_d(\hat{R}{-}r)^{k-1} \cdot 2^{\bigOh(\size{\sigma}\cdot k^2\cdot \nu_d(r)^2)}
$, for $\hat{R}\deff r+(k{-}1)(2r{+}1)$.

\item
Thus, in Case~1 and in Case~2 we have
\[
 \mytime(1) 
 \ \ \leq \ \ 
 2^{((n+k)\cdot \nu_d(\hat{R}))^{\bigOh(\size{\sigma})}}\,.
\]

\item
For Case~3 and Case~4 we obtain that
$\mytime(m)\leq 2^{((n+k)\cdot\nu_d(r))^{\bigOh(\size{\sigma})}}\cdot
\mytime(m{-}1)$.

\item
Inductively, we thus obtain that
\[
 \mytime(m)
 \ \ \leq \ \
 2^{(m-1) \cdot ((n+k)\cdot\nu_d(r))^{\bigOh(\size{\sigma})}} \cdot \mytime(1)
\]

\item
Since $m\leq n{+}k$, the algorithm's runtime upon input of $t$ (and
$\rho$, in case that $n\neq 0$) is at most
\[
  2^{((n+k)\cdot\nu_d(\hat{R}))^{\bigOh(\size{\sigma})}}\,.
\]
\end{itemize}
Finally, the proof of Lemma~\ref{lem:upperbound:terms} is complete.
\qed
\end{proof}

\bigskip

We also use the following lemma from \cite{BKS_ICDT17}.

\begin{lemma}\label{lemma:dynamicHanf}
  Let $\sigma$ be a relational signature.  Let $s\geq 0$ and let
  $\chi_1(\ov{\kappa})$, \ldots, $\chi_s(\ov{\kappa})$ be arbitrary
  number formulas from $\FOCN(\Ps)[\sigma]$.\footnote{The lemma's statement in \cite{BKS_ICDT17} was
    formulated for sentences $\chi_1,\ldots,\chi_s$ of first-order
    logic with modulo-counting 
  quantifiers; the proof, however, is independent of the particular
  kind of formulas and also applies for number formulas in $\FOCN(\Ps)$ (or any
  other logic).}
  Let $r\geq 0$, $k\geq 1$, $d\geq
  2$, and let $\Typeslistrdk=\tau_1,\ldots,\tau_\ell$.  
  Let $\ov{x}=x_1,\ldots,x_k$ be a list of $k$ pairwise distinct
  variables in $\VARS$.
  For every Boolean combination $\psi(\ov{x},\ov{\kappa})$ of the
  formulas $\chi_1(\ov{\kappa}),\ldots,\chi_s(\ov{\kappa})$ and
  of $d$-bounded sphere-formulas of radius at most $r$
  (over~$\sigma$), and for every $J\subseteq [s]$ there is a set
  $I\subseteq [\ell]$ such that
  \[
     \psi_J(\ov{x})
     \quad \equivd \quad
     \Oder_{i \in {I}} \sphere{\tau_i}{\ov{x}},
  \]
  where $\psi_J(\ov{x})$ is the formula obtained from
  $\psi(\ov{x},\ov{\kappa})$ by replacing every occurrence of a 
  formula $\chi_j(\ov{\kappa})$ with $\true$ if $j\in J$ and with
  $\false$ if $j\not\in J$ (for every $j\in[s]$).
  \\
  Given $\psi$ and $J$, the set $I$ can be computed in time
  $\poly(\size{\psi})\cdot 2^{(k\nu_d(r))^{\bigOh(\size{\sigma})}}$.
\end{lemma}

By combining the Lemmas~\ref{lemma:dynamicHanf} and
\ref{lem:upperbound:terms}, we immediately obtain:

\begin{lemma}\label{lem:upperbound:terms2}
  Let $\sigma$ be a relational signature.  Let $s\geq 0$ and let
  $\chi_1(\ov{\kappa}), \ldots, \chi_s(\ov{\kappa})$ be arbitrary
  number formulas from $\FOCN(\Ps)[\sigma]$.
  Let $r\geq 0$, $n\geq 0$, $k\geq 1$, and $d\geq 2$.  Let
  $(x_1,\ldots,x_n,y_1,\ldots,y_k)$ be a tuple of $n{+}k$ pairwise
  distinct variables in $\VARS$, let $\ov{x}=(x_1,\ldots,x_n)$, and
  let $\ov{y}=(y_1,\ldots,y_k)$.  Let
  $\psi(\ov{x},\ov{y},\ov{\kappa})$ be a Boolean combination of the
  formulas $\chi_1(\ov{\kappa}),\ldots,\chi_s(\ov{\kappa})$ and
  of $d$-bounded sphere-formulas of radius at most $r$ (over
  $\sigma$), and let
  \[
     t(\ov{x},\ov{\kappa}) \ \ \deff \ \ \Count{\ov{y}}{\psi(\ov{x},\ov{y},\ov{\kappa})}\,.
  \]
  Let $\Typeslistrd{n{+}k}=\tau_1,\ldots,\tau_\ell$.  For every
  $i\in[\ell]$ let
  \[
    t_i(\ov{x}) 
    \ \ \deff \ \ 
    \Count{\ov{y}}{\sph_{\tau_i}(\ov{x},\ov{y})}\,.
  \]
  For every $J\subseteq [s]$ there is a set $I\subseteq [\ell]$ such
  that the following is true for every $d$-bounded $\sigma$-structure
  $\A$ and every tuple
  $\ov{k}\in\ZZ^{|\ov{\kappa}|}$ with
  \[
    (\A,\ov{k}) \ \ \models \ \
    \chi_J(\ov{\kappa}) \ \ \deff \ \ \Und_{j\in J}\chi_j(\ov{\kappa}) \ \und \Und_{j\in
      [s]\setminus J} \nicht\,\chi_j(\ov{\kappa})\,.
  \]
  \begin{itemize}
  \item If $n=0$, then
    \[
      t^\A[\ov{k}] 
      \ \ = \ \
      \sum_{i\in I}\ \hat{t}_i^{\A}
    \]
    where $\hat{t}_i$ is the simple counting term (without number variables) provided by
    Lemma~\ref{lem:upperbound:terms} for the term $t_i$.  We let
    $\hat{t}_{J}\deff\sum_{i\in I}\hat{t}_{i}$.
  \item If $n\neq 0$, then for every
    $R'\geq R\deff r+k{\cdot}(2r{+}1)$, every
    $\rho\in \TypeslistRStrichd{n}$ and every tuple $\ov{a}\in A^{n}$
    of type $\rho$ we have
    \[
      t^\A[\ov{a},\ov{k}]
      \ \ = \ \
      \sum_{i\in I}\ \hat{t}_{i,\rho}^{\A}\,,
    \]
    where $\hat{t}_{i,\rho}$ is the simple counting term (without number variables) provided by
    Lemma~\ref{lem:upperbound:terms} for the term $t_i(\ov{x})$ and
    the type $\rho$.  We let
    $\hat{t}_{J,\rho}\deff\sum_{i\in I}\hat{t}_{i,\rho}$.
  \end{itemize}
  Furthermore, the locality radii of $\hat{t}_J$ and
  $\hat{t}_{J,\rho}$ are at most
  $\hat{R}\deff r+(k{-}1)(2r{+}1)$.  Moreover, there is an algorithm
  which upon input of $\psi(\ov{x},\ov{y},\ov{\kappa})$ and $J$ (and
  $\rho$, in case that $n\neq 0$), constructs $\hat{t}_J$ (resp,
  $\hat{t}_{J,\rho}$) within time \
  $\poly(\size{\psi})\cdot 2^{((n+k)\cdot
    \nu_d(\hat{R}))^{\bigOh(\size{\sigma})}}$.
  \\
  In addition to that, $\setsize{\setc{\hat{t}_J}{J\subseteq [s]}}$
  and $\setsize{\setc{\hat{t}_{J,\rho}}{J\subseteq [s]}}$ is at most
  $2^\ell$ for $\ell \in 2^{((n+k)\nu_d(r))^{\bigOh(\size{\sigma})}}$.
\end{lemma}

We are now ready to prove Theorem~\ref{thm:main}\eqref{thm:main:ghnf}+\eqref{thm:main:algo}:

\begin{theorem}\label{thm:upperbound}
Let $(\Ps,\ar,\sem{.})$ be a numerical predicate collection.
There is an algorithm which upon input of a degree bound $d\geq
2$, a relational signature $\sigma$, and an $\FOCN(\Ps)[\sigma]$-formula
$\phi$, constructs a hnf-formula $\psi$ for $\FOCN(\Ps)[\sigma]$ with 
$\psi\equivd\phi$.

Furthermore, $\free(\psi)=\free(\phi)$, $\nqr(\psi)\le\nqr(\phi)$, and
the locality radius of $\psi$ is 
$<(2\bw(\phi){+}1)^{\br(\phi)}$.  The number of distinct numerical
oc-type conditions in $\psi$ is at most
\begin{align*}
  \exp_3\bigl(\poly(\size{\phi}+\size{\sigma})\bigr) &\text{ \ for \ } d=2
\quad\text{and }\\
  \exp_4\bigl(\poly(\size{\phi}+\size{\sigma})+\log\log(d)\bigr) &\text{ \ for \ }
  d\geq 3.
\end{align*}
The construction of $\psi$ takes time at most 
\begin{align*}
  \exp_4\bigl(\poly(\size{\phi}+\size{\sigma})\bigr) &\text{ \ for \ } d=2
  \quad\text{and}\\
  \exp_5\bigl(\poly(\size{\phi}+\size{\sigma})+\log\log(d)\bigr) &\text{ \ for \ }
  d\geq 3.
\end{align*}
\end{theorem}

\begin{proof}
  W.l.o.g.\ we assume that $\P_{\exists}\in \Ps$ and that $\phi$ does
  not contain any existential quantifier of the form $\exists y$ with
  $y\in\VARS$ (to achieve this, we add $\P_{\exists}$ to $\Ps$ with
  $\sem{\P_{\exists}}=\NNpos$, and we replace every subformula of
  $\phi$ of the form $\exists y\,\phi'$ by the formula
  $\P_{\exists}(\Count{(y)}{\phi'})$).
  We proceed by induction on the shape of $\phi$.  Throughout the
  proof, we let $\ov{x}=(x_1,\ldots,x_n)$ be the free structure
  variables, and $\ov{\kappa}$ be the free number variables of $\phi$.

  \bigskip

  \emph{Case 1:} \ Suppose that $\phi$ is an atomic formula of the
  form $x_1{=}x_2$ or $R(x_1,\ldots,x_{\ar(R)})$ with $R\in \sigma$.
  Clearly, $\phi$ is equivalent to the formula
\[
  \psi \ \ \deff \ \
  \Oder_{\tau\in J} \sph_{\tau}(\ov{x})
\]
where $J$ is the set of all types $\tau\in\Typeslist{0}{d}{n}$ that
satisfy $\phi$.

Furthermore, $\psi$ has locality radius 0, and 
$(2\bw(\phi){+}1)^{\br(\phi)}= 1^0 = 1>0$.
The number of distinct numerical oc-type conditions in $\psi$ is 0, and $\nqr(\psi)=\nqr(\phi)=0$, and $\free(\psi)=\free(\phi)$.
By Lemma~\ref{lemma:isotypes},
$J$ and $\psi$ can be constructed in time 
$2^{(n\nu_d(0))^{\bigOh(\size{\sigma})}} =
2^{n^{\bigOh(\size{\sigma})}} \leq 2^{\size{\phi}^{\bigOh(\size{\sigma})}}$.

\bigskip

\emph{Case 2:} \ Suppose that $\phi$ is of the form $\nicht\phi'$ or
of the form $(\phi'\oder\phi'')$.
By induction hypothesis, there are hnf-formulas $\psi'$ and $\psi''$
with $\psi'\equivd\phi'$ and $\psi''\equivd\phi''$. Thus,
$\nicht\psi'$ and $(\psi'\oder\psi'')$ are hnf-formulas that are
$d$-equivalent to $\nicht\phi'$ and to $(\phi'\oder\phi'')$,
respectively.

Furthermore, by applying the induction hypothesis, it is
straightforward to see that the free variables, the number
quantifier rank, the locality radius, the number of distinct
numerical oc-type conditions, and the runtime for constructing
$\nicht\psi'$ and $(\psi'\oder\psi'')$ are as stated in the theorem.

\bigskip

\emph{Case 3:} \ Suppose that $\phi$ is of the form 
$\P(t_1,\ldots,t_m)$ with $\P\in \Ps\cup\set{\P_{\exists}}$,
$m=\ar(\P)$, and where $t_1,\ldots,t_m$ are counting terms.

According to Definition~\ref{def:FOCN}, for every $j\in [m]$, the
counting term $t_j$ is built by using addition and multiplication
based on integers, on number variables from $\ov{\kappa}$, and on
counting terms $\theta'$ of the form~$\Count{\ov{y}}{\theta}$.  Let
$\Theta'$ be the set of all these counting terms $\theta'$ and let
$\Theta$ be the set of all the according formulas $\theta$.  By the
induction hypothesis, for each $\theta$ in $\Theta$ there is a
$d$-equivalent hnf-formula $\psi^{(\theta)}$. Let $\Psi$ be the set
of all these $\psi^{(\theta)}$.  Each $\psi$ in $\Psi$ is a Boolean
combination of $d$-bounded sphere-formulas and of numerical oc-type
conditions. Let $\chi_1(\ov{\kappa}),\ldots,\chi_s(\ov{\kappa})$ be a
list of numerical oc-type conditions such that any of the
$\psi\in\Psi$ is a Boolean combination of sphere-formulas and of
formulas in $\set{\chi_1,\ldots,\chi_s}$.
Let $r$ be the maximum locality radius of any of the
sphere-formulas that occur in any $\psi\in\Psi$, and let $k$ be the
maximum arity~$|\ov{y}|$ for any term~$\theta'$ of the form
$\Count{\ov{y}}{\theta}$ in $\Theta'$.

For each
$\theta'(\ov{x},\ov{\kappa})=
\Count{\ov{y}}{\theta(\ov{x},\ov{y},\ov{\kappa})}$ in $\Theta'$, we
apply Lemma~\ref{lem:upperbound:terms2} to the term
\[
  t^{(\theta')}(\ov{x},\ov{\kappa}) 
  \ \ \deff \ \ 
  \Count{\ov{y}}{\psi^{(\theta)}(\ov{x},\ov{y},\ov{\kappa})}
\]
and obtain for every $J\subseteq [s]$ 
\begin{itemize}
\item a simple counting term $\hat{t}^{(\theta')}_{J}$ without number
  variables, in case that $n=0$,
\item and for every $\rho\in \Typeslist{R}{d}{n}$, with
  $R\deff r+k(2r{+}1)$, a simple counting term
  $\hat{t}^{(\theta')}_{J,\rho}$ without number variables, in case that
  $n\neq 0$.
\end{itemize}
By Lemma~\ref{lem:upperbound:terms2}, the following is
true for every $J\subseteq [s]$ and $\chi_J := \Und_{j\in J}\chi_j
\und \Und_{j\in [s]\setminus J}\nicht\chi_j$.

If $n=0$, then
\begin{align*}
 & \big( \ \chi_J(\ov{\kappa}) \ \und \ \P(t_1(\ov{\kappa}),\ldots,t_m(\ov{\kappa})) \ \big)\\
  \equivd \ \ &
  \big( \ \chi_J(\ov{\kappa}) \ \und \ \P(t_{1,J}(\ov{\kappa}),\ldots,t_{m,J}(\ov{\kappa})) \ \big)
\end{align*}
where, for every $i\in [m]$, we let $t_{i,J}(\ov{\kappa})$ be the
simple counting term obtained from $t_i(\ov{\kappa})$ by replacing
each occurrence of a term $\theta'\in\Theta'$ by the
term $\hat{t}^{(\theta')}_J$.  

If $n\neq 0$, then for every $\rho\in \Typeslist{R}{d}{n}$ we have
\begin{align*}
  &\big( \ 
  \sph_{\rho}(\ov{x}) \ \und \
  \chi_J(\ov{\kappa}) \ \und \ 
  \P(t_1(\ov{x},\ov{\kappa}),\ldots,t_m(\ov{x},\ov{\kappa})) \ \big)\\
   \equivd \ \ &
  \big( \ 
  \sph_{\rho}(\ov{x}) \ \und \
  \chi_J(\ov{\kappa}) \ \und \ 
  \P(t_{1,J,\rho}(\ov{\kappa}),\ldots,t_{m,J,\rho}(\ov{\kappa})) \ \big)
\end{align*}
where, for every $i\in [m]$, we let $t_{i,J,\rho}(\ov{\kappa})$ be the
simple counting term obtained from $t_i(\ov{x},\ov{\kappa})$ by
replacing each occurrence of a term $\theta'\in\Theta'$
by the term $\hat{t}^{(\theta')}_{J,\rho}$.

In summary, we obtain the following:
\\
If $n=0$, then
\begin{align*}
  \phi(\ov{\kappa}) 
   & \ \ = \ \  \P(t_1,\ldots,t_m)\\
   & \ \  \equivd \ \
   \displaystyle\Oder_{J\subseteq [s]} \big(\
    \chi_J \ \und \ \P(t_1,\ldots,t_m) 
  \ \big)\\
   & \ \ \equivd \ \
  \displaystyle\Oder_{J\subseteq [s]} \big(\
    \chi_J \ \und \ \P(t_{1,J},\ldots,t_{m,J})
  \ \big)
 \quad =: \ \
 \psi(\ov{\kappa})\,.
\end{align*}
The formula $\chi_J$ is a Boolean combination of the numerical oc-type
conditions $\chi_1,\ldots,\chi_s$.
The terms $t_{i,J}$ are polynomials over the simple
counting terms $\hat{t}_J^{(\theta')}$ and number variables from
$\ov{\kappa}$, i.e., they are simple counting terms. Hence
$\psi$ is a Boolean combination of numerical oc-type conditions and therefore a
hnf-formula without free structure variables.

If $n\neq 0$, then for $L\deff\Typeslist{R}{d}{n}$ we have
\[
\begin{array}{llllll}
  \phi(\ov{x},\ov{\kappa}) \ 
& = \quad
& \P(t_1,\ldots,t_m)
&
&
\smallskip
\\
&   \equivd
& \displaystyle\Oder_{\rho\in L}\Big(\ 
   \sph_{\rho}(\ov{x}) \ \und
   \Oder_{J\subseteq [s]} \big(\
    \chi_J \ \und \ \P(t_1,\ldots,t_m) 
  \ \big)\ \Big)
&
&
\\
\smallskip
& \equivd
& \displaystyle\Oder_{\rho\in L}\Big(\
    \sph_{\rho}(\ov{x}) \ \und 
    \Oder_{J\subseteq [s]} \big(\
      \chi_J \ \und \ \P(t_{1,J,\rho},\ldots,t_{m,J,\rho})
  \ \big)\ \Big)
& \quad =: \ 
& \psi(\ov{x},\ov{\kappa})\,.
\end{array}
\]
As above, the formula $\chi_J$ is a Boolean combination of numerical oc-type conditions. The
terms $t_{i,J,\rho}$ are polynomials over the simple counting terms
$\hat{t}_{J,\rho}^{(\theta')}$ and number variables from $\ov{\kappa}$,
i.e., they are simple counting terms. Hence $\psi(\ov{x},\ov{\kappa})$ is a
Boolean combination of sphere-formulas and of numerical oc-type
conditions, and therefore $\psi(\ov{x},\ov{\kappa})$ is a hnf-formula.

By Lemma~\ref{lem:upperbound:terms2}, each of the terms 
$t_{i,J}$ and $t_{i,J,\rho}$, respectively,
has locality radius at most $r+(k{-}1)(2r{+}1)$.
Furthermore, for each $\rho\in L$ the formula $\sph_\rho(\ov{x})$ has 
locality radius $R=r+k(2r{+}1)$.
Thus, the locality radius of $\psi$ is the maximum of $R$ and the
maximum of the locality radii of $\chi_1,\ldots,\chi_s$.

By the induction hypothesis, each $\chi_j$ has locality radius 
$< \tilde{r}\deff (2\bw(\phi){+}1)^{\br(\phi)-1}$. Furthermore, by the induction
hypothesis we also know that 
$r < \tilde{r}$, i.e., $r\leq \tilde{r}{-}1$.
Since $k\leq \bw(\phi)$, we therefore obtain
\[
\begin{array}{rcl}
  R
&  \ \ = \ \ 
&  r + k(2r{+}1)
\\
&  \ \ \leq \ \ 
&  (\tilde{r}{-}1) + 2\bw(\phi) (\tilde{r}{-}1) + \bw(\phi)
\\
&  \ \ < \ \ 
&  \tilde{r} + 2\bw(\phi)\tilde{r}
\\
&  \ \ = \ \ 
&  (2\bw(\phi)+1)^{\br(\phi)}.
\end{array}
\]
Therefore, the locality radius of $\psi$ is $< (2\bw(\phi)+1)^{\br(\phi)}$.

Furthermore, by applying the induction hypothesis it is easy to see that 
$\free(\psi)=\free(\phi)$ and $\nqr(\psi)\leq\nqr(\phi)$.

To determine the number of distinct numerical oc-type conditions in
$\psi$, let us first consider the case $n=0$.
Recall from Lemma~\ref{lem:upperbound:terms2}
that for each $\theta'\in\Theta'$ we have 
\ $
  |\setc{\hat{t}_J^{(\theta')}}{J\subseteq [s]}| 
  \leq 2^\ell 
$ \ with $\ell\in 2^{((n+k)\nu_d(r))^{\bigOh(\size{\sigma})}}$.
Thus, for each $i\in [m]$ we have
\[
  \setsize{\setc{t_{i,J}}{J\subseteq [s]}} 
  \ \ \leq \ \
  \big( 2^\ell \big)^{|\Theta'|}
  \ \ = \ \
  2^{\ell\cdot |\Theta'|}
  \ \ \leq \ \ 
  2^{\ell\cdot \size{\phi}}\,,
\]
and hence
\[
 \setsize{\setc{({t}_{1,J},\ldots,{t}_{m,J})}{J\subseteq [s]}}
 \ \ \leq \ \ 
 \big( 2^{\ell\cdot\size{\phi}}\big)^m
 \ \ = \ \
 2^{m\cdot \size{\phi}\cdot \ell}\,.
\]
Thus, in case that $n=0$, we obtain that $\psi$ is a Boolean
combination of at most $s+2^{m\cdot \size{\phi}\cdot \ell}$
distinct numerical oc-type conditions.
By a similar reasoning we obtain that if $n\neq 0$, then $\psi$ is a
Boolean combination of sphere-formulas and of at most
$s+ |L|\cdot 2^{m\cdot\size{\phi}\cdot \ell}$
distinct numerical oc-type conditions.
Note that
\[
  |L|\cdot 2^{m\cdot\size{\phi}\cdot \ell}
  \ \ \leq \ \
  2^{(n\nu_d(R))^{\bigOh(\size{\sigma})}+ m\cdot\size{\phi}\cdot 2^{((n+k)\nu_d(r))^{\bigOh(\size{\sigma})}}}
  \ \ \leq \ \ 
  2^{2^{(m+n+k)\cdot\size{\phi}\cdot \nu_d(R))^{\bigOh(\size{\sigma})}}}\,.
\]
We already know that $R<(2\bw(\phi)+1)^{\br(\phi)}\leq
\size{\phi}^{\size{\phi}} = 2^{\poly(\size{\phi})}$. Furthermore,
$m+n+k\leq\size{\phi}$.
Thus,
\[
  |L|\cdot 2^{m\cdot\size{\phi}\cdot \ell}
  \ \ \leq \ \
  2^{2^{\size{\phi}^2\cdot \nu_d(2^{\poly(\size{\phi})}))^{\bigOh(\size{\sigma})}}}\,.
\]
In case that $d=2$, this is at most 
\ $2^{2^{2^{\poly(\size{\phi}+\size{\sigma})}}}$. 
In case that $d\geq 3$, it is at most
\ $2^{2^{d^{2^{\poly(\size{\phi}+\size{\sigma})}}}}$.

From the induction hypothesis we obtain a bound on $s$, and in summary
we obtain that $\psi$ is a Boolean combination of sphere-formulas and
of at most 
\[
  \exp_3(\poly(\size{\phi}+\size{\sigma})) 
  \text{ \ for $d=2$ \quad and \quad}
  \exp_4(\poly(\size{\phi}+\size{\sigma})+\log\log d) 
  \text{ \ for $d\geq 3$}
\]
distinct numerical oc-type conditions.

To verify that the claimed runtime is correct, note that by
Lemma~\ref{lem:upperbound:terms2} 
for each $\theta'\in\Theta'$, each $\rho\in L$, and each $J\subseteq
[s]$, the terms
$\hat{t}^{(\theta')}_J$ and $\hat{t}^{(\theta')}_{J,\rho}$, resp., can
be constructed in time
\ $\poly(\size{\psi^{(\theta)}})\cdot 2^{((n+k)\cdot
  \nu_d(\hat{R}))^{\bigOh(\size{\sigma})}}$, where
$\hat{R}\leq R < (2\bw(\phi)+1)^{\br(\phi)}\leq
\size{\phi}^{\size{\phi}} \leq 2^{\poly(\size{\phi})}$.
By the induction hypothesis, $s$ is at most
\[
  \exp_3\bigl(\poly(\size{\phi}+\size{\sigma})\bigr) \text{ \ for \ } d=2
  \qquad
  \text{and}
  \qquad
  \exp_4\bigl(\poly(\size{\phi}+\size{\sigma})+\log\log(d)\bigr) \text{ \ for \ }
  d\geq 3.
\]
Thus, the number of sets $J\subseteq[s]$ that have to be considered is
at most
\[
  \exp_4\bigl(\poly(\size{\phi}+\size{\sigma})\bigr) \text{ \ for \ } d=2
  \qquad
  \text{and}
  \qquad
  \exp_5\bigl(\poly(\size{\phi}+\size{\sigma})+\log\log(d)\bigr) \text{ \ for \ }
  d\geq 3.
\]
Based on this, it is straightforward to verify that the runtime for constructing $\psi$ is as stated in the theorem.

\bigskip

\emph{Case 4:} \ Suppose that $\phi$ is of the form
$\exists\lambda\,\phi'$ with $\lambda\in\NVARS$. By the induction
hypothesis, there is a hnf-formula
$\psi'(\ov{x},\ov{\kappa},\lambda)$ with $\psi'\equiv_d\phi'$. Let $R$
be the locality radius of $\psi'$.

From every $\tau=(\T,\ov{c})\in\Typeslist{R}{d}{n}$, we now construct a
numerical oc-type condition $\psi'_\tau(\ov{\kappa},\lambda)$ as
follows: 
Consider a type $\rho=(\S,\ov{d})$ such that the sphere-formula
$\sph_\rho(\ov{x})$ occurs in $\psi'$, and let $r$ be the locality 
radius of this sphere-formula.
If
\ $
    \Neighb{r}{\T}{\ov{c}}\cong\S\,,
$ \
then we replace every occurrence of the sphere-formula
$\sph_\rho(\ov{x})$ in $\psi'$ by $\true$, otherwise we replace it by
$\false$. As a result, we get
\begin{align*}
  \exists\lambda\,\phi'
  \ \ \equiv_d \ \ 
  \exists\lambda\,\psi' \ \
    &\equiv_d \ \ \exists\lambda\,\bigvee_{\tau\in\Typeslist{R}{d}{n}}
          (\sph_\tau(\ov{x})\land\psi')\\
    &\equiv_d\ \ \exists\lambda\, \bigvee_{\tau\in\Typeslist{R}{d}{n}}
          (\sph_\tau(\ov{x})\land\psi'_\tau(\ov{\kappa},\lambda))\\
    &\equiv \ \ \bigvee_{\tau\in\Typeslist{R}{d}{n}}
          (\sph_\tau(\ov{x})\land\exists\lambda\,\psi'_\tau(\ov{\kappa},\lambda))\\
    &=: \ \ \psi
\end{align*}
which is a hnf-formula.

Clearly, $\free(\psi)=\free(\phi)$, and
$\nqr(\psi)\leq \nqr(\psi'){+}1\leq \nqr(\phi'){+}1=\nqr(\phi)$. 
Since the locality radius of $\psi$ equals that of $\psi'$, it is bounded by
\[
  (2\bw(\varphi')+1)^{\br(\varphi')}\ \ = \ \ (2\bw(\varphi)+1)^{\br(\varphi)}\,.
\]
The number of distinct numerical oc-type conditions in $\psi$ is at most
$\setsize{\Typeslist{R}{d}{n}}$. We get
\begin{align*}
  \setsize{\Typeslist{R}{d}{n}}
    &\ \ \le \ \ \exp_1\left({(nd^{R+1})^{\bigOh(\size{\sigma})}}\right)
     &&\text{by Lemma~\ref{lemma:isotypes}}\\
    & \ \ \le \ \  \exp_1\left({(nd^{(2\bw(\varphi')+1)^{\br(\varphi')}+1})^{\bigOh(\size{\sigma})}}\right)
     &&\text{by the induction hypothesis}\\
    & \ \ \in \ \  \exp_1\left(d^{2^{\poly(\size{\varphi}+\size{\sigma})}}\right)
     &&\text{since }n,\bw(\phi'),\br(\phi')\le\size{\phi}\\
    & \ \ = \ \  \exp_3(\poly(\size{\varphi}+\size{\sigma})+\log\log(d))\\
    & \ \ < \ \  \exp_4(\poly(\size{\varphi}+\size{\sigma})+\log\log(d))\,.
\end{align*}
Furthermore, it is straightforward to verify that the runtime for constructing
$\psi$ is as stated in the theorem.  
\qed
\end{proof}

\section{Applications}
\label{sec:applications}

\subsection{Fixed-parameter model-checking}

As a straightforward application of
Theorem~\ref{thm:main}\eqref{thm:main:ghnf}$+$\eqref{thm:main:algo},
we obtain
that Seese's \cite{See96} $\FO$ model-checking algorithm for classes
of structures of bounded degree can be generalised to the logic
$\FOCN(\Ps)$ for arbitrary numerical predicate collections
$(\Ps,\ar,\sem{.})$:

\begin{theorem}\label{thm:model-checking-formulas}
  Let $(\Ps,\ar,\sem{.})$ be a numerical predicate collection.
  There is an algorithm with oracle
  $\setc{(\P,\ov{n})}{\P\in\Ps,\ov{n}\in\sem{\P}}$ which receives as input
   a formula $\varphi(\ov{x},\ov{\kappa})\in\FOCN(\Ps)$,
   a $\sigma$-structure $\cA$ (where $\sigma$ consists of
   precisely the relation symbols that occur in $\varphi$), a
   tuple $\ov{a} \in A^{|\ov{x}|}$, and a tuple $\ov{k}\in\ZZ^{|\ov{\kappa}|}$,
  and decides whether $\cA \models \varphi[\ov{a},\ov{k}]$. 

  If $d\ge 2$ is an upper bound on the degree of $\cA$, then
  the algorithm runs in time
  \begin{equation*}
    f(\phi,d) \ \ + \ \ g(\phi,d) \cdot |A| \ \ + \ \ f(\phi,d)\cdot |A|^{\nqr(\phi)}
  \end{equation*}
  where $f(\phi,d)\in\exp_5\bigl(\poly(\size{\varphi})+\log\log(d)\bigr)$ and
  $g(\phi,d)\in\exp_3\bigl(\poly(\size{\varphi})+\log\log(d)\bigr)$.
\end{theorem}
\begin{proof}
Let $\phi(\ov{x},\ov{\kappa})$, $\A$, $\ov{a}$, and $\ov{k}$ be the algorithm's input, where
$\sigma$ is the relational signature that consists of precisely the
relation symbols occurring in $\phi$, 
$\A$ is a $\sigma$-structure, and $\Ps$ is the set of all numerical
predicates that occur in $\phi$. 
For checking whether $\A\models\phi[\ov{a},\ov{k}]$, the algorithm proceeds as follows:
\begin{enumerate}[(1)]
 \item
   Compute an upper bound $d\geq 2$ on the degree of $\A$. 

   This can be done in time $\poly(\size{\A}\cdot\size{\sigma}\cdot d)$.
 \item
   Use the algorithm from 
   Theorem~\ref{thm:main}\eqref{thm:main:algo}
   to transform $\phi(\ov{x},\ov{\kappa})$ into a $d$-equivalent  
   $\FOCN(\Ps\cup\set{\P_\exists})[\sigma]$-formula $\psi(\ov{x},\ov{\kappa})$ in Hanf normal form.

   By Theorem~\ref{thm:main}\eqref{thm:main:ghnf}$+$\eqref{thm:main:algo},
   this takes time at most $f(\phi,d)\in
   \exp_5(\poly(\size{\phi})+\log\log d)$, 
   and $\psi$ has locality radius at most
   $r<(2\bw(\phi){+}1)^{\br(\phi)} \leq 2^{\poly(\size{\phi})}$, and
   $\nqr(\psi)\leq\nqr(\phi)$. 

   Note that $\psi$ is a Boolean combination of sphere-formulas of the
   form $\sph_\rho(\ov{x})$ and numerical oc-type conditions 
   with free variables among $\ov{\kappa}$.

 \item 
   For each sphere-formula $\sph_\rho(\ov{x})$ that occurs in $\psi$, check if $\A\models\sph_\rho[\ov{a}]$, and replace each occurrence
   of $\sph_\rho(\ov{x})$ in $\psi$ with $\True$ if $\A\models\sph_\rho[\ov{a}]$, and with $\False$ otherwise.

   By Lemma~\ref{lem:basic_facts}, each such check takes time at most $2^{\bigOh(\size{\sigma} |\ov{x}|^2 d^{2r+2})}$, and this is in
   $\exp_3(\poly(\size{\phi})+\log\log d)$.

 \item 
   For each basic counting term $\Count{(y)}{\sph_\tau(y)}$ that occurs in $\psi$, compute the number
   $n_\tau$ of elements $b\in A$ with $\A\models\sph_\tau[b]$.

   By using the Lemmas~\ref{lem:basic_facts} and \ref{lemma:isotypes}, the numbers $n_\tau$ for all relevant $\tau$ can be computed 
   in time $|A|\cdot 2^{(d^{r+1})^{\bigOh(\size{\sigma})}}$, and this is in $|A|\cdot \exp_3(\poly(\size{\phi})+\log\log d)$. 
   Furthermore, by Lemma~\ref{lemma:isotypes} the number of relevant $\tau$ is in
   $2^{(d^{r+1})^{\bigOh(\size{\sigma})}}$, and thus in $\exp_3(\poly(\size{\phi})+\log\log d)$.

 \item
   Replace each occurrence of a basic counting term
   $\Count{(y)}{\sph_\tau(y)}$ in $\psi$ with the number $n_\tau$.

   Furthermore, replace each free occurrence of a number variable $\kappa_i$ with the number $k_i$ (where
   $\ov{\kappa}=(\kappa_1,\ldots,\kappa_j)$ and $\ov{k}=(k_1,\ldots,k_j)$).

   Note that the resulting formula can be viewed as a first-order sentence $\chi$ that has to be evaluated in 
   $\ZZ$ with addition, multiplication, and the predicates in $\Ps\cup\set{\P_\exists}$, and where quantifications are 
   relativised to numbers in $\set{0,\ldots,|A|}$. 
   By construction, this sentence evaluates to $\true$ if, and only if, $\A\models\psi[\ov{a},\ov{k}]$.

   When using oracles for evaluating the predicates in $\Ps$,
   the evaluation of $\chi$ in $\ZZ$ can be carried out in time $\size{\psi}\cdot \bigOh(|A|^{\nqr(\psi)})$.
\end{enumerate}
In summary, this yields an algorithm that runs in time
\[
   f(\phi,d) \ + \ g(\phi,d)\cdot |A| \ + \ f(\phi,d)\cdot |A|^{\nqr(\psi)}\,,
\]
with $f(\phi,d)\in \exp_5(\poly(\size{\phi})+\log\log d)$ and
$g(\phi,d)\in \exp_3(\poly(\size{\phi})+\log\log d)$.
This completes the proof of Theorem~\ref{thm:model-checking-formulas}.
\qed
\end{proof}

\begin{remark}
Since $\nqr(\phi)=0$ for all $\phi\in\FOC(\Ps)$, 
Theorem~\ref{thm:model-checking-formulas} in particular implies
that on classes of structures of bounded degree,
model-checking of $\FOC(\Ps)$ is fixed-parameter tractable (even
fixed-parameter linear) when using oracles for the predicates in
$\Ps$.
\end{remark}

\subsection{Hanf-locality of $\FOCN(\Ps)$ and the locality rank of $\FO$}
\label{subsec:locality-rank}
\newcommand{\HT}{\mathsf{HT}}
\newcommand{\cZ}{\mathcal Z}

The following notion is taken from \cite{DBLP:journals/jsyml/HellaLN99} (see also the textbook \cite{Lib04}). 
Let $\cA$ and $\cB$ be
structures over a relational signature $\sigma$, let $k\in\bN$ and
$\ov{a}\in A^k$ and $\ov{b}\in B^k$. Let furthermore $r\in\bN$. Then
$(\cA,\ov{a})$ and $(\cB,\ov{b})$ are \emph{$r$-equivalent} (denoted
$(\cA,\ov{a})\leftrightharpoons_r(\cB,\ov{b})$) if there exists a
bijection $f\colon A\to B$ such that for all $c\in A$ we have
\[
   \Neighb{r}{\cA}{\ov{a},c} \ \cong \ \Neighb{r}{\cB}{\ov{b},f(c)}\,.
\]

Now let $\varphi(\ov{x})$ be an $\FOCN(\Ps)$-formula with $k$ free
structure variables and without free number variables.
The formula $\varphi(\ov{x})$ is \emph{Hanf-local} if there exists
$r\ge0$ such that for all structures $\cA$ and $\cB$ and all
$\ov{a}\in A^k$ and $\ov{b}\in B^k$ with
$(\cA,\ov{a})\leftrightharpoons_r(\cB,\ov{b})$, we have
\[
  \cA\models\varphi[\ov{a}]\ \ \iff \ \ \cB\models\varphi[\ov{b}]\,.
\]
The minimal such $r$ is called the \emph{Hanf-locality rank of
  $\varphi$} and is denoted by $\lr(\varphi)$.

Let $\tau$ be a type with a single centre and let $\cA$ be a
$\sigma$-structure. By $\real_\tau^\cA$, we denote the number of
realisations of the type $\tau$ in $\cA$. For $r,d\in\bN$, the
\emph{Hanf-tuple for radius $r$ and degree $d$} for a structure
$\cA$ is the tuple
\[
    \HT_r^d(\cA) \quad = \quad \big(\real_\tau^\cA\big)_{\tau\in\Typeslist{r}{d}{1}}\,.
\]
By Theorem~\ref{thm:main}\eqref{thm:main:ghnf}, every formula
$\varphi$ has a $d$-equivalent hnf-formula $\psi$ of 
locality radius $r< (2\bw(\varphi)+1)^{\br(\varphi)}$.
Furthermore,
$(\cA,\ov{a})\leftrightharpoons_r(\cB,\ov{b})$ implies that
$\HT_r^d(\cA)=\HT_r^d(\cB)$ and
$\Neighb{r}{\cA}{\ov{a}}\cong\Neighb{r}{\cB}{\ov{b}}$. Since the
validity of $\psi$ only depends on this information, we get the following:

\begin{corollary}\label{cor:FOCN:Hanf:local}
  Every $\FOCN(\Ps)$-formula $\varphi(\ov{x})$ without free number
  variables is Hanf-local with Hanf-locality rank
  $\lr(\phi)< (2\bw(\varphi)+1)^{\br(\varphi)}$.
\end{corollary}
\begin{proof}
  Let $r\deff (2\bw(\varphi)+1)^{\br(\varphi)}-1$.
  Let $k=|\ov{x}|$. Let $\cA$ and $\cB$ be two $\sigma$-structures and
  let $\ov{a}\in A^k$ and $\ov{b}\in B^k$ with
  $(\cA,\ov{a})\leftrightharpoons_r(\cB,\ov{b})$. This implies in
  particular that
  $\HT_r^d(\cA)=\HT_r^d(\cB)$ for all $d\in\bN$, and
  (provided that $k>0$) $\Neighb{r}{\cA}{\ov{a}}\cong\Neighb{r}{\cB}{\ov{b}}$.
 
  Since $\cA$ and $\cB$ are finite, there is some $d\in\bN$ such that
  both structures $\cA$ and $\cB$ are $d$-bounded. By
  Theorem~\ref{thm:main}\eqref{thm:main:ghnf}, 
  there exists a hnf-formula $\psi(\ov{x})$ with
  $\varphi\equiv_d\psi$. Since the locality radius of $\psi$ is at most
  $r$, the Hanf-tuple $\HT_r^d(\cA)$ (if $k>0$, together with the
  isomorphism type of the sphere $\Neighb{r}{\cA}{\ov{a}}$) determines
  whether the hnf-formula $\psi$ holds in $(\cA,\ov{a})$ or not (and
  similarly for $(\cB,\ov{b})$). We therefore get
  \begin{align*}
     (\cA,\ov{a})\models\varphi 
       &\ \ \iff \ \ (\cA,\ov{a})\models\psi 
        &&\text{since }\varphi\equiv_d\psi\\
       &\ \ \iff \ \ (\cB,\ov{b})\models\psi
        &&\text{since }\HT_r^d(\cA)=\HT_r^d(\cB)
          \text{ \ (and $\Neighb{r}{\cA}{\ov{a}}\cong\Neighb{r}{\cB}{\ov{b}}$)}\\
       &\ \ \iff \ \ (\cB,\ov{b})\models\varphi
        &&\text{since }\varphi\equiv_d\psi\,.
  \end{align*}
  Hence $\varphi$ is $r$-Hanf-local.\qed
\end{proof}
\bigskip

For first-order formulas $\varphi$ we have
$\lr(\varphi)\in\exp_1(\bigO(\size{\varphi}))$ 
(actually, $\lr(\varphi)\leq 2^{q-1}-1$ where $q$ is the quantifier depth of $\varphi$
\cite{Lib00}).
Our results allow us to bound the Hanf-locality rank of
$\varphi\in\FO$ by a polynomial in $\size{\varphi}$ whose degree is
the quantifier alternation depth of~$\varphi$. As usual, we write
$\Sigma_{n}$ to denote the set of all $\FO$-formulas of quantifier
alternation depth $\leq n$ whose outermost quantifier block is
existential.

\begin{theorem}\label{thm:locality-rank-of-FO}
  Let $\varphi(\ov{x})\in\FO[\sigma]$ belong to
  $\Sigma_{n}$ and let $m\deff|\ov{x}|>0$. Then, $\lr(\varphi)
  < (2\size{\varphi}+1)^n$.
\end{theorem}
\begin{proof}
  The formula $\varphi$ is equivalent to a formula of the form
  $\exists\ov{x_1}\,\lnot\exists\ov{x_2}\cdots\,\lnot\exists\ov{x_n}\,\lnot\psi$
  where $\psi$ is quantifier-free and $\ov{x_1}$, $\ov{x_2}$, \ldots,
  $\ov{x_n}$ are tuples of variables of length $\le\size{\varphi}$.

  By induction, set $\psi_n=\psi$ and
  $\psi_{i-1}=\P_\exists(\Count{\ov{x_i}}{\lnot\psi_i})$. Clearly,
  $\psi_0\equiv\varphi$ has binding rank~$n$ and binding width
  $\max\setc{|\ov{x_i}|}{1\le i\le n}\le\size{\varphi}$.

  Then Corollary~\ref{cor:FOCN:Hanf:local} implies that
  $\lr(\varphi)=\lr(\psi_0) < (2\size{\varphi}+1)^n$.\qed
\end{proof}

\bigskip

In \cite{KusL11} it is conjectured that, for every $n\in\bN$, the
locality rank (also for infinite structures) of formulas
$\varphi\in\Sigma_n$ is polynomial in the quantifier rank $q$ and
therefore in the size of~$\varphi$. The above theorem confirms 
this conjecture at least for finite
structures.\footnote{\cite[Conjecture~6.2]{KusL11} expected the bound
  $\size{\varphi}\cdot 2^n$ as opposed to our result
  $\size{\varphi}^n$.}

\medskip

We close this subsection by proving a result that is slightly stronger than
Theorem~\ref{thm:locality-rank-of-FO}.  To formulate that result
concisely, we need the following definition. Let
$\ell\in\bN_{\ge1}$. Then $\cB\Sigma_{0,\ell}$ is the set of
quantifier-free formulas from $\FO$ (independent from $\ell$). A
formula belongs to $\cB\Sigma_{n+1,\ell}$ if it is a Boolean
combination of formulas of the form
\[
   \exists x_1\,\exists x_2\,\cdots\,\exists x_k\;\psi
\]
with $k\le \ell$ and $\psi\in\cB\Sigma_{n,\ell}$. Note that the
traditional set $\cB\Sigma_n$ of formulas of quantifier alternation
depth~$\le n$ equals the union of all the sets $\cB\Sigma_{n,\ell}$
with $\ell\ge1$. The index $\ell$ bounds the size of blocks of
quantifiers. As an example, consider
$\varphi(\ov{x}),\psi(\ov{y})\in\cB\Sigma_{n,\ell}$ with disjoint sets
of free variables, and note that the formula
\[
   \exists x_1\dots \exists x_\ell\,\exists y_1\dots\exists y_\ell\,(\varphi(\ov{x})\land\psi(\ov{y}))
\]
belongs to $\cB\Sigma_{n+1,2\ell}$ as well as to
$\cB\Sigma_{n+2,\ell}$, but the following is an equivalent formula
from $\cB\Sigma_{n+1,\ell}$:
\[
   \bigl(\ 
    \exists x_1\dots \exists x_\ell\,\varphi(\ov{x})
    \ \land \ 
    \exists y_1\dots\exists y_\ell\,\psi(\ov{y})
   \ \bigr)\,.
\]

The following lemma translates $\cB\Sigma_{n,\ell}$-formulas into
$\FOC({\set{\P_{\exists}}})$-formulas of restricted binding width and
rank:

\begin{lemma}\label{L:FO->FOCQ}
  For all $n\ge0$, $\ell\ge1$, and $\varphi\in\cB\Sigma_{n,\ell}$,
  there exists an equivalent formula $\varphi'\in\FOC(\Ps)$ with
  $\Ps=\set{\P_{\exists}}$ such that $\free(\varphi')=\free(\varphi)$,
  $\bw(\varphi')\le\ell$ and $\br(\varphi')\le n$.
\end{lemma}
\begin{proof}
  We proceed by induction on $n$. If $n=0$, i.e., $\varphi$ is
  quantifier-free, we set $\varphi'=\varphi$.

  Now consider a formula of the form $\varphi=\exists\ov{x}\,\psi$
  with $\psi\in\cB\Sigma_{n,\ell}$ and $|\ov{x}|=k\le \ell$. By
  induction, there exists a formula $\psi'\in\FOC(\Ps)$ with
  $\psi\equiv\psi'$, $\bw(\psi')\le\ell$, $\br(\psi')\le n$, and
  $\free(\psi')=\free(\psi)$. We set
  $\varphi'=\P_{\exists}\,(\Count{\ov{x}}{\psi'})$. Then, clearly,
  $\varphi$ and $\varphi'$ are equivalent,
  $\bw(\varphi') = \bw(\Count{\ov{x}}{\psi'}) =
  \max(|\ov{x}|,\bw(\psi')) \le \ell$,
  $\br(\varphi')= \br(\Count{\ov{x}}{\psi'}) = 1+\br(\psi')\le n+1$,
  and $\free(\varphi')=\free(\varphi)$.

  Since $\cB\Sigma_{n+1,\ell}$-formulas are Boolean combinations
  of formulas of the form $\exists\ov{x}\,\psi$ with
  $\psi\in\cB\Sigma_{n,\ell}$ and $|\ov{x}|=k\le \ell$, the result follows.\qed
\end{proof}

\bigskip

We obtain the following strengthening of Theorem~\ref{thm:locality-rank-of-FO}:
\begin{theorem}
  Let $\varphi(\ov{x})\in\FO[\sigma]$ with $|\ov{x}|>0$ belong to
  $\cB\Sigma_{n,\ell}$. Then the Hanf-locality
  rank of $\varphi$ is less than $(2\ell+1)^n$, i.e.,
  $\lr(\varphi)< (2\ell+1)^n\le\size{\varphi}^n$.
\end{theorem}
\begin{proof}
  By Lemma~\ref{L:FO->FOCQ}, there exists a formula
  $\varphi'\in\FOC(\Ps)$ with $\Ps=\set{\P_{\exists}}$ that is
  equivalent to $\varphi$ such that $\bw(\varphi')\le\ell$,
  $\br(\varphi')\le n$, and
  $\free(\varphi')=\free(\varphi)\subseteq\VARS$. Then we get
  $\lr(\varphi)=\lr(\varphi')< (2\ell+1)^n$ from
  Corollary~\ref{cor:FOCN:Hanf:local}.\qed
\end{proof}

\subsection{Hanf-locality and bounded arithmetic}
\label{subsec:bounded:arithmetic}

For a numerical predicate collection $(\Ps,\ar,\sem{.})$ consider
the extension
\[
  \cZ_\Ps \quad = \quad
  (\bZ,+,\cdot,0,\le,(\sem{\P})_{\P\in\Ps\cup\{\P_{\exists}\}})
\]
of integer arithmetic with the predicates from $\Ps\cup\{\P_{\exists}\}$. A
first-order formula $\Phi(\ov{v})$ in the signature of this structure
is \emph{bounded} if every quantification $\exists v$ is of the form
$\exists v\,(0\le v\le \sum_{1\le i\le |\ov{v}|}v_i\,\land\, \ldots\,)$, i.e.,
quantification is restricted to numbers between 0 and the sum of the
free variables of $\Phi$. 

Let $\varphi\in\FOCN(\Ps)$ be a sentence, let
$r=(2\bw(\varphi)+1)^{\br(\varphi)}-1$, and let $d\in\bN$. By
Theorem~\ref{thm:main}\eqref{thm:main:ghnf}, validity of
$\varphi$ in a $d$-bounded structure $\cA$ only depends on the tuple
$\HT_r^d(\cA)\in\bN^{\Typeslist{r}{d}{1}}$. Since hnf-sentences are
Boolean combinations of numerical oc-type conditions,
Theorem~\ref{thm:main}\eqref{thm:main:ghnf} ensures that $\cA\models\varphi$ is a
first-order property of the tuple $\HT_r^d(\cA)$ in $\cZ_\Ps$. More generally, we
obtain the following:

\begin{theorem}\label{thm:determining-formulas}
  Let $(\Ps,\ar,\sem{.})$ be a numerical predicate collection.
  There is an algorithm which receives as input a degree bound $d$, a
  relational signature $\sigma$, a formula
  $\phi(\ov{x})\in\FOCN(\Ps)[\sigma]$ with $k$ free structure
  variables and without free number variables, and a type
  $\rho\in\Typeslist{r}{d}{k}$, 
  for $r=(2\bw(\varphi)+1)^{\br(\varphi)}-1$,
  and constructs a bounded first-order formula
  $\Psi_\rho$ in the signature
  of~$\cZ_\Ps$ and with free variables
  $(v_\tau)_{\tau\in\Typeslist{r}{d}{1}}$,
  such that the following holds for all $d$-bounded
  $\sigma$-structures $\cA$ and all $\ov{a}\in A^k$ with
  $\rho\cong\Neighb{r}{\cA}{\ov{a}}$:
  \[
     \cA\models\varphi[\ov{a}] \ \ \iff \ \
     \cZ_\Ps\models\Psi_\rho[\HT_r^d(\cA)]\,.
  \]
\end{theorem}
\begin{proof}
  Let $d\in\bN$, let $\sigma$ be a signature, and let
  $\phi\in\FOCN(\Ps)[\sigma]$ have $k$ free structure variables and no free
  number variables. Set $r=(2\bw(\varphi)+1)^{\br(\varphi)}-1$. By
  Theorem~\ref{thm:main}\eqref{thm:main:ghnf}$+$\eqref{thm:main:algo},
  we can construct a $d$-equivalent 
  hnf-formula $\psi$ of locality radius at most $r$. Without loss of
  generality, we can assume that all sphere-formulas $\sph_\tau(\ov{x})$ that appear in $\psi$ 
  have radius exactly $r$ and $\tau\in\Typeslist{r}{d}{k}$.  Let
  $\rho\in\Typeslist{r}{d}{k}$. In $\psi$, replace all occurrences of
  $\sph_\tau(\ov{x})$ with $\tau\neq\rho$ by $\false$ and all
  occurrences of $\sph_\rho(\ov{x})$ by $\true$. The resulting formula
  $\psi_\rho$ is a hnf-sentence of locality radius $\le r$, i.e., a
  Boolean combination of numerical oc-type conditions from
  $\FOCN(\Ps\cup\set{\P_\exists})$. In other words, it is built from counting terms using
  number variables $\kappa$ and basic counting terms
  $\Count{(y)}{\sph_\tau(y)}$ for $\tau\in\Typeslist{r}{d}{1}$ using
  arithmetic operations $+$ and $\cdot$, predicates from
  $\Ps\cup\{\P_{\exists}\}$, Boolean connectives, and number
  quantification~$\exists\kappa$.

  We construct $\Psi_\rho$ from $\psi_\rho$ by replacing every
  occurrence of $\Count{(y)}{\sph_\tau(y)}$ by the variable $v_\tau$, for
  $\tau\in\Typeslist{r}{d}{1}$, and by replacing every subformula of the form $\exists\kappa\,\chi$
  by
  \[
     \exists
     \kappa\,\big(\; 0\le\kappa\le\!\!\!\!\sum_{\tau\in\Typeslist{r}{d}{1}}\!\!\!\!v_\tau
     \ \und\ \chi\;\big)\,.
  \]
  These replacements turn the hnf-sentence $\psi_\rho$ into a bounded
  formula $\Psi_\rho$.

  Now let $\cA$ be some $d$-bounded $\sigma$-structure and let
  $\ov{a}\in A^k$ with $\Neighb{r}{\cA}{\ov{a}}\cong\rho$. For each
  $\tau\in\Typeslist{r}{d}{1}$ let $n_\tau\deff\real_\tau^\cA$ be the
  number of realisations of the type $\tau$ in the
  structure~$\cA$. Then we have
  \[
     (\cA,\ov{a})\models\varphi 
     \quad \xLeftrightarrow{\varphi\equiv_d\psi} \quad
     (\cA,\ov{a})\models\psi 
      \quad \xLeftrightarrow{\Neighb{r}{d}{\ov{a}}\cong\rho} \quad
     \cA\models\psi_\rho 
       \quad \iff \quad
     \cZ_\Ps\models\Psi_\rho[(n_\tau)_{\tau\in\Typeslist{r}{d}{1}}]\,.
  \]\qed
\end{proof}

\bigskip

Note that if the formula $\varphi$ belongs to $\FOC(\Ps)$, i.e., contains no
number variables, then the formula $\Psi_\rho$ is
quantifier-free. With $\Us$ the numerical predicate collection from
Example~\ref{E-list}\eqref{E-ultimately-periodic}, the formula
$\Psi_\rho$ can be rewritten into a formula in the signature of
$(\bZ,+,0,\le)$ (for $\varphi\in\FOC(\Us)$). Furthermore, using
\cite{FagSV95} and in particular \cite{BolK12}, a similar proof for
$\varphi\in\FO[\sigma]$ yields a quantifier-free formula $\Psi_\rho$
in the signature of~$(\bZ,\le)$. Recall that the counting logics
$\FO(\mathrm{Cnt})$ from \cite{Lib04} and $\FO{+}\mathsf{C}$ from
\cite{Gro13} are fragments of $\FOCN(\Ps)$ where $\Ps$ contains only
arithmetical predicates. Consequently, the $d$-bounded models of any
formula from these logics are determined by some set definable in
bounded arithmetic.

\nc{\Yes}{\texttt{yes}}
\nc{\No}{\texttt{no}}

\nc{\Dom}{\ensuremath{\textbf{dom}}}
\nc{\schema}{\ensuremath{\sigma}}
\nc{\DB}{\ensuremath{D}} 
\nc{\DBStrich}{\ensuremath{D'}} 
\nc{\DBstart}{\ensuremath{{\DB_0}}} 
\nc{\DBempty}{\ensuremath{{\DB_{\emptyset}}}} 

\nc{\DBold}{\ensuremath{{\DB_{\textit{old}}}}}
\nc{\DBnew}{\ensuremath{{\DB_{\textit{new}}}}}
\nc{\UpdateSet}{\ensuremath{{U}}}

\nc{\Adom}{\ensuremath{\textrm{\upshape adom}}}
\nc{\adom}[1]{\ensuremath{\Adom(#1)}} 

\nc{\DS}{\ensuremath{\mathtt{D}}} 

\nc{\UpdateFont}[1]{\ensuremath{\textsf{#1}}}
\nc{\Delete}{\UpdateFont{delete}}
\nc{\Insert}{\UpdateFont{insert}}
\nc{\Update}{\UpdateFont{update}}

\nc{\AlgoFont}[1]{\ensuremath{\textbf{#1}}}
\nc{\PREPROCESS}{\AlgoFont{preprocess}}
\nc{\INIT}{\AlgoFont{init}}
\nc{\UPDATE}{\AlgoFont{update}}
\nc{\ENUMERATE}{\AlgoFont{enumerate}}
\nc{\COUNT}{\AlgoFont{count}}
\nc{\ANSWER}{\AlgoFont{answer}}
\nc{\TEST}{\AlgoFont{test}}

\nc{\EOE}{\texttt{EOE}} 
\nc{\Null}{\ensuremath{0}}

\nc{\preprocessingtime}{\ensuremath{t_p}}
\nc{\inittime}{\ensuremath{t_i}}
\nc{\delaytime}{\ensuremath{t_d}}
\nc{\updatetime}{\ensuremath{t_u}}
\nc{\updatetimeStrich}{\ensuremath{t'_u}}

\nc{\answertime}{\ensuremath{t_a}}
\nc{\countingtime}{\ensuremath{t_c}}
\nc{\testingtime}{\ensuremath{t_t}}

\nc{\arrayfont}[1]{\ensuremath{\texttt{#1}}}
\nc{\arrayA}{\arrayfont{A}}

\nc{\query}{\ensuremath{\varphi}}
\nc{\card}[1]{\ensuremath{|#1|}}

\subsection{Query-evaluation on dynamic databases}
\label{subsec:dynamicdbs}

In \cite{BKS_ICDT17}, Berkholz, Keppeler, and Schweikardt used the Hanf normal form result of \cite{HeiKS16} to design efficient algorithms for evaluating queries of first-order logic with modulo-counting quantifiers on dynamic databases.
It turns out that the methods of \cite{BKS_ICDT17} can easily be adapted to generalise to 
$\FOC(\Ps)$-queries, when using the Hanf normal form for $\FOC(\Ps)$ obtained from 
Theorem~\ref{thm:main}\eqref{thm:main:ghnf}$+$\eqref{thm:main:algo}.

To give a precise statement of the results, we need to provide some notation from \cite{BKS_ICDT17}.
We fix a countably infinite set $\Dom$, the \emph{domain} of potential
database entries. 
Consider a relational signature $\schema=\set{R_1,\ldots,R_{|\schema|}}$.
A $\schema$-\emph{database} $\DB$ ($\schema$-db, for short)
is of the form $\DB=(R_1^\DB,\ldots,R_{|\schema|}^\DB)$, where each 
$R_i^\DB$ is a finite subset of $\Dom^{\ar(R_i)}$.
The \emph{active domain} $\adom{\DB}$ of $\DB$ is the smallest subset
$A$ of $\Dom$ such that $R_i^\DB\subseteq A^{ar(R_i)}$ for each $R_i$
in $\schema$.
As usual in database theory, we identify a $\sigma$-db $\DB$ with the
$\schema$-structure $\A_{\DB}$ with universe $\adom{\DB}$ and relations
$R_i^{\DB}$ for each $R_i\in\schema$.
The \emph{degree} of $\DB$ is the degree of $\A_{\DB}$.
The \emph{cardinality} $\card{\DB}$ of $\DB$ is defined
as the number of tuples stored in $\DB$, i.e.,
$\card{\DB}\deff\sum_{R\in\schema} |R^{\DB}|$. 
The \emph{size} $\size{\DB}$ of $\DB$ is defined as
$\size{\schema}+|\Adom(\DB)|+\sum_{R\in\schema} \ar(R){\cdot}
|R^D|$ and corresponds to the size of a reasonable encoding of $\DB$.

For an $\FOCN(\Ps)$-formula $\phi$ with $\free(\phi)\subseteq \VARS$ and
for any tuple $\ov{x}=(x_1,\ldots,x_k)$ of pairwise distinct structure variables such 
that $\free(\phi)\subseteq\set{x_1,\ldots,x_k}$, the query result $\sem{\phi(\ov{x})}^{\DB}$ of $\phi(\ov{x})$ on $\DB$ is defined via
\[
  \sem{\phi(\ov{x})}^{\DB}
  \ \ = \ \ 
  \setc{\ov{a}\in \adom{\DB}^k}{\A_{\DB}\models\phi[\ov{a}]}\,.
\]
If $\phi$ is a sentence, then
the \emph{answer} $\sem{\phi}^{\DB}$ of $\phi$ on $\DB$ is defined as
$\sem{\phi}^{\DB}=\sem{\phi}^{\A_{\DB}}\in\set{0,1}$.

We allow to update a given $\sigma$-database by inserting or deleting
tuples as follows (note that both types of commands may change the
database's active domain and the database's degree).
A \emph{deletion} command is of the form
\Delete\,$R(a_1,\ldots,a_m)$
for $R\in\schema$, $m=\ar(R)$, and $a_1,\ldots,a_m\in \Dom$. When
applied to a $\schema$-db $\DB$, it results in the updated $\schema$-db
$\DB'$ with $R^{\DB'}= R^{\DB}\setminus\set{(a_1,\ldots,a_m)}$ and
$S^{\DB'}= S^{\DB}$ for all $S\in\schema\setminus\set{R}$.
An \emph{insertion} command is of the form
\Insert\,$R(a_1,\ldots,a_m)$
for $R\in\schema$, $m=\ar(R)$, and $a_1,\ldots,a_m\in \Dom$. 
When applied to a $\schema$-db $\DB$ in the unrestricted setting, it
results in the updated $\schema$-db 
$\DB'$ with $R^{\DB'}= R^{\DB}\cup\set{(a_1,\ldots,a_r)}$ and
$S^{\DB'}= S^{\DB}$ for all $S\in\schema\setminus\set{R}$.
Here, we restrict attention to databases of degree at most
$d$. Therefore, when applying an insertion command to a $\schema$-db $\DB$ of degree
$\leq d$, the command is carried out only if the resulting
database $\DB'$ still has degree $\leq d$; otherwise $\DB$ remains
unchanged and instead of carrying out the insertion command, an error
message is returned.

As in \cite{BKS_ICDT17}, we adopt the framework for dynamic algorithms for query evaluation of 
\cite{BKS_enumeration_PODS17}.
These algorithms are based on Random Access Machines (RAMs)
with $\bigoh(\log n)$ word-size and a uniform cost 
measure (cf., e.g., \cite{Cormen.2009}).
We assume that the RAM's memory is initialised to $\Null$.
Our algorithms will take as input a $\FOC(\Ps)$-formula $\phi(\ov{x})$ with
$\ov{x}=(x_1,\ldots,x_k)\in\VARS^k$ and a $\schema$-db $\DBstart$ of degree $\leq
d$.
For all query evaluation problems considered here, we aim at
routines
$\PREPROCESS$ and $\UPDATE$ which achieve the following.
 
Upon input of $\query(\ov{x})$ and $\DBstart$,
   $\PREPROCESS$ \allowbreak builds a data
   structure $\DS$ which represents $\DBstart$ (and which is designed in
   such a way that it supports the evaluation of $\query(\ov{x})$ on $\DBstart$).
Upon input of a command
   $\Update\ R(a_1,\ldots,a_m)$ (with
   $\Update\in\set{\Insert,\Delete}$), 
   calling $\UPDATE$  modifies the data structure $\DS$ such that it
   represents the updated database $\DB$.
The \emph{preprocessing time} $\preprocessingtime$ is the
time used for performing $\PREPROCESS$;
the \emph{update time} $\updatetime$ is the time used for performing
an $\UPDATE$.
By $\INIT$ we denote the particular case of the routine $\PREPROCESS$
upon input of a formula $\query(\ov{x})$ and the \emph{empty} database
$\DBempty$ (where $R^{\DBempty}=\emptyset$ for all $R\in\schema$).
The \emph{initialisation time} $\inittime$
is the time used for performing $\INIT$.
In the dynamic algorithms presented here, the $\PREPROCESS$ routine
for input of $\query(\ov{x})$ and $\DBstart$ 
carries out the $\INIT$ routine for $\query(\ov{x})$ and then
performs a sequence of $\card{\DBstart}$ update operations to
insert all the tuples of $\DBstart$ into the data structure.
Consequently, 
$\preprocessingtime = \inittime +  \card{\DBstart}\cdot\updatetime$.

Whenever speaking of a \emph{dynamic algorithm} we mean an algorithm
that has at least the routines $\PREPROCESS$ and $\UPDATE$.
In the following, $\DB$ will always denote the database that is
currently represented by the data structure $\DS$.
To \emph{answer} a \emph{sentence} $\phi$ under updates,
apart from the routines $\PREPROCESS$ and $\UPDATE$,
we aim at a routine $\ANSWER$ 
that outputs $\sem{\phi}^{\DB}$.
The \emph{answer time} $\answertime$ is the time used for
performing $\ANSWER$.

The following corollary is obtained by a straightforward adaptation of the proof of Theorem~4.1 of \cite{BKS_ICDT17}, where all uses of the Hanf normal form result for 
first-order logic with modulo-counting quantifiers of 
\cite{HeiKS16} are replaced by uses of
Theorem~\ref{thm:main}\eqref{thm:main:ghnf}$+$\eqref{thm:main:algo}. 

\begin{corollary}\label{cor:AnsweringBooleanQueries}
Let $(\Ps,\ar,\sem{.})$ be a numerical predicate collection.
There is a dynamic algorithm with oracle 
$\setc{(\P,\ov{n})}{\P\in\Ps, \ov{n}\in\sem{\P}}$
which receives as input a relational signature $\sigma$, a
degree bound $d\geq 2$, an $\FOC(\Ps)[\schema]$-sentence $\phi$, and a
$\schema$-db $\DBstart$ of degree $\leq d$, and computes within 
$\preprocessingtime= f(\phi,d)\cdot\size{\DBstart}$
preprocessing time a data structure that can be updated in time
$\updatetime= f(\phi,d)$ and allows to
return the query result $\sem{\phi}^{\DB}$ with answer time
$\answertime= \bigOh(1)$.
\\
The function $f(\phi,d)$ is of the form 
$\exp_5(\poly(\size{\phi})+\log\log d)$.
\end{corollary}
\begin{proof}
The proof is a straightforward adaptation of the proof of Theorem~4.1 of \cite{BKS_ICDT17}.

In addition to the statement made in
Corollary~\ref{cor:AnsweringBooleanQueries},
we also show the following:
If $\phi$ is a $d$-bounded hnf-sentence of locality radius $r$ (i.e.,
each sphere-formula that occurs in $\phi$ or in a numerical oc-type
condition of $\phi$, regards a sphere of degree $\leq d$ and radius
$\leq r$),
then
$f(\phi,d)= \poly(\size{\phi})+2^{\bigOh(\size{\schema}
  d^{2r+2})}$,
and the initialisation
time is $\inittime=\bigOh(\size{\phi})$.

W.l.o.g.\ we assume that all the symbols of $\schema$ occur in $\phi$
(otherwise, we remove from $\schema$ all symbols that do not occur in
$\phi$).
In the preprocessing routine, we first use 
Theorem~\ref{thm:main}\eqref{thm:main:ghnf}$+$\eqref{thm:main:algo} to
transform $\phi$ into a $d$-equivalent $\FOCN(\Ps)[\sigma]$-sentence $\psi$ in Hanf normal
form; this takes time 
$f(\phi,d)=\exp_5(\poly(\size{\phi})+\log\log d)$.

The sentence $\psi$ is a Boolean combination of 
atomic numerical oc-type conditions for \allowbreak $\FOCN(\Ps)[\sigma]$, each of which is of the form
$\P(t_1,\ldots,t_m)$ with $\P\in\Ps\cup\set{\P_\exists}$, $m=\ar(\P)$, and each 
$t_i$ is a simple counting term without number variables, i.e., it is a polynomial 
over basic counting terms and with integer coefficients. Recall that each basic counting term that occurs in $\psi$ is of the form 
$\Count{(y)}{\sph_\rho(y)}$ where $\rho$ is an $r$-type with 1 centre (over $\sigma$);
and from 
Theorem~\ref{thm:main}\eqref{thm:main:ghnf}
we know that 
$r< (2\bw(\phi){+}1)^{\br(\phi)}\leq 2^{\poly(\size{\phi})}$.

Let $\rho_1,\ldots,\rho_s$ be the list of all types $\rho$ that occur in
$\psi$. Thus, every basic counting term that occurs in $\psi$ is of the form
$\Count{(y)}{\sph_{\rho_j}(y)}$ for some $j\in[s]$.
For each $j\in[s]$ let $r_j$ be the radius of $\rho_j$. Thus, 
$\rho_j$ is an $r_j$-type with 1 centre (over $\schema$).

For each $j\in[s]$ our data structure will store the number
$\arrayA[j]$ of all elements $a\in\adom{\DB}$ whose $r_j$-type is
isomorphic to $\rho_j$, i.e., 
$\arrayA[j]=\sem{\Count{(y)}{\sph_{\rho_j}(y)}}^{\DB}$.
The initialisation for the empty database $\DBempty$ lets
$\arrayA[j]= 0$ for all $j\in[s]$.

In addition to the array $\arrayA$, our data structure
stores a Boolean value $\texttt{Ans}$ where $\texttt{Ans}=\sem{\phi}^{\DB}$
is the answer of the Boolean query $\phi$ on the current database
$\DB$. This way, the query can be answered in time
$\bigOh(1)$ by simply outputting $\texttt{Ans}$.

The initialisation for the empty database $\DBempty$ computes
$\texttt{Ans}$ as follows. 
Consider each atomic numerical oc-type condition $\P(t_1,\ldots,t_m)$ in $\psi$, 
evaluate each $t_i$ after replacing each basic counting term in $t_i$ by the number 0, and
let $n_i$ be the resulting integer. Use the oracle to determine if 
$(n_1,\ldots,n_m)\in\sem{\P}$. If the oracle answers ``yes'', replace each occurrence of $\P(t_1,\ldots,t_m)$ in $\psi$ by $\True$, and otherwise replace it by $\False$.
The resulting formula, a
Boolean combination of the Boolean constants $\True$ and $\False$,
then is
evaluated, and we let $\texttt{Ans}$ be the obtained result.
The entire initialisation takes time at most 
$\inittime= f(\phi,d)$. If $\phi$ is
a hnf-sentence, we even have 
 $\inittime=\poly(\size{\phi})$.

To update our data structure upon a command
$\Update\,R(a_1,\ldots,a_k)$, for $k=\ar(R)$ and 
$\Update\in\set{\Insert,\Delete}$, we proceed as follows.
The idea is to remove from the data structure the
information on all the database elements  
whose $r_j$-neighbourhood (for some $j\in[s]$) is affected by
the update, and then to recompute
the information concerning all these elements
 on the updated database.

Let $\DBold$ be the database before the update is received and
let $\DBnew$ be the database after the update has been performed. 
We consider each $j\in[s]$.
All elements whose $r_j$-neighbourhood might have changed, belong to
the set 
\;$\UpdateSet_j\deff\neighb{r_j}{\DBStrich}{\ov a}$, \ where
$\DBStrich\deff\DBnew$ if the update command is $\Insert\; R(\ov{a})$,
and $\DBStrich\deff\DBold$ if the update command is $\Delete\; R(\ov{a})$.

 To remove the old information from $\arrayA[j]$, we compute for each
 $a \in \UpdateSet_j$ the neighbourhood
 $\B_a\deff \Neighb{r_j}{\DBold}{a}$, check whether
 $\B_a\isom\rho_j$, and if so, 
 decrement the value $\arrayA[j]$. 
\\
 To recompute the new information for $\arrayA[j]$, we compute for all 
 $a \in \UpdateSet_j$ the neighbourhood
 $\B'_a\deff \Neighb{r_j}{\DBnew}{a}$, check whether
 $\B'_a\isom\rho_j$, and if so, 
 increment the value $\arrayA[j]$. 

 Using Lemma~\ref{lem:basic_facts} we obtain for each $j\in[s]$ that
 $|\UpdateSet_j|\leq kd^{r_j+1}$.
 For each $a\in \UpdateSet_j$, the
 neighbourhoods $\B_a$ and $\B'_a$ can be computed in time
 $\big(d^{r_j+1}\big)^{\bigOh(\size{\schema})}$, and testing for
 isomorphism with $\rho_j$ can be done in time 
 $\big(d^{r_j+1}\big)^{\bigOh(\size{\schema}+d^{r_j+1})}$. 
 Thus, 
 the update of $\arrayA[j]$ is done in time
 $k{\cdot}
 \big(d^{r_j+1}\big)^{\bigOh(\size{\schema}+d^{r_j+1})} = 
 2^{d^{2^{\poly(\size{\phi})}}}= \exp_3(\poly(\size{\phi})+\log\log d)$.

After having updated $\arrayA[j]$ for each $j\in[s]$, we recompute the 
query answer $\texttt{Ans}$ in a similar way as in the initialisation for the empty database: 
Consider each atomic numerical oc-type condition $\P(t_1,\ldots,t_m)$ in $\psi$, 
evaluate each $t_i$ after replacing each basic counting term of the form
$\Count{(y)}{\sph_{\rho_j}(y)}$ in $t_i$ by the number $\arrayA[j]$, and
let $n_i$ be the resulting integer. Use the oracle to determine if 
$(n_1,\ldots,n_m)\in\sem{\P}$. If the oracle answers ``yes'', replace each occurrence of $\P(t_1,\ldots,t_m)$ in $\psi$ by $\True$, and otherwise replace it by $\False$.
The resulting formula, a
Boolean combination of the Boolean constants $\True$ and $\False$,
then is
evaluated, and we let $\texttt{Ans}$ be the obtained result.
Thus, recomputing $\texttt{Ans}$ takes time $\poly(\size{\psi})$.

In summary, the entire update time is 
$\updatetime=f(\phi,d)=\exp_5(\poly(\size{\phi})+\log\log d)$.

In case that $\phi$ is a $d$-bounded hnf-sentence of locality radius $r$, we even have 
$\updatetime=\poly(\size{\phi})+k{\cdot}\big(d^{r+1}\big)^{\bigOh(\size{\schema}+d^{r+1})}
\leq \poly(\size{\phi})+ 2^{\bigOh(\size{\schema}d^{2r+2})}$.
Note that for a $d$-bounded $r$-type $\rho$ with 1 centre (over $\schema$),
the formula $\sphere{\rho}{y}$ has size
$(d^{r+1})^{\Omega(\size{\schema})}$.
Hence, if $\phi$ is a
$d$-bounded hnf-sentence of locality radius $r$, then the update time $\updatetime$
also is in $2^{\poly(\size{\phi})}$.
This completes the proof of
Corollary~\ref{cor:AnsweringBooleanQueries}.
\qed
\end{proof}

\bigskip

Regarding the evaluation of queries $\phi(\ov{x})$ where $\ov{x}=(x_1,\ldots,x_k)$ is 
a tuple of arity $k>0$, the framework of \cite{BKS_ICDT17} considers the following 
problems.
To \emph{test} if a given tuple belongs to the query result,
we aim at a routine $\TEST$ which upon
input of a tuple $\ov{a}\in\Dom^k$ checks whether $\ov{a}\in
\sem{\query(\ov{x})}^{\DB}$.
The \emph{testing time} $\testingtime$ is the time used for
performing a $\TEST$.
To solve the \emph{counting problem under updates}, 
we aim at a routine $\COUNT$ which outputs the cardinality
$|\sem{\query(\ov{x})}^{\DB}|$ of the query result.
The \emph{counting time} $\countingtime$ is the time used for
performing a $\COUNT$.
To solve the \emph{enumeration problem under updates}, 
we aim at a routine $\ENUMERATE$ such that
calling $\ENUMERATE$ invokes an enumeration of all tuples
(without repetition) that belong to
the query result $\sem{\query(\ov{x})}^{\DB}$.
The \emph{delay} $\delaytime$ is the maximum time
used during a call of $\ENUMERATE$
\begin{itemize}
\item until the output of the first tuple (or the end-of-enumeration
  message $\EOE$, 
  if $\sem{\phi(\ov{x})}^{\DB}=\emptyset$),
\item between the output of two consecutive tuples, and 
\item between the output of the last tuple and the end-of-enumeration
  message $\EOE$.
\end{itemize}

\medskip

The proof of the following corollary is obtained from the proofs of Theorem~6.1, Theorem~8.1, and Theorem~9.4  of \cite{BKS_ICDT17} by replacing
all uses of Theorem~4.1 of \cite{BKS_ICDT17} by 
Corollary~\ref{cor:AnsweringBooleanQueries}.

\begin{corollary}\label{cor:TestingCountingEnumerating}
Let $(\Ps,\ar,\sem{.})$ be a numerical predicate collection.
There is a dynamic algorithm with oracle 
$\setc{(\P,\ov{n})}{\P\in\Ps, \ov{n}\in\sem{\P}}$
which receives as input 
a relational signature $\schema$, a
degree bound $d\geq 2$, an $\FOC(\Ps)[\schema]$-formula $\phi(\ov{x})$ with
free variables $\ov{x}=(x_1,\ldots,x_k)\in\VARS^k$ (for
some $k\in\NN$), and a
$\schema$-db $\DBstart$ of degree $\leq d$, and computes within 
$\preprocessingtime= f(\phi,d)\cdot\size{\DBstart}$
preprocessing time a data structure that can be updated in time
$\updatetime= f(\phi,d)$ and allows to
\begin{itemize}
\item
test for any input tuple 
$\ov{a}\in\Dom^k$ whether $\ov{a}\in\sem{\phi(\ov{x})}^{\DB}$ within testing time
$\testingtime= \bigOh(k^2)$
\item
return the cardinality 
$|\sem{\phi(\ov{x})}^{\DB}|$ 
of the query result within time $\bigOh(1)$
\item
enumerate $\sem{\phi(\ov{x})}^{\DB}$ with $\bigOh(k^3)$ delay.
\end{itemize}
The function $f(\phi,d)$ is of the form 
$\exp_5(\poly(\size{\phi})+\log\log d)$.
\end{corollary}

\section{The lower bound}\label{sec:lowerbound}
\newcommand{\tree}{\mathrm{tree}}

This section is devoted to the proof of Theorem~\ref{thm:main}\eqref{thm:main:lowerbound}.
In fact, we will prove a slightly stronger lower bound, for which the following notation is needed.

A counting term $t(\ov{x})$ is \emph{linear} if it is of the form
$i$ with $i\in\NN$ or of the form
\[
    \sum_{j=1}^n b_j\cdot\Count{(y)}{\varphi_j(\ov{x},y)} \ \ - \ \ b_0
\]
with $n\ge0$ and $b_0,b_1,\dots,b_n\in\bN$. In particular, linear
counting terms do not count \emph{tuples} of variables, but only
\emph{single} variables. In addition, they are not arbitrary
polynomials, but linear polynomials with non-negative coefficients and
a negative constant term. An $\FOCN(\Ps)$-formula is \emph{linear}
if it only uses linear counting terms. Note that all formulas from
$\FO(\Ps)$ are linear since the only counting terms allowed there are
of the form $\Count{(y)}{\varphi(\ov{x},y)}-k$, for $k\in \NN$. 

By Theorem~\ref{thm:main}\eqref{thm:main:ghnf}, for any formula $\varphi\in\FOCN(\Ps)$
and any $d\ge2$, there exists a $d$-equivalent hnf-formula $\psi$. 
A close inspection of the proof of Theorem~\ref{thm:main}\eqref{thm:main:ghnf} shows that
if $\varphi$ is linear, then $\psi$ is linear as well; moreover, if $\psi$ belongs to
$\FO(\Ps)$, then all counting terms that appear in $\psi$ even are of the form
$i$ with $i\in\bN$ or of the form
\[
  \sum_{\tau\in T}\Count{(y)}{\sph_\tau(y)} \ \ - \ \ k
\]
for some $k\in\bN$ and some set $T$ of types of radius~$r$.
As the bounds from Theorem~\ref{thm:main}\eqref{thm:main:ghnf}
apply here as well, it follows that the
number of distinct numerical oc-type conditions in $\psi$ is at most 
$\exp_4\bigl(\poly(\size{\phi}+\size{\sigma})+\log\log(d)\bigr)$.
In this section, we show a matching lower bound for \emph{linear} hnf-formulas for $\FOC(\Ps)$,
provided that $\Ps$ contains a predicate that is \emph{rich} in the following sense.

\begin{definition}[Rich numerical predicate] \ \\
  A set $R\subseteq\bN$ of natural numbers is \emph{rich} if for all
  $s,u,v\in\bN$, all $\ov{a_0}\in\set{0,1}^s\setminus\set{\ov{0}}$,
  $\ov{a_1},\dots,\ov{a_u}\in\bN^s$, and all $c_1,\dots,c_u\in\bN$
  with $(\ov{a_0},0)\neq(\ov{a_i},c_i)$ for all $i\in\set{1,\ldots,u}$, there
  exist $\ov{x},\ov{y}\in(v+\bN)^s$ such that\footnote{For $s$-dimensional vectors $\ov{a}$ and $\ov{x}$ 
   we write $\ov{a}\transpose\,\ov{x}$ for the usual inner product, i.e., the number
  $\sum_{j=1}^s a_jx_j$, where $a_j$ (and $x_j$) denotes the $j$-th component of $\ov{a}$ (and $\ov{x}$).}
  \begin{itemize}
  \item
    $\ov{a_0}\transpose\,\ov{x}\in R\iff
    \ov{a_0}\transpose\,\ov{y}\notin R$ \ and
  \item
    $\ov{a_i}\transpose\,\ov{x}-c_i\in R\iff
    \ov{a_i}\transpose\,\ov{y}-c_i\in R$ for all $i\in\set{1,\ldots,u}$.
  \end{itemize}
\end{definition}

\begin{example}\label{example:rich}
  We can produce probabilistically a set $R$ of natural numbers as
  follows: for every $n\in\bN$, toss a fair coin and place $n$ into
  $R$ iff the outcome is tail. Then, with probability 1, we get a rich
  set (cf.~Section~\ref{SS-random}). Furthermore,
  Section~\ref{SS-large-gaps} shows that a set $R\subseteq\bN$ is
  rich whenever it has ``large gaps'', i.e., $R$ is infinite, $0\notin R$, and for all
  $d\in\bN_{\ge1}$, there 
  exists $q\in R$ such that
  $\left[\left\lfloor\nicefrac{q}{d}\right\rfloor,d\cdot q\right]\cap
  R=\set{q}$. Examples of such sets are $\setc{n^n}{n\in\bN}$,
  $\setc{\lfloor 2^{n^c}\rfloor}{n\in\bN}$ for all reals $c>1$,
  $\setc{n!}{n\in\bN}$, as well as all infinite subsets of these sets.
  But note that neither the set $\setc{2^n}{n\in\bN}$ nor (by
  Bertrand's postulate) the set of primes has large gaps.
\end{example}
 
Our main lower bound result reads as follows.

\begin{theorem}\label{thm:mainlowerbound}
Let $\Ps=\set{\quant{R}}$ with $\ar(\quant{R})=1$ and $\sem{\quant{R}}\subseteq\bN$ rich.
Let $\sigma_\tree=\set{E_0,E_1,X}$ be the signature consisting of two binary relation symbols $E_0$, $E_1$ and a
unary relation symbol $X$.
There is a sequence $(\varphi_n)_{n\geq 1}$ of $\FO(\Ps)[\sigma_\tree]$-sentences
of size $\bigO(n)$ such that for all $n\geq 1$ and for every
linear hnf-sentence $\psi_n\in \FOC(\Ps\cup\set{\P_{\exists}})[\sigma_\tree]$ that is 
3-equivalent to $\varphi_n$, the number of
distinct numerical oc-type conditions in $\psi_n$ 
is at least $\exp_4(n)$.
\end{theorem}

For proving Theorem~\ref{thm:mainlowerbound} let us consider the following $\sigma_\tree$-structures.
For $n\in\bN$ let $\cT_n$ denote the set of all \emph{(complete
  labeled ordered binary) trees of height~$2^n$}, i.e., of all
structures $T=(D,E_0,E_1,X)$ with $D$ the set of binary words of
length at most $2^n$, $E_b=\setc{(u,ub)}{ub\in D}$ for
$b\in\set{0,1}$, and $X\subseteq D$.

A tree~$T$ is \emph{marked} if the root $\varepsilon$ belongs to $X$,
i.e., is labeled. Otherwise, $T$ is \emph{unmarked}. For a tree
$T=(D,E_0,E_1,X)$, let $\marked(T)=(D,E_0,E_1,X\cup\set{\varepsilon})$
denote the marked tree that is obtained by labelling the root; the
unmarked tree $\unmarked(T)=(D,E_0,E_1,X\setminus\set{\varepsilon})$
is defined analogously.

For a finite set $S$ of trees, a \emph{forest over $S$} 
(an \emph{$S$-forest}, for short) is a disjoint union of finitely many copies of trees
from~$S$. Since every tree is finite, the same applies to every
$S$-forest.

\begin{definition}[Property $P_R(\cF)$] \ \\
For a $\cT_n$-forest $\cF$ and a set $R\subseteq\bN$, let $P_R(\cF)$
be the following property:
\begin{center}
    The number of unmarked trees $T$ in $\cF$ such that $\marked(T)$
    also appears in $\cF$ belongs to $R$.
\end{center}
\end{definition}

\noindent
The following lemma is the technical core of the proof of
Theorem~\ref{thm:mainlowerbound}.

\begin{lemma}\label{lem:lower-bound-linear}
  Let $\Ps=\set{\quant{R}}$ with $\ar(\quant{R})=1$ and
  $\sem{\quant{R}}\subseteq\bN$ rich.
  Let $n\in\bN$ and
  $\psi\in\FOC(\Ps\cup\{\P_{\exists}\})[\sigma_\tree]$ be a linear
  hnf-sentence such that for all $\cT_n$-forests $\cF$ we have
  \begin{equation}
    \label{eq:assumption}
     \cF\models\psi
     \quad \iff \quad 
     P_{\sem{\quant{R}}}(\cF)\,.
  \end{equation}  
  Then, the number of distinct atomic numerical oc-type conditions of
  the form $\quant{R}(t)$ in $\psi$
  is at least as big as the number of non-empty sets 
  $B\subseteq\marked(\cT_n)$ of marked trees of height $2^n$.
\end{lemma}

\begin{proof}
  Let $\sigma\deff\sigma_\tree$.
  Since $\psi$ is a linear hnf-sentence from $\FOC(\Ps\cup\set{\P_{\exists}})[\sigma]$, it is a
  Boolean combination of atomic numerical oc-type conditions of the form $\P(t)$ with
  $\P\in\set{\quant{R},\P_{\exists}}$ and $t$ a linear counting term.
  Thus, there exist numbers $u,u'\in\NN$ and linear simple counting
  terms $t_1,\ldots,t_{u+u'}$ such that $\psi$ is a Boolean
  combination of the atomic numerical oc-type conditions
  \[
     \quant{R}(t_1),\ \ldots, \ \quant{R}(t_u),\ \
     \P_{\exists}(t_{u+1}),\ \ldots,\ \P_{\exists}(t_{u+u'})\,.
  \]
  Let $r$ be the locality radius of $\psi$ and let
  $L\deff \Typeslistd{\le r}{1}$ denote the list of all $d$-bounded types of
  radius $\le r$ with a single centre. 
  Now consider an $i\in\set{1,\ldots,u+u'}$.
  Since the counting term $t_i$ is
  linear, there are natural numbers $b_{0,i}$ and $b_{\tau,i}$ for
  $\tau\in L$ such that
  \begin{equation}
    \label{eq:ti}
    t_i \quad = \quad \sum_{\tau\in L}b_{\tau,i}\cdot\Count{(y)}{\sph_\tau(y)} \ \ - \ \ b_{0,i}\,.
  \end{equation}

For each $\tau\in L$ and each tree $T\in\cT_n$ we write
$\real_\tau^T$ to denote the number of
nodes $v$ from $T$ whose neighbourhood is isomorphic to $\tau$, i.e.,
\[
  \real_\tau^T \quad \deff \quad \Setsize{\setc{\,v\in D\;}{\;(T,v)\models\sph_\tau\,}}\,.
\]
To apply the richness of $\sem{\quant{R}}$, we consider the following vectors from $\bN^{\cT_n}$
and numbers from~$\bN$:

For each $i\in\set{1,\ldots,u}$ let
\[
 \ov{a_i}\ \deff \ (a_{i,T})_{T\in \cT_n}
 \quad\text{with}\quad
 a_{i,T}\ \deff \ \sum_{\tau\in L} \big(b_{\tau,i}\cdot\real_{\tau}^T\big)
\]
and let
\[
  c_i \ \deff \ b_{0,i}\,.
\]

Furthermore, let us fix a number $v\in\bN$ with
$v>b_{0,i}$ for all $i\in\set{1,\ldots,u+u'}$.
I.e., $v> c$ for all constant terms of
linear counting terms that appear in some numerical oc-type
condition of $\psi$.

Finally, for every non-empty set $B\subseteq\marked(\cT_n)$
let $S\deff B\cup\unmarked(\cT_n)$ and
consider the vectors
\[
 \ov{a_i}^B \ \deff\ (a_{i,T})_{T\in S} \qquad \text{for all $i\in\set{1,\ldots,u}$},
\]
and
\[
  \ov{a_0}^B\ \deff \ (a_T^B)_{T\in S} \quad \text{with}\quad 
   a_T^B \ \deff \
    \begin{cases}
      1 & \text{if }T\in \unmarked(B)\\
      0 & \text{otherwise.}
    \end{cases}
\]

\begin{claim}
 For every non-empty $B\subseteq\marked(\cT_n)$ there is an
 $i\in\set{1,\ldots,u}$ with $(\ov{a_0}^B,0)=(\ov{a_i}^B,c_i)$. 
\end{claim}  
\begin{proof}
Consider an arbitrary non-empty set $B\subseteq\marked(\cT_n)$ and let 
$S\deff B\cup\unmarked(\cT_n)$.

    Towards a contradiction, suppose that $(\ov{a_0}^B,0)\neq(\ov{a_i}^B,c_i)$
    for all $i\in\set{1,\ldots,u}$. Note that $\ov{a_0}^B\neq\ov{0}$ since
    $B\neq\emptyset$.
    Since the set $\sem{\quant{R}}$ is rich, there are vectors
    $\ov{x},\ov{y}\in(v+\bN)^S$ such that
    \begin{enumerate}[(i)]
    \item\label{eq:claim:i}
      $(\ov{a_0}^B)\transpose\,\ov{x}\in \sem{\quant{R}}\iff
      (\ov{a_0}^B)\transpose\,\ov{y}\notin \sem{\quant{R}}$ \ \ and
    \item\label{eq:claim:ii}
      $(\ov{a_i}^B)\transpose\,\ov{x}-c_i\in \sem{\quant{R}}\iff
      (\ov{a_i}^B)\transpose\,\ov{y}-c_i\in \sem{\quant{R}}$ \ \ for all
      $i\in\set{1,\ldots,u}$.
    \end{enumerate}

    From these vectors $\ov{x}$ and $\ov{y}$, we build $S$-forests $\cF_{\ov{x}}$ and
    $\cF_{\ov{y}}$ as follows: For each $T\in S$, the $S$-forest
    $\cF_{\ov{x}}$ contains $x_T$ copies of $T$, and the $S$-forest
    $\cF_{\ov{y}}$ contains $y_T$ copies of $T$. I.e.,
    \[
    \cF_{\ov{x}}\ \ \deff \ \ \biguplus_{T\in S} x_T\cdot T
    \qquad\text{ and }\qquad
    \cF_{\ov{y}}\ \ \deff \ \ \biguplus_{T\in S} y_T\cdot T
    \]
    where $x_T\cdot T$ denotes the disjoint union of $x_T$ copies of
    the tree $T$. Let
    \begin{equation}
      \label{eq:t}
      t \quad = \quad \sum_{\tau\in L}
      b_\tau\cdot\Count{(y)}{\sph_\tau(y)} \ \ - \ \ b_0
    \end{equation}
    be one of the linear counting terms $t_1,\ldots,t_{u+u'}$ that
    occur in $\psi$. By definition of $\cF_{\ov{x}}$ we have
    \begin{align}
      \sem{t}^{\cF_{\ov{x}}} 
      & \quad=\quad \sum_{\tau\in
        L}\left(b_{\tau}\cdot\sum_{T\in S}\big(\real_\tau^T\cdot
        x_T\big)\right) \ \ - \ \ b_{0}\notag\\
      & \quad=\quad \sum_{T\in S}\left(\left(\sum_{\tau\in L}
          \big(b_{\tau}\cdot\real_\tau^T\big)\right)\cdot x_T\right) \
      \ - \ \ b_{0}\label{align:1}\intertext{and similarly}
        \sem{t}^{\cF_{\ov{y}}} 
      &\quad=\quad \sum_{T\in S}\left(\left(\sum_{\tau\in L}
          \big(b_{\tau}\cdot\real_\tau^T\big)\right)\cdot y_T\right) \
      \ - \ \ b_{0}\label{align:2}
    \end{align}

    We next show that
    \begin{equation}
      \label{eq:indistinguishable-by-ghs}
      \cF_{\ov{x}}\models\delta \quad \iff \quad \cF_{\ov{y}}\models\delta
    \end{equation}
    is true for all atomic numerical oc-type conditions $\delta$ that appear in
    $\psi$.

    First consider the case that $\delta=\P_{\exists}(t_i)$ for some
    $i\in\set{u+1,\ldots,u+u'}$, let $t\deff t_i$ and suppose that
    $\cF_{\ov{x}}\models\P_{\exists}(t)$, i.e.,
    $\sem{t}^{\cF_{\ov{x}}}\geq 1$. By \eqref{align:1}, there exists $T\in S$
    with $\sum_{\tau\in L}\big(b_\tau\cdot\real_\tau^T\big)\geq 1$. Recall that $y_T\ge v>b_0$. Hence
    \[
      b_0 \ \ < \ \ v \ \ \le \ \ y_T \ \ \le \ \ \sem{t}^{\cF_{\ov{y}}}+b_0
    \]
    by \eqref{align:2}. Hence $\sem{t}^{\cF_{\ov{y}}}\geq 1$ and therefore
    $\cF_{\ov{y}}\models\P_{\exists}(t)$, i.e.,
    $\cF_{\ov{y}}\models\delta$. By symmetry, we obtain the
    equivalence \eqref{eq:indistinguishable-by-ghs} in this case.

    Next, we consider the case $\delta=\quant{R}(t_i)$ for some
    $i\in\set{1,\ldots,u}$.
    Note that \eqref{align:1},
    \eqref{align:2} and the definition of $\ov{a_i}^B$ and $c_i$ imply
    $\sem{t_i}^{\cF_{\ov{x}}}=(\ov{a_i}^B)\transpose\,\ov{x}-c_i$ and
    $\sem{t_i}^{\cF_{\ov{y}}}=(\ov{a_i}^B)\transpose\,\ov{y}-c_i$. Hence
    we get
    \begin{align*}
      \cF_{\ov{x}}\models\delta
      &\quad\iff\quad (\ov{a_i}^B)\transpose\,\ov{x}-c_i\in \sem{\quant{R}}\\
      &\quad\stackrel{\eqref{eq:claim:ii}}{\iff}\quad (\ov{a_i}^B)\transpose\,\ov{y}-c_i\in \sem{\quant{R}}\\
      &\quad\iff\quad \cF_{\ov{y}}\models\delta\,.
    \end{align*}
    Thus, the equivalence \eqref{eq:indistinguishable-by-ghs} holds
    also in this case.\bigskip

    Since $\psi$ is a Boolean combination of numerical
    oc-type conditions for which \eqref{eq:indistinguishable-by-ghs}
    holds, we obtain that
    \[
      P_{\sem{\quant{R}}}(\cF_{\ov{x}}) \quad\stackrel{\eqref{eq:assumption}}{\iff}\quad
      \cF_{\ov{x}}\models\psi\quad\iff\quad\cF_{\ov{y}}\models\psi
      \quad\stackrel{\eqref{eq:assumption}}{\iff}\quad P_{\sem{R}}(\cF_{\ov{y}})\,.
    \]

    Now we derive a contradiction as follows: Let $T\in\marked(\cT_n)$
    be some marked tree. Then $T$ appears in $\cF_{\ov{x}}$ if and
    only if $T\in B$. Hence we get the following:
    \begin{align*}
      P_{\sem{\quant{R}}}(\cF_{\ov{x}})
        &\ \ \iff \ \
          \begin{minipage}[t]{.7\linewidth}
            the number of copies of unmarked trees $T$ in
            $\cF_{\ov{x}}$ such that the marked tree $\marked(T)$ appears in
            $\cF_{\ov{x}}$ belongs to $\sem{\quant{R}}$
          \end{minipage}\\
        &\ \ \iff \ \
          \begin{minipage}[t]{.7\linewidth}
            the number of copies of trees from $\unmarked(B)$ in
            $\cF_{\ov{x}}$ belongs to $\sem{\quant{R}}$
          \end{minipage}\\
        &\ \ \iff \ \ \sum_{T\in\unmarked(B)}x_T \ \ \in \ \ \sem{\quant{R}}\\
        &\ \ \stackrel{\eqref{eq:claim:i}}{\iff} \ \
        \sum_{T\in\unmarked(B)}y_T \ \ \notin \ \ \sem{\quant{R}}\\
        &\ \ \iff \ \ P_{\sem{\quant{R}}}(\cF_{\ov{y}})\text{ does not hold.}
    \end{align*}
    This contradiction completes the indirect proof of the claim.
  \end{proof}
  \bigskip

  Now
  consider two distinct non-empty sets $B,B'\subseteq\marked(\cT_n)$.
  By the claim we know that there are $i,i'\in\set{1,\ldots,u}$ such
  that 
  $(\ov{a_0}^B,0)=(\ov{a_i}^B,c_i)$ and
  $(\ov{a_0}^{B'},0)=(\ov{a_{i'}}^{B'},c_{i'})$.
  To finish the proof of Lemma~\ref{lem:lower-bound-linear}, it
  suffices to show that $i\neq i'$. 
  For contradiction, assume that $i=i'$.
  W.l.o.g., there is a marked tree $\tilde{T}\in B\setminus B'$. Let
  $T$ be the unmarked version of $\tilde{T}$.
  Thus, $a_T^B=1$ since $T\in \unmarked(B)$, and $a_T^{B'}=0$ since
  $T\not\in\unmarked(B')$.
  But as we assume that $i=i'$, we have
  $a_T^B=a_{i,T}=a_{i',T}=a_T^{B'}$. 
  This contradiction completes the proof of Lemma~\ref{lem:lower-bound-linear}.
\qed
\end{proof}

\bigskip

We are now ready for the proof of Theorem~\ref{thm:mainlowerbound}.

\bigskip

\begin{proof}[\textbf{Proof of Theorem~\ref{thm:mainlowerbound}}]
  A construction by Frick \& Grohe~\cite[Lemma~25]{FriG04} provides us
  with a sequence of formulas $\varphi_n\in\FO(\Ps)[\sigma_\tree]$ of
  size $\bigOh(n)$ such that, for all $n\in\bN$ and all
  $\cT_n$-forests~$\cF$, we have
  \[
   P_{\sem{\quant{R}}}(\cF) \quad \iff \quad \cF\models\varphi_n.
  \]
  Precisely, the proof of \cite[Lemma~25]{FriG04} provides us with a
  sequence $\psi_n(x_1,x_2,y_1,y_2)$ of $\FO[\sigma_\tree]$-formulas
  of length $\bigOh(n)$ such that for all $n\in\NN$, all
  $\cT_n$-forests $\cF$ and all nodes $a_1,a_2,b_1,b_2$ of $\cF$ we
  have
  $\cF\models\psi_n[a_1,a_2,b_1,b_2]$
  if and only if to go from $a_1$ to $a_2$ in $\cF$ we must follow the
  same sequence of $E_0,E_1$-edges as to go from $b_1$ to $b_2$ in $\cF$.
  Using these formulas $\psi_n(x_1,x_2,y_1,y_2)$, we can choose
  \[
    \phi_n \quad \deff \quad
    \quant{R}\,\big(\, \Count{(x)}{\chi_n(x)}\,\big)\,,\qquad \text{where}
  \]    
  $\chi_n(x)$ states that $x$ is an unmarked root node for which there
  exists a marked root node $y$ such that for all nodes $x'$ and $y'$
  that satisfy $\psi_n(x,x',y,y')$ and $x'\neq x$ and $y'\neq y$,
  we have $X(x')\gdw X(y')$.

  Let $n\ge 1$ and let $\psi\in\FOC(\Ps\cup\set{\P_{\exists}})[\sigma_\tree]$ be a
  linear hnf-sentence that is $3$-equivalent to $\varphi_n$. Then, by
  Lemma~\ref{lem:lower-bound-linear}, $\psi$ contains at least
  $2^{\frac{1}{2}\setsize{\cT_n}}-1$ distinct atomic numerical oc-type
  conditions.
  Note that each tree $T$ in $\cT_n$ is a complete labeled ordered binary
  tree of height $2^n$. Therefore, each $T\in\cT_n$ has
  $\sum_{h=0}^{2^n}2^h=2^{2^n+1}-1 \geq 2^{2^n}+3$ nodes.
  Thus, $\setsize{\cT_n}\geq 2^{(2^{2^n}+3)}$, and hence
  $\frac{1}{2}\setsize{\cT_n}\geq 2^{(2^{2^n}+2)}\geq \exp_3(n)+1$.
  In summary, the number of distinct atomic numerical oc-type
  conditions in $\psi$ is at least 
  $2^{\frac{1}{2}\setsize{\cT_n}}-1 \geq 2^{\exp_3(n)+1} -1\geq \exp_4(n)$.
\qed
\end{proof}

\bigskip

Note that Theorem~\ref{thm:main}\eqref{thm:main:lowerbound} is an immediate consequence of Theorem~\ref{thm:mainlowerbound}:

\bigskip

\begin{proof}[\textbf{Proof of Theorem~\ref{thm:main}\eqref{thm:main:lowerbound}}]
  Let $\psi\in\FO(\Ps\cup\set{\P_{\exists}})$ be a formula in weak Hanf normal
  form. Replacing every counting term
  $\Count{(y)}{\bigvee_{\tau\in T}\sph_\tau(y)} \ - k$ by 
  $\sum_{\tau\in T}\Count{(y)}{\sph_\tau(y)} - k$ yields an equivalent
  linear hnf-formula from $\FOC(\Ps\cup\set{\P_{\exists}})$. Hence, the
  lower bound of Theorem~\ref{thm:mainlowerbound} also applies to formulas in weak Hanf normal form, 
  providing a proof of Theorem~\ref{thm:main}\eqref{thm:main:lowerbound}.
\end{proof}

\bigskip

The remainder of the section is devoted to proving the statements of
Example~\ref{example:rich}.
In Section~\ref{SS-random} we show that random sets are rich, and in
Section~\ref{SS-large-gaps} we show that sets with ``large gaps'' are rich.

\subsection{Random sets are rich}\label{SS-random}
 
In this subsection we show that
``almost all'' sets of natural numbers are rich ---
something one would presumably not expect when looking at the
definition.

\begin{lemma}\label{L-semirich}
  Let $R\subseteq\bN$ such that, for all $s,u,v\in\bN$, all
  $\ov{a_0}\in\set{0,1}^s\setminus\set{\ov{0}}$,
  $\ov{a_1},\dots,\ov{a_u}\in\bN^s$, and all
  $c_1,\dots,c_u\in\bN$ with $(\ov{a_0},0)\neq(\ov{a_i},c_i)$
  for all $i\in\set{1,\ldots,u}$ and $c_0\deff 0$, there exist $\ov{x},\ov{y}\in(v+\bN)^s$ such
  that
  \begin{subequations}
    \begin{align}
      \label{eq:x}
        &\ov{a_i}\transpose\,\ov{x}-c_i\in R
          \text{ \quad for all } i\in\set{0,\ldots,u} \text{ with }\ov{a_i}\neq\ov{0}\,,\\     
      \label{eq:y1}
        &\ov{a_0}\transpose\,\ov{y}-c_0\notin R\text{, \ and }\\
      \label{eq:y2}
        &\ov{a_i}\transpose\,\ov{y}-c_i\in R
          \text{ \quad for all } i\in\set{1,\ldots,u} \text{ with }\ov{a_i}\neq\ov{0}.     
    \end{align}
  \end{subequations}
  Then $R$ is rich.
\end{lemma}

\begin{proof}
  To show that $R$ is rich, let $s,u,v\in\bN$,
  $\ov{a_0}\in\set{0,1}^s\setminus\set{\ov{0}}$,
  $\ov{a_1},\dots,\ov{a_u}\in\bN^s$, and $c_1,\dots,c_u\in\bN$ with
  $(\ov{a_0},0)\neq(\ov{a_i},c_i)$ for all $i\in\set{1,\ldots,u}$. For
  notational simplicity, we set $c_0\deff 0$. By the assumption on $R$,
  there exist vectors $\ov{x},\ov{y}\in(v+\bN)^s$ satisfying
  \eqref{eq:x}--\eqref{eq:y2}.

  Then, by \eqref{eq:x} and \eqref{eq:y1}, we have
  \[
     \ov{a_0}\transpose\,\ov{x}-c_0\in R \quad\iff\quad
     \ov{a_0}\transpose\,\ov{y}-c_0\notin R\,.
  \]
  By \eqref{eq:x} and \eqref{eq:y2}, we furthermore get
  \[
     \ov{a_i}\transpose\,\ov{x}-c_i\in R \quad\iff\quad
     \ov{a_i}\transpose\,\ov{y}-c_i\in R
  \]
  for all $i\in\set{1,\ldots,u}$ with $\ov{a_i}\neq\ov{0}$. If
  $\ov{a_i}=\ov{0}$, this also holds since then
  \[
     \ov{a_i}\transpose\,\ov{x}-c_i \quad = \quad -c_i \quad =\quad
     \ov{a_i}\transpose\,\ov{y}-c_i\,.
  \]
  \qed
\end{proof}

This subsection's main result reads as follows:

\begin{proposition}\label{P-random-rich}
  A random set $R\subseteq\bN$ is, with probability 1, rich.
\end{proposition}

\begin{proof}
  Let $s,u,v\in\bN$, $\ov{a_0}\in\set{0,1}^s\setminus\set{\ov{0}}$,
  $\ov{a_1},\dots,\ov{a_u}\in\bN^s$, and $c_1,\dots,c_u\in\bN$ with
  $(\ov{a_0},0)\neq(\ov{a_i},c_i)$ for all $i\in\set{1,\ldots,u}$. For notational
  simplicity, we set $c_0\deff 0$. Without loss of generality, we can
  assume that
  $\ov{a_0}=(\underbrace{1,1,\dots,1}_{j\text{
      entries}},0,0,\dots,0)$. We show that, for a random set
  $R\subseteq\bN$, there exist with probability~1 vectors
  $\ov{x},\ov{y}\in(v+\bN)^s$ satisfying the conditions \eqref{eq:x}--\eqref{eq:y2}
  of Lemma~\ref{L-semirich}.

  First let $B>j$ be properly larger than any of the entries in
  $\ov{a_i}$ and any of the number $c_i$ for $i\in\set{0,\ldots,u}$.  Let
  furthermore $d=4s\sum_{1\le i\le j+1}B^i$. From
  Lemma~\ref{L-large-gaps} below, we obtain a sequence of natural numbers
  $(q_n)_{n\in\bN}$ with $q_n<q_{n+1}$ for all $n\in\NN$ and a
  sequence of vectors $(\ov{x_n})_{n\in\bN}$ 
  such that the following holds for all $n\in\bN$, all
  $\ov{a}\in\set{0,\ldots,B{-}1}^s$ and all $c\in\set{0,\ldots,B{-}1}$:
  \begin{enumerate}[(a)]
  \item If $\ov{a}\neq\ov{0}$, then
    $\frac{q_n}{d}<\ov{a}\transpose\,\ov{x_n}-c$,
  \item $\ov{a}\transpose\,\ov{x_n}-c<q_n d$, \quad and
  \item $\ov{a}=\ov{a_0}$ and $c=0$ \ if, and only if, \ $\ov{a}\transpose\,\ov{x_n}-c=q_n$.
  \end{enumerate}

  Since the sequence $(q_n)_{n\in \NN}$ is infinite, we can in addition assume
  $v\le \frac{q_0}{d}$ and $q_n d<\frac{q_{n+1}}{d}$ for all $n\ge0$.

  Setting $\ov{a}=(0,\dots,0,1,0,\dots,0)$ and $c=0$, condition (a)
  implies
  $v<\frac{q_0}{d}\le\frac{q_n}{d}<\ov{a}\transpose\,\ov{x_n}$
  which is an entry of $\ov{x_n}$. Hence $\ov{x_n}\in(v+\bN)^s$ for
  all $n\in\bN$.

  Let $n\in\bN$ be fixed.  We now estimate the probability for
  \eqref{eq:x}--\eqref{eq:y2} from Lemma~\ref{L-semirich} to hold with
  $\ov{x}=\ov{x_{2n}}$ and $\ov{y}=\ov{x_{2n+1}}$: Then
  condition~\eqref{eq:x} expresses that some fixed numbers
  $m_1,\dots,m_j$ for some $j\le u+1$ belong to $R$. Thus, this
  condition holds with probability
  $\geq \frac{1}{2^j}\ge\frac{1}{2^{u+1}}$. Similarly,
  condition~\eqref{eq:y1} holds with probability $\frac{1}{2}$ and
  condition~\eqref{eq:y2} with probability $\ge\frac{1}{2^u}$. Note
  that
  $\ov{a_i}\transpose\,\ov{x_{2n}}-c_i< q_{2n} d <
  \frac{q_{2n+1}}{d}<\ov{a_j}\transpose\,\ov{x_{2n+1}}-c_j$ for all
  $i,j$ with $\ov{a_i},\ov{a_j}\neq\ov{0}$ and, furthermore, that
  $\ov{a_0}\neq\ov{0}$. Hence condition~\eqref{eq:x} is independent
  from both, condition~\eqref{eq:y1} and
  condition~\eqref{eq:y2}. Furthermore,
  $\ov{a_0}\transpose\,\ov{x_{2n+1}}-c_0\neq\ov{a_i}\transpose\,\ov{x_{2n+1}}-c_i$
  for all $i\in\set{1,\ldots,u}$ with $\ov{a_i}\neq\ov{0}$. Hence also
  condition~\eqref{eq:y1} is independent from
  condition~\eqref{eq:y2}. It follows that the probability for all
  three conditions to hold is $\ge\frac{1}{2^{2u+2}}$. 
  Let $p\deff \frac{1}{2^{2u+2}}$ and note that $p>0$ and $p$ is
  independent from $n$.

  For each fixed $N\in\NN$, the probability that for all $n\leq N$ 
  at least one of the conditions~\eqref{eq:x}--\eqref{eq:y2} from Lemma~\ref{L-semirich}
  with $\ov{x}=\ov{x_{2n}}$ and $\ov{y}=\ov{x_{2n+1}}$ is violated, is
  $\leq (1-p)^N$. Thus, 
  the probability that for all $n\in\bN$, 
  at least one of the conditions~\eqref{eq:x}--\eqref{eq:y2} from Lemma~\ref{L-semirich}
  with $\ov{x}=\ov{x_{2n}}$ and $\ov{y}=\ov{x_{2n+1}}$ is violated, is
  $\lim_{N\to\infty}(1-p)^N=0$. Hence, with probability $1$, there is some $n\in\NN$
  satisfying all conditions \eqref{eq:x}--\eqref{eq:y2} with $\ov{x}=\ov{x_{2n}}$ and
  $\ov{y}=\ov{x_{2n+1}}$. 

  Note that there are only countably many legitimate choices of
  $s,u,v,\ov{a_0},\dots,\ov{a_u},c_1,\dots,c_u$, and recall that the
  intersection of countably many events of probability 1 has
  probability 1, again. Hence, with probability 1, the condition from
  Lemma~\ref{L-semirich} holds. Consequently, $R$ is with probability
  1 rich.\qed
\end{proof}

\begin{lemma}\label{L-large-gaps}
  Let $1\le j\le s$ and $B>j$ be natural numbers and let
  $d=4s\sum_{1\le i\le j+1} B^i$. For all sufficiently large
  natural numbers $q$,
  there exists $\ov{x}\in\bN^s$ such that the following hold for
  all $\ov{a}\in \set{0,\ldots,B{-}1}^s$ and all $c\in\set{0,\ldots,B{-}1}$:
  \begin{enumerate}[(a)]
  \item If $\ov{a}\neq\ov{0}$, then
    $\frac{q}{d}<\ov{a}\transpose\,\ov{x}-c$,
  \item $\ov{a}\transpose\,\ov{x}-c<qd$, \quad and
  \item $\ov{a}\transpose\,\ov{x}-c=q$ \ if, and only if, \
    $\ov{a}\transpose=(\underbrace{1,1,\dots,1}_{j\text{
        entries}},0,\dots,0)$ and $c=0$.
  \end{enumerate}
\end{lemma}

\begin{proof}
  Consider an arbitrary
  \[
    q \quad >\quad d\cdot j\cdot B^{j+2}\,.
  \]

  For $1\le i\le j$, set
  \[
    q_i' \ \ \deff \ \
    \begin{cases}
      1 & \text{if }1\le i\le q\bmod j\\
      0 & \text{otherwise.}
    \end{cases}
  \]
  and
  \[
     q_i \ \ = \ \ \left\lfloor\frac{q}{j}\right\rfloor+q_i'
  \]
  Then, $\sum_{i=1}^{j} q_i=q$ since
  $j\cdot\left\lfloor\frac{q}{j}\right\rfloor=q-(q\mod j)$.

  Let us define the values $x_i$ for $i\in\set{1,\ldots,s}$ as follows:
  \[
    x_i \ \ \deff \ \
    \begin{cases}
      \displaystyle
      q_i-\sum_{k=1}^{j} B^k + jB^i &\text{ for }1\le i\le j\\
      q+B & \text{ for }j<i\le s
    \end{cases}
  \]

  First note the following:
  \begin{align*}
    \sum_{i=1}^{j}x_i 
      &\quad = \quad q - j\cdot\sum_{k=1}^{j} B^k + \sum_{i=1}^{j} jB^i\\
      &\quad = \quad q
  \end{align*}
  which verifies the implication ``$\Leftarrow$'' in condition~(c).

  To verify condition~(a), consider arbitrary
  $\ov{a}=(a_1,\dots,a_s)\in \set{0,\ldots,B{-}1}^s\setminus\set{\ov{0}}$ and
  $c\in\set{0,\ldots,B{-}1}$. Then there exists $\ell\in\set{1,\ldots,s}$ with
  $a_\ell>0$. First consider the case where $1\le \ell\le j$.

  Note that
  \[
    d\cdot j\cdot\sum_{k=1}^{j}B^k \quad < \quad d\cdot j\cdot B^{j+2}
    \quad <\quad q \quad \le \quad q(d-j)
  \]
  and therefore
  \[
    qj \quad < \quad qd \ - \ d\cdot j\cdot\sum_{k=1}^{j}B^k\,,
  \]
  implying that 
  \begin{align*}
    \frac{q}{d} 
      &\quad <\quad  \frac{q}{j}-\sum_{k=1}^{j}B^k\\
      &\quad <\quad \left\lfloor\frac{q}{j}\right\rfloor + 1 -\sum_{k=1}^{j}B^k\\
      &\quad \le\quad  q_\ell + 1 -\sum_{k=1}^{j}B^k
          &&\text{since }\ell\le j\text{ and therefore }
              \left\lfloor\frac{q}{j}\right\rfloor \le q_\ell\\
      &\quad \le\quad  q_\ell -\sum_{k=1}^{j}B^k +jB^\ell-B +1
          &&\text{since }j,\ell\ge 1\\
      &\quad =\quad  x_\ell - (B-1)\\
      &\quad \le\quad  \ov{a}\transpose\,\ov{x} - c
          &&\text{since } c<B, a_\ell>0\text{, and }a_i\ge 0\text{ for all }i\in\set{1,\ldots,s}
  \end{align*}

  Next, we consider the case that $j<\ell\le s$. In this case,
  $\frac{q}{d}< q = x_\ell-B < \ov{a}\transpose\,\ov{x}-c$. This
  completes the verification of condition~(a).

  To verify condition~(b), consider arbitrary $\ov{a}\in\set{0,\ldots,B{-}1}^s$ and $c\in\set{0,\ldots,B{-}1}$.
  Then we get%
  \begin{align*}
    \ov{a}\transpose\,\ov{x}-c
      & \quad <\quad\sum_{i=1}^{s} Bx_i
         &&\text{since $c\ge0$ and $a_i<B$ for all $i\in\set{1,\ldots,s}$}\\
      &\quad =\quad  B\left(\sum_{i=1}^{j}x_i+\sum_{i=j+1}^{s}x_i\right)\\
      &\quad =\quad  B(q+(s{-}j)(q{+}B))
         &&\text{since $\sum_{i=1}^{j}x_i=q$ and $x_i=q{+}B$ for all $i>j$}\\
      &\quad <\quad  B\cdot (sq+sB) = sB(q+B)
         &&\text{since $j>0$}\\
      &\quad \le\quad  sqB^2
         &&\text{since }q,B\ge2\\
      &\quad <\quad qd 
         && \text{since }d>sB^2\,. 
  \end{align*}
  Therefore, condition~(b) is satisfied.

  It remains to verify the implication ``$\Rightarrow$'' of
  condition~(c). So let $\ov{a}=(a_1,\dots,a_s)\in\set{0,\ldots,B{-}1}^s$ and
  $c\in\set{0,\ldots,B{-}1}$ with $\ov{a}\transpose\,\ov{x}-c=q$. If there
  exists $i\in\set{j{+}1,\ldots,s}$ with $a_i>0$, then
  $q=\ov{a}\transpose\,\ov{x}-c\ge x_i-c=q+B-c>q$, a
  contradiction. Hence $a_{j+1}=\cdots=a_s=0$.

  We now distinguish the cases $\sum_{i=1}^{j}a_i=j$ and
  $\sum_{i=1}^{j}a_i\neq j$. In the former, we have
  \begin{align*}
      q &\quad =\quad  \sum_{i=1}^{j}a_ix_i-c\\
        &\quad =\quad  \sum_{i=1}^{j} a_i q_i
             - \sum_{i=1}^{j} a_i\Big(\sum_{k=1}^{j} B^k\Big) 
             + \sum_{i=1}^{j} (a_i j B^i) 
             - c\\
        &\quad \le\quad  \sum_{i=1}^{j} a_i \left(\frac{q}{j}+1\right)
             - \sum_{i=1}^{j} a_i\Big(\sum_{k=1}^{j} B^k \Big)
             + j \sum_{i=1}^{j} (a_i B^i) 
             - c\\
        &\quad =\quad  (q+j)-j\sum_{k=1}^{j} B^k + j\sum_{i=1}^{j} (a_iB^i) -c\,,
          &&\text{since }\sum_{i=1}^{j}a_i=j\,.
  \end{align*}
  Therefore,
  \[ 
       j\sum_{k=1}^{j} B^k \qquad \le \qquad j\sum_{i=1}^{j}(a_iB^i)+j-c\,.
  \]
  Since $|j-c|\le\max(j,c)<B$, this implies that
  \[
     -1 
     \quad<\quad
     \frac{c-j}{B}
     \quad \le \quad
     j\sum_{i=1}^{j}(a_i{-}1)B^{i-1} 
     \quad \in \quad \bZ
  \]
  Therefore,
  \[
     0
     \quad\le \quad
     \sum_{i=1}^{j}(a_i{-}1)B^{i-1}
     \quad\le \quad
     \left(\sum_{i=1}^{j}(a_i{-}1)\right)\cdot\left(\sum_{i=1}^{j}B^{i-1}\right)
     \quad = \quad
     0\,,
  \]
  since $\sum_{i=1}^{j}(a_i{-}1)=\sum_{i=1}^{j}a_i \ -  j =0$.
  Thus, we have $\sum_{i=1}^{j}(a_i{-}1)B^{i-1}=0$, and hence
  \[ 
       \sum_{i=1}^{j} B^{i-1} \qquad = \qquad \sum_{i=1}^{j}(a_iB^{i-1})\,.
  \]
  Since $0\leq a_i <B$ for all $i\in\set{1,\ldots,j}$, this implies that
  $a_1=a_2=\dots=a_j=1$. I.e., we obtain the implication ``$\Rightarrow$'' of condition (c).

  It remains to consider the case where $\sum_{i=1}^{j}a_i\neq j$.  As
  above, we obtain
  \begin{align*}
    q&\quad\le \quad    
       \sum_{i=1}^{j} a_i \left(\frac{q}{j}+1\right)
       - \sum_{i=1}^{j} a_i\Big(\sum_{k=1}^{j} B^k \Big)
       + j \sum_{i=1}^{j} (a_i B^i) -c\\
     &\quad= \quad    
       \sum_{i=1}^{j} a_i\cdot \frac{q}{j}
       \ + \ \sum_{i=1}^{j} a_i\left(1-\sum_{k=1}^{j} B^k\right)
       \ + \ j \sum_{i=1}^{j} (a_i B^i) \ - \ c\,.
  \end{align*}
  This implies that
  \[
    q\left(\frac{j-\sum_{i=1}^{j} a_i}{j}\right)
     \qquad \le\qquad \sum_{i=1}^{j}a_i\left(1-\sum_{k=1}^{j}B^k\right)
         \ + \ j\sum_{i=1}^{j}a_iB^i \ - \ c\,.
  \]
  Hence
  \begin{align*}
    q &\quad \leq\quad  \frac{j\left( \sum_{i=1}^{j} a_i\left( 1-\sum_{k=1}^{j}B^k
                                             \right)
                       +j\sum_{i=1}^{j} a_iB^i 
                       -c
                \right)}
              {j-\sum_{i=1}^{j}a_i} \\[1ex]
      &\quad \le\quad  \frac{\left|
                   j\left( \sum_{i=1}^{j} a_i\left( 1-\sum_{k=1}^{j}B^k
                                                  \right)
                           +j\sum_{i=1}^{j} a_iB^i 
                           -c
                    \right)
                \right|}
              {\left|j-\sum_{i=1}^{j}a_i\right|} \\[1ex]
      &\quad \le\quad  \left|
             j\left( \sum_{i=1}^{j} a_i\left( 1-\sum_{k=1}^{j}B^k
                                            \right)
                     +j\sum_{i=1}^{j} a_iB^i 
                     -c
              \right)
           \right|
        &&\text{ since }\textstyle|j-\sum_{i=1}^{j}a_j|>0 \\[1ex]
      &\quad \le\quad  j\left( \sum_{i=1}^{j} a_i\left( 1+\sum_{k=1}^{j}B^k
                                             \right)
                       +j\sum_{i=1}^{j}a_iB^i 
                       +c
            \right)\\[1ex]
      &\quad <\quad  j\left( jB + jB\cdot B^{j+1} + j B^{j+2} + B
          \right) &&\text{ since }a_i<B \text{ and }c<B\\[1ex]
      &\quad <\quad  4j^2 B^{j+2}\\
      & \quad < \quad d\cdot j\cdot B^{j+2}\,,
  \end{align*}
  contradicting our choice of $q$.
  Therefore, the case where $\sum_{i=1}^{j}a_i\neq j$ cannot occur.
  This completes the proof of Lemma~\ref{L-large-gaps}.
  \qed
\end{proof}

\subsection{Sets with large gaps are rich}\label{SS-large-gaps}

While Proposition~\ref{P-random-rich} shows that almost all sets of
natural numbers are rich, it does not provide us with a single such
set, let alone a natural one. Next we prove that sets with ``large gaps''
(defined as follows) are rich:

\begin{definition}\label{def:LargeGaps}
  A set $R\subseteq\bN$ \emph{has large gaps} if $R$ is infinite, $0\notin R$, and for
  every $k\in\bN$ with $k>0$ there exists $q\in R$ such that
  \begin{equation}\label{eq:def:LargeGaps}
    \left[\,\left\lfloor\frac{q}{k}\right\rfloor\,,\,k\cdot
      q\,\right]\ \cap\ R \ \ = \ \ \set{q}\,.
  \end{equation}
\end{definition}

Examples of sets with large gaps are $\setc{n^n}{n\in\bN}$,
$\setc{\lfloor 2^{n^c}\rfloor}{n\in\bN}$ for all reals $c>1$,
$\setc{n!}{n\in\bN}$ as well as all infinite subsets of these sets.
But note that neither the set $\setc{2^n}{n\in\bN}$ nor (by Bertrand's
postulate) the set of all primes has large gaps.

\begin{proposition}\label{P-large-gaps-rich}
  If $R\subseteq\bN$ has large gaps, then $R$ is rich.
\end{proposition}

\begin{proof}
  Let $s,u,v\in\bN$, $\ov{a_0}\in\set{0,1}^s\setminus\set{\ov{0}}$,
  $\ov{a_1},\dots,\ov{a_u}\in\bN^s$, and $c_1,\dots,c_u\in\bN$ with
  $(\ov{a_0},0)\neq(\ov{a_i},c_i)$ for all $i\in\set{1,\ldots,u}$. For
  notational simplicity, we set $c_0\deff 0$. We can, without loss of
  generality, assume that
  \[
     \ov{a_0} \quad = \quad \big(\,\underbrace{1,\dots,1}_{j\text{ entries}},0,\dots,0\,\big)\,.
  \]

  Let $B>j$ be larger than any of the entries in $\ov{a_i}$ and
  larger than any of the numbers $c_i$, for all $i\in\set{0,\ldots,u}$. Let
  furthermore $d\deff 4s\sum_{1\le i\le j+1}B^i$.  
 \begin{claim}\label{claim:largegaps-newclaim}
  There are infinitely many values $q\in\bN$ such that
  \begin{equation}
    \label{eq:gross}
      \left[\;\left\lfloor\frac{q}{d(d{+}1)}\right\rfloor\;,\;d(d{+}1)\cdot
        q\;\right]
      \ \cap\ R
            \quad = \quad\set{q}\,.
  \end{equation}
 \end{claim}
 \begin{proof}
   For each $q\in R$ let $k_q\in\NN$ be the maximal value $k$ such that
   \eqref{eq:def:LargeGaps} holds (note that Definition~\ref{def:LargeGaps} implies that $k_q$ does exist 
   for every $q\in R$).
   Since $R$ has
   large gaps, the set $\setc{k_q}{q\in R}$ contains arbitrarily large
   elements. Thus, the set $\setc{k_q}{q\in R, \ k_q\geq d(d{+}1)}$ is infinite.
   Since every $q\in R$ with $k_q\geq d(d{+}1)$ 
   satisfies \eqref{eq:gross}, the claim follows.
 \end{proof}

 \bigskip

  By Lemma~\ref{L-large-gaps} and
  Claim~\ref{claim:largegaps-newclaim}, there exists a $q>vd(d+1)$
  satisfying~\eqref{eq:gross} and there exist vectors $\ov{x},\ov{y}\in(v+\bN)^s$ such that
  for all $i\in\set{0,\ldots,u}$, the following holds with
  $q'=\left\lfloor\frac{q+d+1}{d+1}\right\rfloor$:
  \begin{enumerate}[(a)]
  \item If $\ov{a_i}\neq\ov{0}$, then
    $\frac{q}{d}<\ov{a_i}\transpose\,\ov{x}-c_i$ and
    $\frac{q'}{d}<\ov{a_i}\transpose\,\ov{y}-c_i$.
  \item $\ov{a_i}\transpose\,\ov{x}-c_i<qd$ and $\ov{a_i}\transpose\,\ov{y}-c_i<q'd$.
  \item $\ov{a_i}\transpose\,\ov{x}-c_i=q$ \ if, and only if, \ $i=0$.
  \end{enumerate}

  In particular, \eqref{eq:gross} implies that
  \[
     \left[\,\left\lfloor\frac{q}{d}\right\rfloor\,,\,dq\,\right]\
     \cap \ R \quad =\quad\set{q}\,.
  \]

  Consequently, (a)--(c) imply for all $i\in\set{0,\ldots,u}$ with
  $\ov{a_i}\neq\ov{0}$ that
  \begin{equation}
    \label{eq:gaps-x}
    \ov{a_i}\transpose\,\ov{x}-c_i\in R \quad \iff \quad i=0\,.
  \end{equation}
  The same holds if $\ov{a_i}=\ov{0}$, since then
  \ $\ov{a_i}\transpose\,\ov{x}-c_i\ = \ -c_i\ \le\ 0$ implies that
  \ $\ov{a_i}\transpose\,\ov{x}-c_i \notin R$ \ and \ $i\neq 0$.

  We now consider $q'=\lfloor\frac{q+d+1}{d+1}\rfloor$ and $\ov{y}$.
  Clearly,
  \begin{align*}
    \frac{q'}{d}
      &\quad =\quad \frac{1}{d}\cdot\left\lfloor\frac{q+d+1}{d+1}\right\rfloor\\
      &\quad \ge \quad\frac{1}{d}\cdot\left(\frac{q+d+1}{d+1}-1\right)\\
      &\quad =\quad \frac{q}{d(d+1)}
  \end{align*}
  and
  \begin{align*}
    q'd &\quad =\quad d\cdot\left\lfloor\frac{q+d+1}{d+1}\right\rfloor\\
     &\quad \le\quad d\cdot\frac{q+d+1}{d+1}\quad = \quad \frac{d}{d+1}q+d\\
     &\quad <\quad \frac{d}{d+1}q+ \frac{1}{d+1}q
       &&\text{since }vd(d+1)<q\\
     &\quad =\quad q\,.
  \end{align*}
  Consequently, we have
  \[
     \left[\,\left\lfloor\frac{q'}{d}\right\rfloor\,,\,q'd\,\right]\cap R
     \quad \subseteq \quad
     \left[\,\left\lfloor\frac{q}{d(d{+}1)}\right\rfloor\,,\,q{-}1\,\right]\cap R
     \quad \stackrel{\eqref{eq:gross}}{=} \quad \emptyset\,.
  \]

  Hence, (a) and (b) imply for all $i\in\set{0,\ldots,u}$ with
  $\ov{a_i}\neq\ov{0}$ that
  \begin{equation}
    \label{eq:y}
    \ov{a_i}\transpose\,\ov{y}-c_i \ \notin\ R
  \end{equation}
  Since $\ov{a_0}\neq\ov{0}$, \eqref{eq:gaps-x} and \eqref{eq:y} imply that
  \[
     \ov{a_0}\transpose\,\ov{x}\ \in \ R \qquad \iff\qquad
     \ov{a_0}\transpose\,\ov{y}\ \notin\ R\,.
  \]
  Similarly, for all $i\in\set{1,\ldots,u}$ with $\ov{a_i}\neq\ov{0}$,
  \eqref{eq:gaps-x} and \eqref{eq:y} implies that
  \[
     \ov{a_i}\transpose\,\ov{x}-c_i \ \in \ R \qquad \iff \qquad
     \ov{a_i}\transpose\,\ov{y}-c_i\ \in\ R\,.
  \]
  Finally, for all $i\in\set{1,\ldots,u}$ with $\ov{a_i}=\ov{0}$, we get
  \[
     \ov{a_i}\transpose\,\ov{x}-c_i \ \ = \ \ -c_i \ \ = \ \ 
     \ov{a_i}\transpose\,\ov{y}-c_i
  \]
  and therefore
  \[
     \ov{a_i}\transpose\,\ov{x}-c_i\ \in\ R\qquad\iff\qquad
     \ov{a_i}\transpose\,\ov{y}-c_i\ \in\ R\,.
  \]
  Hence $R$ is rich.
 This completes the proof of Proposition~\ref{P-large-gaps-rich}.
 \qed
\end{proof}

\end{document}